\documentclass[aps,pra,twocolumn,nofootinbib,floatfix,longbibliography,superscriptaddress]{revtex4-2}

\usepackage{bm}

\usepackage[utf8]{inputenc}
\usepackage[english]{babel}
\babeladjust{autoload.bcp47 = on}
\usepackage[T1]{fontenc}
\usepackage{listings}
\usepackage{hyperref}
\usepackage{thm-restate}
\usepackage{thmtools}

\usepackage{amssymb}
\usepackage{amsthm}
\usepackage{bbm}
\usepackage{qcircuit}
\usepackage{adjustbox}

\expandafter\let\csname equation*\endcsname\relax

\expandafter\let\csname endequation*\endcsname\relax

\newtheorem{defi}{Definition}
\newtheorem{thm}{Theorem}
\newtheorem{lem}{Lemma}

\newtheorem{obs}{Observation}
\newtheorem{cor}{Corollary}
\newtheorem{conj}{Conjecture}

\usepackage{natbib}
\usepackage{physics}
\usepackage{xcolor}
\usepackage{graphicx}
\usepackage{wrapfig}
\usepackage{qcircuit}

\usepackage{tikz}
\usepackage{tikz-cd}
\usetikzlibrary{positioning}
\tikzset{
  symbol/.style={
    draw=none,
    every to/.append style={
      edge node={node [sloped, allow upside down, auto=false]{$#1$}}}
  }
}

\def\fC{\mathbb{C}}

\def\ra{\rangle}
\def\la{\langle}
\def\id{\mathbb{I}}
\def\Tr{\text{Tr}}
\newcommand{\bigzero}{\mbox{\normalfont\Large\bfseries 0}}
\newcommand{\rvline}{\hspace*{-\arraycolsep}\vline\hspace*{-\arraycolsep}}

\newcommand{\rd}[1]{{#1}}
\newcommand{\bl}[1]{{#1}}

\theoremstyle{plain}

\begin{document}

\preprint{}

\title[Local Thermal Operations and Classical Communication]{Local Thermal Operations and Classical Communication}

\author{Rafał Bistroń}
\email{rafal.bistron@doctoral.uj.edu.pl}
\affiliation{Doctoral School of Exact and Natural Sciences, Jagiellonian University, ul. Łojasiewicza 11, 30-348 Kraków, Poland}
\affiliation{Faculty of Physics, Astronomy and Applied Computer Science, Jagiellonian University, ul. Łojasiewicza 11, 30-348 Kraków, Poland}
\author{Jakub Czartowski}
\affiliation{Faculty of Physics, Astronomy and Applied Computer Science, Jagiellonian University, ul. Łojasiewicza 11, 30-348 Kraków, Poland}
\affiliation{School of Physical and Mathematical Sciences, Nanyang Technological University, 21 Nanyang Link, 637361 Singapore, Republic of Singapore}
	
\date{\today}

\begin{abstract}
    In quantum thermodynamics, understanding the interplay between locality, thermal constraints, and communication remains an open challenge. In this manuscript, we introduce Local Thermal Operations and Classical Communication (LTOCC), a novel operational framework that unifies the distant laboratories paradigm with thermodynamic restrictions, defining the fundamental limits on transformations between spatially separated systems. We establish a hierarchy of LTOCC protocols, demonstrating inclusion relations between different levels and revealing their deep connection to semilocal thermal operations. To formalize this framework, we develop thermal tensors and bithermal tensors, extending stochastic and tristochastic tensors to thermodynamic settings and providing new mathematical tools for constrained quantum processes. Finally, we present limitations imposed by LTOCC on single- and multi-copy CHSH scenario, demonstrating no violation in former and a gap between thermal and athermal local operations in the latter with respect to their capability to detect entanglement.
\end{abstract}

\maketitle


\section{Introduction}

		
	Amongst the pillars on which modern physics stands, quantum mechanics and thermodynamics present themselves as two of the most robust ones. As Sir Arthur Eddington has famously stated, \textit{``(...) if your theory is found to be against the second law of thermodynamics, I can give you no hope; there is nothing for it but to collapse in deepest humiliation~\cite{eddington1929nature}.''} Similarly, despite significant controversies regarding its ontological status, quantum mechanics has withstood all experimental inquiry it has been subject to over the century of its progress.
	
    Recent years have seen a growth of the efforts focused on joining these two fields in what is called quantum thermodynamics~\cite{Big_book1, Big_book2, vinjanampathy2016quantum}, with results applicable to the problems of quantum battery charging~\cite{Alicki2013entanglement, Hovhannisyan2013entanglement, Binder2015, campaioli2017enhancing, Mitchison2021chargingquantum, QUACH2023quantum, campaioli2024colloquium, Morrone2023, razzoli2024cyclic}, cooling~\cite{Boykin2002algorithmic, RodriguezBriones2017, Baugh2005, Brassard2014, Elias2011semioptimal, zeng2021universal, laflamme2022algorithmic}, autonomous quantum machines~\cite{linden2010howSmall,  brunner2012virtual, lipka2024thermodynamic, guzman2023useful} and more \cite{campbell2025roadmap}. Most important in the current context, however, are efforts to provide ultimate bounds on transformations of low-dimensional systems in contact with a heat bath that would respect the conservation of thermal equilibrium. This approach has resulted in the introduction of the resource-theoretic framework of thermal operations, with a wide body of related works already present within the literature~\cite{Janzing2000, horodecki2013, horodecki2013fundamental, brandao2015secondlaws, Lostaglio_2018, deoliveira2022geometric, czartowski2023thermalrecall, czartowski2024catalyse}, extended by a concept of semilocal thermal operations, where local thermal baths in different temperatures interact only via local proxy systems, but interactions between the said systems can be arbitrarily nonlocal, classical or quantum, within the thermodynamical limits \cite{Bera2021, Bera2022}. 
	
	In parallel, the discovery of quantum entanglement and the related notion of quantum
    nonlocality lead to the realization of the fundamental difference between classical and quantum distributed systems, with diverse applications, such as quantum teleportation~\cite{PhysRevLett.70.1895}, quantum key distribution~\cite{PhysRevLett.67.661}, quantum secret sharing~\cite{PhysRevA.59.1829} and quantum network nonlocality~\cite{Cavalcanti2011}, to name but a few. Early on it has also been realised that local operations and classical communication (LOCC) constitute a class of transformations that, when applied on a nonlocal state, do not create entanglement. 

    In this work, following Sir Eddington's advice, we propose a framework for operations limited both by locality and communication restricted to classical means, as well as incorporating thermodynamic laws locally for all involved parties, thus adhering to the basic requirements of thermodynamics. For simplicity, we will refer to such operations as \textit{local thermal operations and classical communication} (LTOCC). In Section~\ref{sec:prelim} we recall all the notions necessary to introduce our framework, including: noisy and thermal operations and cooling maps, local operations and classical communication, and concepts connected to stochastic tensors and multistochasticity. Furthermore, in Section~\ref{sec:SLTO} we discuss the concept of semilocal thermal operations (SLTO) and present new results, which are required to establish a connection between SLTO and LTOCC. In Section~\ref{sec:framework} we introduce the proper framework of LTOCC together with a hierarchy of protocols requiring different amounts of communication and control, demonstrating their relationship with semilocal thermal operations. In Section~\ref{sec:BTTensors} we study the basic building blocks of the framework -- thermal tensors -- with a special focus on the nontrivial, symmetric subset of bithermal tensors, which are thermal counterparts of tristochastic tensors.
    Section \ref{sec:nonloc_corr} is devoted to study of nonlocal correlations under LTOCC. More specifically, in Section~\ref{sec:sltocc_sepstates} we discuss correlation generation under LTOCC from energy-incoherent product states, highlighting significant correlating power of LTOCC with memory. Finally, in Section \ref{sec:bell_LTOCC} we consider CHSH scenario under LTOCC and demonstrate gap between full and thermally restricted LOCC, thus showing possiblity of detection of athermality usage in the context of Bell nonlocality.
    Section \ref{Sec:out} serves as a brief summary of the proposed framework and sets up a list of open problems for future investigation. 
    
    Appendix~\ref{App:cold_original} and Appendix~\ref{App:cold} present a discussion of cooling maps and their alternative derivation, Appendix~\ref{App:Geo_evol} provides a detailed discussion of one- and two-round LTOCC protocols with and without memory, and Appendix~\ref{App:semilocal} presents new, auxiliary results on semilocal thermal operations and proof of inclusion relation between LTOCC and SLTO. Appendix~\ref{App:coch_evol} presents additional auxiliary results on coherence evolution under SLTO and LTOCC protocols. Finally Appendix~\ref{app:subsec_logic_gates} discusses realisability of exemplary classical gates, CNOT and SWAP, within the framework of LTOCC.

 \begin{figure}[h]
     \centering
     \includegraphics[width=\linewidth]{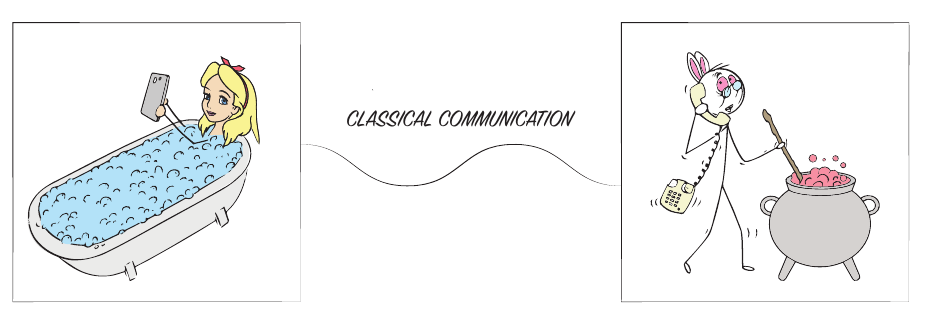}
     \caption{Artist's depiction of Alice and Bob using classical communication, in contact with local thermal baths, thus realising the backbone of LTOCC, Local Thermal Operations and Classical Communication \textit{(Courtesy of A. de Oliveira Junior)}.}
     \label{fig:Famework_sheme}
 \end{figure}

    \subsection*{Summary of original results}

        Due to the extent of this manuscript, we provide a concise summary of the novel results provided:
        \begin{enumerate}
            \item Proofs of convexity and closedness under composition for semilocal thermal operations (SLTO) (Theorems \ref{SLTO_closed} and \ref{SLTO_convex}) and additional auxiliary results on coherence evolution under SLTO (Appendix~\ref{App:coch_evol}),
            \item Introduction of a new paradigm of local thermal operations with classical communication (LTOCC) and corresponding communication complexity hierarchy, together with proofs of inclusion between different strata of hierarchy (Section \ref{sec:hierarchy}),
            \item Proof of inclusion of LTOCC with shared randomness in the SLTO in a general quantum scenario for any number of rounds (Theorem \ref{thm:subset_of_SLTO}) and conjecture concerning equality of the sets in the infinite-round limit for energy-incoherent states (Conjecture \ref{conj:subset_of_SLTO}),
            \item Detailed study of basic building blocks of the framework -- thermal and bithermal tensors -- with particular attention to the latter as analogues of tristochastic tensors, with preliminary results on the structure of thermal analogue of tristochastic Birkhoff polytope (Section \ref{sec:BTTensors}),
            
            \item Degree of correlations that can be generated under LTOCC with and without memory  (Section \ref{sec:sltocc_sepstates})
            \item Limitations on Bell nonlocality with LTOCC-restricted operations, including no CHSH inequality breaking in single copy scenario and appearance of gap between LTOCC and Tsirelson bound in multicopy setting (Section \ref{sec:bell_LTOCC})
            
        \end{enumerate}
        Additionally, in Appendix~\ref{App:Geo_evol} we present a detailed consideration of 1- and 2-round LTOCC together with placing them in the context of sets of free states for composite resource theories as defined in \cite{son2024robust}, and in Appendix~\ref{app:subsec_logic_gates} we present the limitations of approximating classical logic gates as CNOT and SWAP by LTOCC.
	
    \section{Preliminaries}\label{sec:prelim}

    \subsection{Noisy operations, thermal operations and cooling maps}
	
	Let us begin by introducing the essential frameworks that describe the evolution of finite-dimensional quantum states under thermodynamical constraints resulting from contact with a thermal bath, characterised by inverse temperature $\beta = (k_B T)^{-1}$ which will also govern local evolution of our subsystem.
    The primary system of dimension $d$, prepared in an initial state $\rho$ is characterised by its Hamiltonian $H = \sum_{i} E_i\op{E_i}$, where we assume the energy levels $E_i$ to be non-degenerate. 
	The total system-bath initial state is assumed to be uncorrelated, $\rho\otimes\gamma_E$, with the bath characterised by the Hamiltonian $H_E$ starting in the Gibbs state defined as
	\begin{equation}\label{eq:gibbs_state}
		\gamma_E = \frac{e^{-\beta H_E}}{\Tr(e^{-\beta H_E})}.
	\end{equation}
    Hereafter, subscript $E$ denotes the thermal environment. 
	If not stated otherwise, in what follows we will be restricting our considerations of states under thermodynamic evolution to the energy-incoherent states, ie. such that they commute with the Hamiltonian of the primary system, $\comm{\rho}{H} = 0$. With this restriction, there is a one-to-one correspondence between states $\rho = \sum_i p_i\op{E_i}$ and their population vectors $\vb{p} = \text{diag}_H(\rho)$. Furthermore, we define the thermal equilibrium state of the system ${\gamma}$ by Eq.~\eqref{eq:gibbs_state} with $H_E$ replaced with $H$.
	
	Thermal operations framework (TOs)~\cite{Janzing2000, Brando2015, Lostaglio2019}, with its infinite temperature version, $\beta = 0$, referred to as noisy operations (NOs), works with a minimal assumption that system and bath undergo joint unitary evolution $U$ after which the bath is discarded, thus defining the channel
	\begin{equation}
		\mathcal{E}(\rho) = \Tr_E\qty[U\qty(\rho\otimes\gamma_E)U^\dagger].
	\end{equation}
	Moreover, to ensure total energy conservation,
 it is assumed that the unitary operation $U$ commutes with the joint Hamiltonian,
	\begin{equation}
		\label{E_prev}
		\comm{U}{H\otimes\mathbbm{1}_E+\mathbbm{1}\otimes H_E} = 0,
	\end{equation}
	so it preserves the energy of the joint system.
	
	In order to provide proper background to the main topic of this manuscript we need to consider the conditions for the convertibility of states. In the case of finite temperatures, given a pair of states $\rho$ and $\sigma$, the transition is possible if and only if there exists a thermal operation $\mathcal{E}$ such that $\mathcal{E}(\rho) = \sigma$ and $\mathcal{E}(\gamma) = \gamma$.
	However, as we limit ourselves to energy-incoherent states, $\comm{\rho}{H} = \comm{\sigma}{H} = 0$, the condition reduces to existence of a stochastic matrix $\Lambda$ acting on the respective populations $\vb{p} = \operatorname{diag}_H(\rho),\,\vb{q} = \operatorname{diag}_H(\sigma)$ such that~\cite{horodecki2013}
	\begin{align}
		\Lambda\vb{p} = \vb{q},&&
		\Lambda\boldsymbol{\gamma} = \boldsymbol{\gamma}.
	\end{align}
    In the following paragraphs, we discuss the resulting restrictions on possible transitions between energy-incoherent states, $\vb{p}\rightarrow\vb{q}$. Note that the following considerations apply as necessary conditions also whenever the original state $\rho$ is coherent and $\operatorname{diag}(\rho) = \vb{p}$~\cite{horodecki2013fundamental}.
	
	In the particular case of infinite temperature (noisy operations), $\beta = 0$, the Gibbs state reduces to the maximally mixed state, with constant populations $\eta_i = \frac{1}{d}$, and the Gibbs-preserving condition $\Lambda \boldsymbol{\eta} = \boldsymbol{\eta}$ defines the set of noisy operations acting on the energy-incoherent states equivalent to the set of bistochastic matrices. 
    The state convertibility is then governed by \textit{majorization relation}.
	
	\begin{defi}[Majorisation]\label{def_Majorisation} 
		Given two $d$-dimensional probability distributions $\vb{p}$ and $\vb{q}$, we say that $\vb{p}$ \emph{majorises} $\vb{q}$, and denote it by $\vb{p} \succ \vb{q}$, if and only if the following condition holds:
		\begin{equation}
			\label{eq_majorisation}
			\sum_{i=1}^k p_i^{\downarrow}\geq\sum_{i=1}^k q_i^{\downarrow} \quad \text{for all} \quad  k\in\{1\dots d\},
		\end{equation}
		where $\vb{p}^{\downarrow}$ denotes the vector $\vb{p}$ rearranged in a non-increasing order. 
	\end{defi}
	
	\begin{thm}[Theorem~II.1.10 of Ref.~\cite{bhatia1996matrix}]\label{thm_HLP}
		There exists a bistochastic matrix $\Lambda$, $\Lambda \boldsymbol{\eta}=\boldsymbol{\eta}$, mapping $\vb{p}$ to $\vb{q}$ if and only if $\vb{p} \succ \vb{q}$.
	\end{thm}
	
	
	The finite temperature case of $\beta > 0$ has been considered in detail in~\cite{horodecki2013fundamental}, and here we present only the resulting relation of thermomajorisation. 
    together with the corresponding notion of $\beta$-ordering.
	\begin{defi}[$\beta$-ordering]
		Consider a $d$-dimensional population vector $\vb{p}$ and a corresponding Gibbs state $\boldsymbol{\gamma}$. A permutation $\pi_{\vb{p}} \in\mathcal{S}$ is called a $\beta$-order of the state $\vb{p}$ if it orders the ratios between populations of $\vb{p}$ and $\boldsymbol{\gamma}$ in a non-increasing order,
		\begin{equation}
			i < j \Longrightarrow \frac{p_{\pi_{\vb{p}}(i)}}{\gamma_{\pi_{\vb{p}}(i)}} \geq  \frac{p_{\pi_{\vb{p}}(j)}}{\gamma_{\pi_{\vb{p}}(j)}} 
            .
		\end{equation}
		Furthermore, we will call the reordered version $p^\beta_i = p_{\pi_{\vb{p}}}(i)$ as the $\beta$-ordered state.
	\end{defi}
	
	A thermomajorisation curve \mbox{$f^{\, \beta}_{\vb{p}}:\left[0,1\right]\rightarrow\left[0,1\right]$} is defined as a piecewise linear curve composed of linear segments connecting the point $(0,0)$ and the points defined by consecutive subsums of the $\beta$-ordered form of the probability $\vb{p}^\beta$ and the Gibbs state~$\vb{\gamma}^\beta$,
	\begin{equation}
    \label{cures_def}
		\left(\sum_{i=1}^k\gamma^{\, \beta}_i,~\sum_{i=1}^k p^{\, \beta}_i\right):=\left(\sum_{i=1}^k\gamma_{\vb \pi_{\vb{p}}(i)},~\sum_{i=1}^k p_{\vb \pi_{\vb{p}}(i)}\right),
	\end{equation}
	for $k\in\{1,\dots,d\}$. 
	
	Using the concept of thermomajorisation curves one can define the relation of thermomajorisation. Given two $d$-dimensional probability distributions $\vb p$ and $\vb q$, and a fixed inverse temperature $\beta$, we say that $\vb p$ \emph{thermomajorises} $\vb q$ and denote it as $\vb p \succ_{\beta} \vb q$, if the thermomajorisation curve $f^{\, \beta}_{\vb{p}}$ is above $f^{\, \beta}_{\vb{q}}$ everywhere, i.e.,
	\begin{equation}
		\vb p \succ_{\beta} \vb q \iff \forall x\in[0,1]:~ f^{\, \beta}_{\vb{p}}(x) \geq f^{\, \beta}_{\vb{q}}(x) \, .
	\end{equation}
	It may happen that a given pair of vectors, $\vb p$ and $\vb q$, are incomparable, meaning that $f^{\, \beta}_{\vb p}$ and $f^{\, \beta}_{\vb q}$ cross at some point (see Fig.~\ref{fig:Majorization_curves}). Furthermore, the thermomajorisation order cannot be seen as a lattice in the sense that for $\vb p$ and $\vb q$, which do not share the same $\beta$-order, there is no unique join or meet, corresponding to the unique first point in the joint future and last point in the joint past~\cite{Korzekwa2017}. 
    
    \begin{figure*}[t]
     \centering
     \includegraphics[width=1\linewidth]{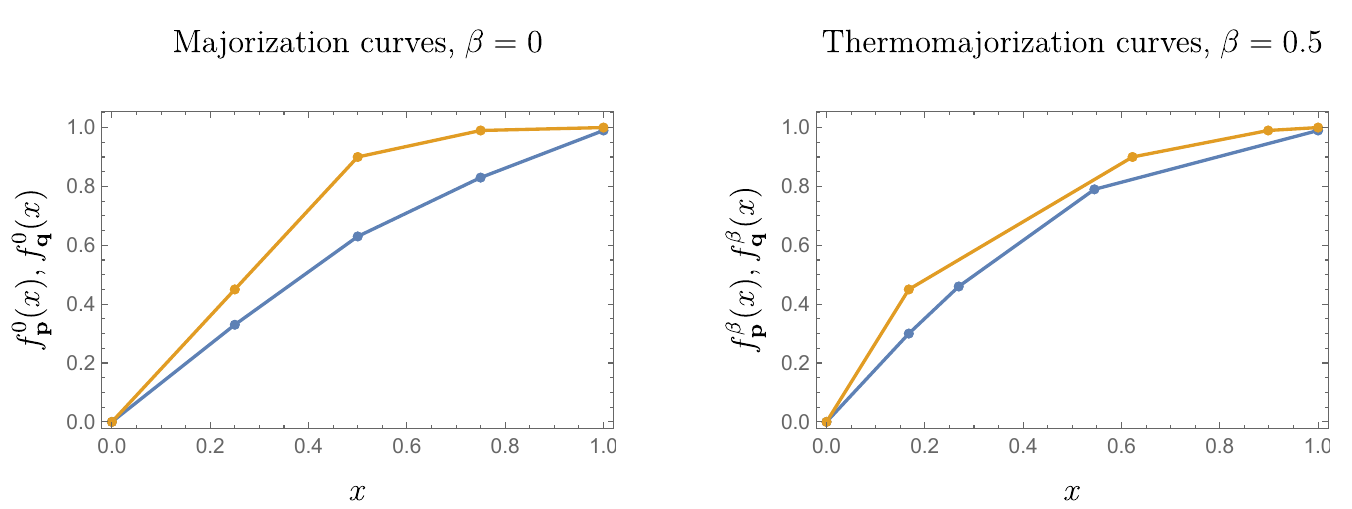}
     \caption{Example of majorization and thermomajorization curves for a pair of distributions $\vb{p} = (0.2, 0.33, 0.3, 0.16),\,\vb{q}=(0.45, 0.09, 0.45, 0.01)$ in $d = 4$ with energy levels $E_i = i$ as defined in eq~\eqref{cures_def}. In both examples, the distribution $\vb{p}$ corresponding to the blue curve majorizes the distribution $\vb{q}$ corresponding to the orange curve, $\vb{p}\succ \vb{q},\,\vb{p}\succ_\beta \vb{q}$, since the blue curve lies entirely above the orange one. Notice that in infinite temperature $\beta = 0$, the thermomajorization problem reduces to majorization one since all elbows of the curves happen for the same values of $x$, so it is sufficient to consider only these points.}
     \label{fig:Majorization_curves}
    \end{figure*}

	We can now state the generalisation of Theorem~\ref{thm_HLP} from bistochastic matrices to Gibbs-preserving stochastic matrices which provides the necessary and sufficient conditions for the existence of a thermal operation between two energy-incoherent states~\cite{horodecki2013fundamental,Rusch1978}.
	\begin{thm}[Theorem~1.3 of Ref.~\cite{horodecki2013}]
		\label{thm_HLPgeneralisation}
		There exists a stochastic Gibbs-preserving matrix $\Lambda$, $\Lambda \vb \gamma=\vb \gamma$, mapping $\vb{p}$ to $\vb{q}$ if and only if $\vb{p} \succ_{\beta} \vb{q}$.
	\end{thm}
	
	Bistochastic matrices correspond to thermal operations on energy-incoherent states in infinite temperature. However one can also study the opposite zero temperature limit.
	In Appendix~\ref{App:cold_original} we summarize the concept of \textit{cooling maps}, originating for such a study~\cite{Low_temerature_cooling_maps}, in a scenario with diagonal quantum states. Furthermore in Appendix  \ref{App:cold} we show an alternative derivation of cooling maps for diagonal quantum states as a zero temperature limit of Gibbs-preserving matrices. 
    \medskip
 
	In further discussion, following the study of thermal operations,
	we restrict ourselves to the states that are diagonal in the Hamiltonian eigenbasis, unless stated otherwise. This lets us identify quantum states with their populations $\rho = \sum_i p_i \; |E_i\ra\la E_i| \Longleftrightarrow \bf{p}$.

	\subsection{Local operations and classical communication}
		
	Whenever considering a multipartite quantum system composed of spatially separated subsystems, each with local Hilbert space $\mathcal{H}_i$, it is natural to limit the operations that can be effected on the total system. Such a scenario is often described as a ``distant laboratories'' paradigm. The most rigorous restriction comes from allowing the involved parties to perform operations only locally. Thus, for a bipartite system $\mathcal{H}_A\otimes\mathcal{H}_B$ the allowed operations would be limited to channels of the form $\mathcal{E}_A\otimes\mathcal{E}_B$. 
	
	It is, however, natural to note that distant laboratories should not be necessarily equated to disconnected laboratories, which leads to inclusion of classical communication, such as shared randomness or exchange of results of experiments. Together with local operations, possibly conditioned on the measurement results, it is referred to as \textit{local operations and classical communication} (LOCC)~\cite{Chitambar2014}. Some of the most prominent protocols in quantum information, including quantum teleportation~\cite{Bennett1993} and superdense coding~\cite{Bennett1992}, rely on LOCC operations; furthermore, the study of convertibility of entangled states under LOCC, which are known not to increase entanglement, is one of the cornerstones of the quantum resource theories~\cite{Nielsen1999, Vidal1999}.
	
	Despite their widespread use and intuitive understanding, the structure of the set of LOCC operations is intricate -- we refer the reader interested in detailed description to~\cite{Chitambar2014}, where LOCC is introduced with full rigorousness within the framework of quantum instruments. In this work, however, the most important aspect of LOCC is a possibility of performing local operations conditioned on results of measurement in a distant laboratory.

	\subsection{Multistochastic tensors and multi stochastic Birkhoff polytope}

	Another important foundation for our work are stochastic tensors. Similarly, as stochastic matrix $\Lambda_{ij}$ encompasses the mapping of probability distributions, 
	a stochastic tensor $T_{ijk}$ describes the mapping  of two distributions into one,
	\begin{equation}
		T(\vb{p},\vb{q})_i = \sum_{jk} T_{ijk} p_j q_k~,
	\end{equation}
	and the concept can be extended for more distributions by considering tensors of higher rank.
	To ensure, that the operation is well-defined, all elements of the stochastic tensor must be positive, $T_{ijk} \geq 0$, and sums of its hypercolumns are normalised, $\sum_i T_{ijk} = 1$, for any $j,k$. 
    
    Stochastic tensors arise naturally in the study of higher-order Markov chains \cite{Li2013, Hu2014, Liu2019}, in which the next probability  distribution $\vb{p}^{(n+1)}$ depends not only on the current $\vb{p}^{(n)}$, but also previous one $\vb{p}^{(n-1)}$:
    \begin{equation*}
        p_i^{(n+1)} = \sum_{j,k} T_{ijk} p_j^{(n)} p_k^{(n-1)}. 
    \end{equation*}

	Typically, physically interesting phenomena are described using restricted scenarios such as bistochasticity, which appears naturally in noisy operations or as transition probabilities arising from unitary evolution of quantum systems~\cite{KZ_2003, korzekwa2018coherifying}. 
	In the case of tensors instead of one additional condition for the normalisation of rows (hypercolumns in relation to a selected index), one can demand similar conditions to be satisfied by all hypercolumns, which gives rise to \textit{multistochastic} tensors.
	
	\begin{defi}
		A tensor $T$ is called \textit{tristochastic} if $T_{ijk} \geq 0$ and $\sum_{i} T_{ijk} = \sum_{j} T_{ijk} = \sum_{k} T_{ijk} =  1$ for  any $i,j,k$.
	\end{defi}
	
	The additional conditions correspond to the preservation of flat distribution $\boldsymbol{\eta}$, if it appears as any input:
	\begin{equation}
    \label{tristoch_def0}
		\forall\;\vb{p} ~~ T(\vb{p},\boldsymbol{\eta}) = T(\boldsymbol{\eta},\vb{p}) = \boldsymbol{\eta} 
	\end{equation}
    These extra properties were leveraged in works such as \cite{Aniello2019, bistron2023tristochastic} to generalize the notion of convolution of probability vectors
    \begin{equation*}
        \vb{p} * \vb{q} = T(p,q),
    \end{equation*}
    with $T_{ijk} = \delta_{j+k-i}$ in stardard case,
    and subsequently obtain its quantum analogues.  
    \medskip
    
	The resemblance between bistochastic matrices and multistochastic tensors suggests, that many of the properties of the former should translate to the latter. However, such a connection is not entirely straightforward and becomes even more troublesome when considering quantum counterparts for both of them~\cite{RBJCKZ24}.
 
    For example, sets of both bistochastic matrices and tristochastic tensors are convex, forming the so-called Birkhoff polytopes~\cite{stoch_multistoch}.  
	All extremal bistochastic matrices, vertices of Birkhoff polytope, are permutation matrices~\cite{Birkhoff_thm} -- all its entries are either one or zero.
	Thus one would expect, that permutation tensors, in which all entries are equal to either one or zero, also fully describe the extremal tensors. However, as proven in~\cite{Birkhoff_tensors}, there are also other extremal multistochastic tensors that do not correspond to permutation tensors. 
	Note, that while layers of permutation tensors are extremal bistochastic matrices, it no longer holds for other extremal multistochastic tensors.
	
	An example of a tristochastic permutation tensor for dimension $d = 3$ is provided below~\cite{bistron2023tristochastic}:
	\begin{equation}
		\label{example_permutation_tensor}
		\begin{aligned}
			& T = (T_{0jk},T_{1jk},T_{2jk}) = \left(\begin{matrix}
				1 & 0 & 0\\
				0 & 1 & 0\\
				0 & 0 & 1\\
			\end{matrix}\right.
			\left|\begin{matrix}
				0 & 1 & 0\\
				0 & 0 & 1\\
				1 & 0 & 0\\
			\end{matrix}\right.
			\left|\begin{matrix}
				0 & 0 & 1\\
				1 & 0 & 0\\
				0 & 1 & 0\\
			\end{matrix}\right)\,, \\
		\end{aligned}
	\end{equation}
	where the reader should imagine the square sub-matrices arranged in a $3\times 3\times 3$ cube, while an example of a non-permutation extremal tensor is presented below
	\begin{equation}
		\label{example_strange_tensor}
		\begin{aligned}
			& T = (T_{0jk},T_{1jk},T_{2jk}) = \left(\begin{matrix}
				0 & \frac{1}{2} & \frac{1}{2}\\
				\frac{1}{2} & \frac{1}{2} & 0\\
				\frac{1}{2} & 0 & \frac{1}{2}\\
			\end{matrix}\right.
			\left|\begin{matrix}
				\frac{1}{2} & \frac{1}{2} & 0\\
				0 & \frac{1}{2} & \frac{1}{2}\\
				\frac{1}{2} & 0 & \frac{1}{2}\\
			\end{matrix}\right.
			\left|\begin{matrix}
				\frac{1}{2} & 0 & \frac{1}{2}\\
				\frac{1}{2} & 0 & \frac{1}{2}\\
				0 & 1 & 0\\
			\end{matrix}\right)\,. \\
		\end{aligned}
	\end{equation}

    \section{Interlude: Auxiliary results on semilocal thermal operations}\label{sec:SLTO}

    Recently, there has been an effort to formulate a notion similar to thermal operations, which would be applicable to distributed systems interacting with different local thermal baths \cite{Bera2021}. This approach resulted in the construction of the theory named semilocal thermal operations (SLTO), which is suitable to consider the ultimate limitations of finite-size thermal engines. This framework is very general, incorporating simultaneous interactions between all system and bath elements. 

    In our work we treat SLTO as an important background, which operates within minimal assumptions similar to thermal operations involving two systems in contact with their own local thermal environments, but does not provide any additional structural restrictions on the manner in which the involved systems interact.
        The exact definition of semilocal thermal operation is as follows:

        \begin{defi}[\cite{Bera2021} Definiton 1]
            Let us consider a system $S$ consisting of two subsystems $S^{(A)}$, $S^{(B)}$ with Hamiltonian $H = H^{(A)}$ + $H^{(B)}$. An operation $\Lambda^{(AB)}$ with respect to local inverse temperatures $\beta^{(A)}$ and $\beta^{(B)}$ is called \textit{semilocal thermal operation} if there exist two local thermal baths $\mathcal{B}^{(A)}$, $\mathcal{B}^{(B)}$ with Hamiltonians $H^{\mathcal{B}^{(A)}}$, $H^{\mathcal{B}^{(B)}}$ such that
            \begin{equation}\label{SLTO_def}
            \Lambda^{(AB)}\qty(\rho^{(AB)}) = \Tr_{\mathcal{B}^{(A)},\mathcal{B}^{(B)}}\left[U(\gamma^{\mathcal{B}^{(A)}} \otimes \gamma^{\mathcal{B}^{(B)}} \otimes \rho^{(AB)})U^\dagger \right]
            \end{equation}
            where the global unitary matrix $U$ satisfies the following commutation relations:
            \begin{subequations}\label{SLTO_constr}
            \begin{align}
            & [U,H^{(A)} + H^{(B)} + H^{\mathcal{B}^{(A)}} + H^{\mathcal{B}^{(B)}}] = 0 ,\label{SLTO_constr1}\\
            & \qty[U, \beta^{(A)}( H^{(A)} + H^{\mathcal{B}^{(A)}}) + \beta^{(B)}( H^{(B)} + H^{\mathcal{B}^{(B)}})] = 0 \label{SLTO_constr2}
            \end{align}
            \end{subequations}
        \end{defi}
        The first constraint \eqref{SLTO_constr1} imposes conservation of total system-bath energy across both systems. To understand the second one, one should consider the systems and baths in the thermodynamic limit. Then, the energy flow from or to the thermal bath corresponds to heat exchange,  
        which in this limit is equal to entropy change times bath's temperature. Thus, the second equation \eqref{SLTO_constr2}, conservation of weighted heat flow, is in fact encoding the ``second law of thermodynamics'' -- entropy conservation.
        We stress that when both thermal baths have the same temperature, semilocal thermal operations are equivalent to thermal operations on a joint system with one joint bath.
    
        The authors of \cite{Bera2021} provided a strong operational bound on possible transformations of bipartite states using SLTO:
    
        \begin{thm}[\cite{Bera2021}, Supplementary Information, Theorem 11]\label{SLTO_majo}
         The transition between two energy-incoherent states $\rho^{(AB)} \to \sigma^{(AB)}$ can occur under semi-local thermal operation if, and only if, the diagonal of $\rho^{(AB)}$ thermo-majorizes the diagonal of $\sigma^{(AB)}$ with respect to state $\gamma^{(A)}\otimes\gamma^{(B)}$.
        \end{thm}

        The above theorem can be extended by an important corollary. Let us perform an artificial rescaling of Hamiltonian $\tilde{H}^{(A)}  = \sqrt{\frac{\beta^{(A)}}{\beta^{(B)}}} H^{(B)},\,\tilde{H}^{(B)}  = \sqrt{\frac{\beta^{(B)}}{\beta^{(A)}}} H^{(B)}$, so that the Gibbs states with common inverse temperature $\tilde{\beta} =\sqrt{\beta^{(A)}\beta^{(B)}}$ and Hamiltonians $\tilde{H}^{(A)}$ and $\tilde{H}^{(B)}$ are the same as the original Gibbs states $\gamma^{(A)}$ and $\gamma^{(B)}$. Then the condition from Theorem \ref{SLTO_majo}, by the Theorem \ref{thm_HLPgeneralisation}, is equivalent to the preservation of global Gibbs state $\gamma^{(AB)} = \gamma^{(A)}\otimes\gamma^{(B)}$ at a common temperature $\tilde{\beta} =\sqrt{\beta^{(A)}\beta^{(B)}}$. Thus, we have.

        \begin{cor}[\cite{Bera2021}, Supplementary Information, Corollary 4]\label{cor:SLTO_sim}
        The transition between two energy-incoherent states $\rho^{(AB)} \to \sigma^{(AB)}$ can occur under semi-local thermal operation if, and only if, there exist stochastic matrix preserving product of Gibbs states $M(\gamma^{(A)}\otimes\gamma^{(B)}) = \gamma^{(A)}\otimes\gamma^{(B)}$, that maps the spectrum of one state into the other $M(\text{diag}(\rho^{(AB)})) = \text{diag}(\sigma^{(AB)})$. 
        \end{cor}
        Therefore, if one discusses only energy-incoherent states, SLTO can be effectively treated as stochastic Gibbs-preserving matrices. Note that those matrices are defined by simple linear equations from Corollary \ref{cor:SLTO_sim} and stochasticity conditions, which let us state new corollary:
        \begin{cor}[Set-closure of SLTO]
        The set of possible transitions between finite-dimensional energy-incoherent states $\rho^{(AB)} \to \sigma^{(AB)}$ under SLTO is closed.
        \end{cor}
        The study of SLTO action on coherent states remains largely unexplored, thus in Appendix \ref{App:coch_evol} we develop preliminary result in this direction, which will prove useful in a later part of the manuscript.

        \medskip
        Finally, we present two additional properties of SLTO, which are crucial from the perspective of resource theories and important in further part of this work but, to our knowledge, were not proven before.
    
        \begin{restatable}{thm}{SLTOclosed}
        \label{SLTO_closed} \emph{(Composition-closure of SLTO)}
        The set of semilocal thermal operations is closed under composition, i.e if $\Lambda_1^{(AB)}$ and $\Lambda_2^{(AB)}$ are SLTO then $\Lambda_2^{(AB)} \circ \Lambda_1^{(AB)}$ is SLTO with the same pair of inverse temperatures $\beta^{(A)},\,\beta^{(B)}$.
        \end{restatable}
    
        \begin{restatable}{thm}{SLTOconvex}
        \label{SLTO_convex} \emph{(Convexity of SLTO)}
        The set of semilocal thermal operations is convex, i.e if $\Lambda_1^{(AB)}$ and $\Lambda_2^{(AB)}$ are SLTO then $ \alpha \Lambda_1^{(AB)}  + (1-\alpha) \Lambda_2^{(AB)}$ is SLTO with the same inverse temperatures for any $\alpha \in [0,1]$.
        \end{restatable}

        \noindent Proofs of Theorems \ref{SLTO_closed} and \ref{SLTO_convex} are provided in Appendix~\ref{App:semilocal}. \qed

        In what follows we will establish relation between the general framework of SLTO and a communication-constrained scenario of local thermal operations and classical communication, which will be described at length in the sections to follow.

	
	\section{Framework of Local Thermal Operations and Classical Communication}\label{sec:framework}
	
	
	Having presented all the necessary foundational ideas and concepts, we may proceed to the introduction of the central idea of this work, combining spatially separated parties with thermal restrictions on the local operations that they can effect, with classical communication between them -- \textit{Local Thermal Operations and Classical Communication} (LTOCC).
	
	\subsection{Local thermal operations and preservation of product Gibbs state}
	
	Let us start by considering a scenario in which Alice and Bob have two separate physical subsystems and each of them is able to perform thermal operations locally.
	Since two subsystems are physically separated, we assume that the interactions between them are negligible. Otherwise, the parties might leverage those interactions, effectively treating two subsystems as one joint system.
	
	Therefore, the total Hamiltonian can be expressed as a sum of local Hamiltonians $H = H^{(A)} + H^{(B)}$. Hence, the thermal state of the entire system is a product of thermal states on each subsystem, ${\gamma}^{(AB)} = {\gamma}^{(A)} \otimes {\gamma}^{(B)}$ together with their corresponding local inverse temperatures $\beta^{(A)}$ and $\beta^{(B)}$. This assumption, natural for separated subsystems, results in a crucial property: the product of any local thermal operations is (semilocal) thermal as well. 
	
	\subsection{Classical communication and Landauer cost of measurement} \label{sec:landauer_cost}
	
	The final ingredient one requires before going into the framework of LTOCC is the problem of local measurements and classical communication between the parties, as there are certain subtleties that need to be discussed, especially the thermodynamic cost of measuring the system in the energy eigenbasis. Such measurements can be emulated using the SLTO framework given memory erasure after each round of measurement and forwarding of information, which we demonstrate explicitly in Appendix~\ref{App:semilocal}. This result is similar in spirit to the incoherently conditioned thermal operations as a subset of TO, as demonstrated in \cite{Watanabe2024}. Below, we lay down more heuristic arguments and point to a possible hidden cost connected with Landauer erasure. 
	
	The energy-incoherent quantum states \mbox{$\rho = \sum_i p_i \; \op{E_i}$} can be interpreted as a probabilistic combination or a statistical ensemble of pure Hamiltonian eigenstates. The projective measurement onto the energy eigenbasis collapses it to a random eigenstate $\ket{E_i}$ with probability $p_i = \ev{\rho}{E_i}$, thus single-shot measurement is not a thermal operation -- it can project a Gibbs state onto a high or low energy eigenstate, resulting in a gain of free energy, otherwise inaccessible by thermal operations alone. 
    However, statistically, incoherent measurements leave diagonal states invariant $\sum_i |E_i\ra\la E_i| \la E_i|\rho|E_i\ra = \sum_i p_i \; \op{E_i} = \rho$. Thus while performing measurement on consecutive states from a Gibbs ensemble, with respect to energy eigenbasis $N$ times,
	the discrepancy between an ensemble of measurement outcomes and initial distribution vanishes on average as we go into the limit of a large number of rounds, $N\rightarrow\infty$. 
    \medskip

    The second element is the transfer and storage of information. According to Landauer’s principle~\cite{Landauer,BENNETT2003501}, erasing information in a classical memory incurs a fundamental thermodynamic cost. In our analysis, we assume that the classical computers used to store measurement outcomes are thermodynamically isolated from the quantum system, except for the moment of saving the outcome. Thus, although the erasure of classical memory (required between rounds or runs of the protocol) carries a cost of approximately $N k_B T \log d$ for $N$ measurements (with $d$ being the dimension of the quantum system), this cost does not affect the quantum system itself. Here, ``erasure'' refers solely to the resetting of the classical memory used in the LOCC protocol, not to any direct manipulation of the quantum state. Consequently, the Landauer cost can be disregarded when considering the quantum system’s thermodynamics. 

    Additionally, we note that arguments made by Watanabe and Takagi in~\cite{Watanabe2024} do not translate to our setting, as Alice and Bob do not have access to shared bath, and both systems may be in different temperatures. Thus, in order to formalise heuristic considerations put forward above we explicitly consider measurement in energy eigenbasis as a channel $\Phi: \mathcal{H}_A \mapsto \mathcal{H}_A\otimes \mathcal{M}$ represented by a set of operators $K_i = \qty(\ket{E_i}\otimes\ket{i}_\mathcal{M})\bra{E_i}$ for all $i$, which imprint the detected information about the energy level $E_i$ on an auxiliary (classical) memory system $\mathcal{M}$ and assume that energies are non-degenerate. It is straightforward to see that after discarding memory the map returns pinched state $\Tr_\mathcal{M}\qty[\Phi(\rho)] = \mathcal{P}(\rho) = \sum_i \ev{\rho}{E_i} \op{E_i}$. Pinching map, in general, preserves all energy-incoherent states and, in particular, the Gibbs state.  Furthermore, the use of classical information obtained from measurement to apply a classically-conditioned map is described by channel $\Xi:\mathcal{M}\otimes\mathcal{H}_B\mapsto\mathcal{M}\otimes\mathcal{H}_B$ defined as $\Xi(\mu\otimes\sigma) = \sum_i \la i |\mu|i\ra | i\ra\la i| \otimes\Xi_i(\sigma)$, with $\qty{\Xi_i:\mathcal{H}_B\mapsto \mathcal{H}_B}$ being a collection of CPTP maps. From this discussion one can infer that any further manipulation of subsystems $\mathcal{M}$ and $\mathcal{H}_B$ cannot alter the marginal belonging to the system $\mathcal{H}_A$:
    \begin{equation*}
    \begin{aligned}
    & \Tr_{B,\mathcal{M}}\qty[(\mathbb{I}_A \otimes\Xi)\,\circ\,(\Phi\otimes\mathbb{I}_B)\qty(\rho\otimes\sigma)] = \Tr_\mathcal{M}\qty[\Phi(\rho)] \\
    & = \mathcal{P}(\rho)  = \rho.
    \end{aligned}
    \end{equation*}
    We highlight that the above consideration assumes free access to, updating and erasure of a memory register -- which is far from obvious, but explained and justified in the prior paragraphs.
    \medskip

    Furthermore, let us highlight that one could extend the possible measurements to more general POVMs, e.g. with operators of the form $K_i = \sum_j a_{ij}\op{E_j}$ such that for all $j$ $\sum_i a_{ij}^2 = 1$, while satisfying the Gibbs preservation condition. There are three points to be made in this context. Firstly, when restricted to energy-incoherent states, such operations can be easily achieved equivalently by classical post-processing on the memory side. The second, more important aspect is that such measurements would in general provide partial coherence and/or entanglement preservation while giving only imperfect information about the system. 
    Thirdly, althou such energy-incoherent measurements preserve Gibbs state, in case of general quantum states it may not be strong enought restriction to consider it ``thermal'', due to the distinction between thermal operations and Gibbs-preserving operations, which emerge on the quantum level.
    However, due to their connection to coherence such measurements fall outside of the scope of this manuscript, and will be considered in detail in a future work~\cite{Czartowski2025QuantumLTOCC}.

	Finally, because statistically neither (conditional) thermal operations nor measurement in energy eigenbasis can create any coherences, the description of energy-incoherent states by probability distributions is complete, and instead of quantum channels, we can continue our discussion using stochastic maps.
    
	\subsection{Framework outline}
    \label{sec:hierarchy}
    
    In what follows for simplicity of exposition, we will focus on the bipartite scenario, with extension of the general multipartite definitions following in a similar manner.
	
	We start our discussion of local thermal operations aided by classical communication with one round scenario restricted by operational consideration, where Alice and Bob are restricted to thermal operations in each of their respective subsystems, but can also perform incoherent measurements. We will denote local Gibbs state as $\gamma^{(A)}$ and $\gamma^{(B)}$, respectively.
    \medskip
    
    The protocol proceeds as follows: in the first step, Alice can perform a measurement on her subsystem and send the information to Bob who, depending on the received information, can perform different thermal operations; in other words, Bob performs local thermal operations conditioned on the information obtained by Alice from measurement. If we stack Bob operations conveniently in one stochastic tensor $T_{jkl}^{(B)}$, the discussed scheme corresponds to a stochastic matrix of the form  
	\begin{equation}
		M_{ij,kl} = \delta_{ik} T_{jkl}^{(B)}~,
	\end{equation}
	since the ensemble of Alice's states remains unchanged; keep in mind that the superscript ${}^{(B)}$ refers to Bob, the party performing thermal operations. 
    Combining that with possible local thermal operation $\Lambda^{(A)}$ on Alice's side
	\begin{equation}
    \label{new_one_ounds1}
		M_{ij,kl} = \Lambda_{ik}^{(A)} T_{jkl}^{(B)}
	\end{equation} 
	we obtain one-round \textit{Local Thermal Operations and Classical communication} (LTOCC).\footnote{Note that $\text{LTOCC} \;\text{``}\neq\text{''}\; \text{LOCC}_{d^2}+\text{TO}_{d^2}$, but rather $\text{LTOCC} \;\text{``}=\text{''}\; \text{LTO}_{d}^{\otimes 2} + \text{CC}_d^2$, with subscripts pointing to dimensions of the involved subsystems. In particular, the former would result in a 
    theory identical to LOCC when restricted to classical states.
    }. The mirror scenario of measurement on Bob's side corresponds to
    \begin{equation}
    \label{new_one_ounds2}
        M_{ij,kl} = T_{ikl}^{(A)} \Lambda_{jl}^{(B)}~
    \end{equation}  
    From the thermality of each conditioned thermal operation in~\eqref{new_one_ounds1} and~\eqref{new_one_ounds2} one obtains thermal restrictions for the tensors $T$: 
	\begin{equation}
    \label{thermal_cond}
		\forall_l\sum_{k} T_{ikl}^{(A)} \gamma_k^{(A)} = \gamma_i^{(A)} \;,~~ 
		\forall_k \sum_{l} T_{jkl}^{(B)} \gamma_l^{(B)} = \gamma_j^{(B)} \;.
	\end{equation}
	We name any stochastic tensor satisfying one of the above requirements a \textit{thermal tensor}.

	The framework of local thermal operations and classical communication can be easily extended to multiple rounds without memory
	\begin{align}
		M = \prod_{\alpha}  M^{(\alpha)} ~~\text{where}~~& M_{ij,kl}^{(\alpha)} = \Lambda_{ik}^{(A,\alpha)} T_{jkl}^{(B,\alpha)} \nonumber \\ ~~\text{or}~~&M_{ij,kl}^{(\alpha)} =T_{ikl}^{(A,\alpha)} \Lambda_{jl}^{(B,\alpha)}~,\label{eq:LTOCC_n_no_mem}
	\end{align}
	with $\alpha$ indexing the consecutive rounds. We denote all operations realizable using $n$-round protocol of the form~\eqref{eq:LTOCC_n_no_mem} as $\text{LTOCC}_n$. A schematic representation of this multiround protocol is provided in Fig.~\ref{fig:LTOCC2}.
    \medskip
    
    Further extension can be made by allowing both Alice and Bob to retain a memory of measurement results from previous rounds, referred to as LTOCC+M.
    \begin{equation}
        \vb{M}\qty(\vb{p}\otimes\vb{q}) = \Tr_{\mathcal{M}_1,\hdots,\mathcal{M}_n}\qty(\qty[\bigcirc_{\alpha} \vb{M}^{(\alpha)}]\qty(\vb{p}\otimes\vb{q})).
    \end{equation}

    Each round consists of measurement on one subsystem combined with local thermal operations conditioned by all previous measurements, ie: 
    
    \begin{align}
        & r_{ij,m_1\hdots m_{\alpha}}^{\alpha} = \left[\vb{M^{(\alpha)}}(\vb{r}^{(\alpha-1)})\right]_{ijm_1\hdots m_{\alpha}} \nonumber \\
        = & \sum_{kl} {T^{(A,\alpha)}}_i^{k;\,m_1 \hdots m_{\alpha-1}} \delta_{k,m_{\alpha}} \; {T^{(B,\alpha)}}_j^{l;\,m_1 \hdots m_{\alpha-1} m_{\alpha}} \times \nonumber \\
        {\color{white}=} & {\color{white}\sum_{kl}}\times r_{kl,m_1\hdots m_{\alpha-1}}^{(\alpha-1)} 
    \label{LTOCCM_def}
    \end{align}
    where $m_{i},\hdots, m_{\alpha}$ are indices related to the memory registers, $k,l$ are system-related input state indices, and $i,j$ are system-related output-state indices. 
    Here, $\delta_{k,m_\alpha}$ is a measurement on Alice's subsystem, recorded on a new classical register $m_\alpha$. This is followed by thermal operation ${T^{(A)}}_i^{k;\,m_1 \hdots m_{\alpha-1}}$ conditioned on all previous measurement results $m_1\hdots,m_{\alpha - 1}$.
    Furthermore  ${T^{(B)}}_j^{l;\,m_1 \hdots m_{\alpha-1} m_{\alpha}}$ is a thermal operation on Bob's side conditioned by the measurement's outcome in previous rounds $m_{1},\hdots,m_{\alpha-1}$ and new measurement result on the Alice side $m_{\alpha}$. 
    We provide a diagrammatic representation of this protocol in Fig.~ \ref{fig:LTOCC2M}. Additional detailed discussion of 1- and 2-round LTOCC protocols with and without protocols is provided in Appendix~\ref{App:Geo_evol}. 
    
    \medskip

We supplement the exhibition of multiround LTOCC protocols with and without memory with two simple observations regarding inclusions.

    \begin{obs} \label{obs:more_rounds_inclusion}
        $\text{LTOCC}_n$(+M) forms a superset of $\text{LTOCC}_{n'}$(+M) for any $n'\leq n$. 
    \end{obs}
    \begin{proof}
        The observation follows by considering $\text{LTOCC}_{n+1}$(+M), where Bob measures in $n+1$-st round, and setting $T^{(A,n+1)}$ and $\Lambda^{(B,n+1)}$ ($T^{(B,n+1)}$) to unconditional identity operations, thus realising arbitrary $n$-round protocol. The claim follows by induction.
    \end{proof}

\begin{figure}[h]
\centering
\begin{minipage}[h]{0.49\textwidth}
\centering
    \renewcommand{\thefigure}{3.a}
\caption{$\text{LTOCC}_2$}
\scalebox{0.7}{
\Qcircuit @C=1.0em @R=1.0em @!R { \\
	\lstick{{m}_{1} :  }&  \ar @{<=} [1,0] & \control \cw \cwx[2]  & \barrier[0em]{2} &  & \ar @{<=} [2,0] & \control \cw \cwx[1] & &\\
	\lstick{A :  }& \meter & \qw & \gate{\Lambda^{(A,1)}} & \qw & \qw & \gate{T^{(A,2)}} & \qw & \qw \\
	\lstick{B :  }& \qw & \gate{T^{(B,1)}} & \qw & \qw & \meter & \qw & \gate{\Lambda^{(B,2)}} & \qw \\  \\ \\}}
\label{fig:LTOCC2}
\end{minipage} 
\hfill
\begin{minipage}[h]{0.49\textwidth}
\centering
    \renewcommand{\thefigure}{3.b}
\caption{$\text{LTOCC}_2+\text{M}$}
\scalebox{0.7}{
\Qcircuit @C=1.0em @R=1.0em @!R { \\
	\lstick{{m}_{1} :  }&  \ar @{<=} [1,0] & \control \cw^(0.0){^{\mathtt{}}} \cwx[2]  & \cw \barrier[0em]{3} & \cw & \cw & \control \cw^(0.0){^{\mathtt{}}} \cwx[1] & \control \cw^(0.0){^{\mathtt{}}} \cwx[2] \\
	\lstick{A :  }& \meter & \qw & \gate{\Lambda^{(A,1)}} & \qw & \qw & \gate{T^{(A,2)}} & \qw & \qw \\
	\lstick{B :  }& \qw & \gate{T^{(B,1)}} & \qw & \qw & \meter & \qw & \gate{T^{(B,2)}} & \qw \\
	\lstick{{m}_{2} :  }&   & &  &  &  \ar @{<=} [-1,0] & \control \cw^(0.0){^{\mathtt{}}} \cwx[-2] &  &  \\
\\ }}
\label{fig:LTOCC2M}
\end{minipage}
    \setcounter{figure}{2}
    \renewcommand{\thefigure}{\arabic{figure}}
\caption{
\textbf{Two-round LTOCC with and without memory:} In both versions of $\text{LTOCC}_2$ protocol first round consists of Alice measuring her subsystem, storing the result in the memory register $m_1$, which is then sent to Bob to implement conditional thermal operation, represented by the thermal tensor $T^{(B,1)}$, which is followed with Alice implementing a local post-processing TO $\Lambda^{(A,1)}$.  In the no-memory case (a) in the second round the memory register $m_1$ is overwritten with Bob's measurement of his system and the operations are continued with the roles swapped. For memory case (b) the register $m_1$ is retained and Bob stores his result in the new register $m_2$. Subsequently, Alice performs operations conditioned on both $m_1$ and $m_2$, whereas Bob post-processes his subsystem conditioned on $m_1$ only.
}
\label{fig:ciruits}
\end{figure}
\newpage

    \begin{obs} \label{obs:memory_inclusion}
        $\text{LTOCC}_n+\text{M}$ forms a superset of $\text{LTOCC}_{n}$.
    \end{obs}
    \begin{proof}
        The observation follows by setting $T^{(A,n)}$ to be conditioned nontrivially only on $m_{n}$ and $T^{(B,n)}$ to be unconditioned in rounds when Bob measures, and similarly for the rounds when Alice measures.
    \end{proof}

    Note that we do not provide the arguments for proper inclusion relations between protocols with and without memory, as well as for protocols of different depths. Such a problem is nontrivial even for standard LOCC~\cite{Chitambar2014}.
\medskip
    
    Utilising two-round LTOCC+M one can construct a particularly interesting example, in which Alice and Bob measure their respective subsystems in parallel, and then operate on their respective systems based on information received from their counterpart. It can be realised in two rounds as follows. The first round consists of Alice measuring her system and storing the result in memory register $m_1$ without performing any postprocessing $T^{(A,1)}:=\Lambda^{(A,1)} = \id$. Bob also does not perform any action $T^{(B,1)} = \mathbb{I}$ (unconditionally). When the second round begins, Bob measures his system and stores the result in a second memory register $m_2$. Then we set Alice operation $T^{(A,2)}$ to be conditioned on $m_2$ only and, similarly, Bob's postprocessing $T^{(B,2)}$  on previously stored $m_1$. Evaluating~\eqref{LTOCCM_def} according to the above description leads to a compact 1-step transition matrix
	\begin{equation}
		\label{one_round}
		M_{ij,kl} = T_{ikl}^{(A)} T_{jkl}^{(B)}.
	\end{equation}
    In a slight abuse of terminology, we refer to this particular restricted scenario as \textit{1-round parallel LTOCC}, as it does not involve any feedback. This procedure is depicted in Fig.~ \ref{fig:ParLTOCC}.

    \begin{figure}[h]
    \centering
    \scalebox{.8}{
    \Qcircuit @C=1.0em @R=1.0em @!R { \\
    	  \lstick{{m}_{1}:}&  \ar @{<=} [1,0] \barrier[0em]{3} & \cw & \cw & \cw & \control \cw^(0.0){^{\mathtt{}}} \cwx[2] &  &  \\
    	\lstick{A:}&\meter & \qw & \qw & \gate{T^{(A)}} & \qw & \qw & \qw \\
    	\lstick{B:}& \qw & \qw & \meter & \qw & \gate{T^{(B)}} & \qw & \qw \\
    	  \lstick{{m}_{2}:}&  & &   \ar @{<=} [-1,0] & \control \cw^(0.0){^{\mathtt{}}} \cwx[-2]  &   \\
    \\ }}
    \caption{\textbf{Parallel $\text{LTOCC}$ as a special case of $\text{LTOCC}_2+\text{M}$:} Parallel protocol can be implemented as a variant of 2-round protocol with memory by limiting first round of the protocol to measurement only and then conditioning the local operations on the memory register of the other party -- Alice on Bob's measurement and Bob on Alice's measurement.}
    \label{fig:ParLTOCC}
    \end{figure}

    \medskip
	Finally, one may further extend the framework by introducing shared randomness (LTOCC+R). Alice and Bob can share a random variable $\lambda_\alpha$ and perform different LTOCC depending on its value:
	
	\begin{equation}
		M_{ij,kl} = \sum_\alpha \lambda_\alpha \Lambda_{ik}^{(A,\alpha)} T_{jkl}^{(B,\alpha)}  ~~\text{or}~~ M_{ij,kl} = \sum_\alpha \lambda_\alpha T_{ikl}^{(A,\alpha)} \Lambda_{jl}^{(B,\alpha)}
	\end{equation}

    Below, we demonstrate that shared randomness provides an actual advantage over LTOCC without shared randomness, showing that there exist operations which can be realised by LTOCC+R and not by LTOCC.

    \begin{thm}\label{thm:strict_sup_LTOCC}
        Local Thermal Operations and Classical Communication (LTOCC) form a strict subset of LTOCC with randomness (LTOCC+R).
    \end{thm}

    \begin{proof}
    To show that LTOCC is a proper subset of LTOCC+R we first show that LTOCC is not a convex set and then that LTOCC+R is its convex hull. Hereafter, for the sake of clarity, we focus on one round of the protocol, however, the presented arguments hold for any number of rounds. 
    
    Let us first note that, similarly to thermal operation $\Lambda$, a set of thermal tensors $T_{ikl}$ forms a convex polytope since it is defined by linear relations together with the linear inequalities $0\leq T_{ikl} \leq 1$. Thus each tensor $T_{ikl}$ can be decomposed as a convex combination of extremal tensors $T_{ikl} = \sum_{\alpha_1} t_{\alpha_1} T_{ikl}^{\alpha_1}$ with some weights $t_{\alpha_1}$; similarly for $\Lambda_{ik} = \sum_{\alpha_2} u_{\alpha_2} \Lambda_{ik}^{\alpha_2}$. However, the set of LTOCC maps $M_{ij,kl} = \Lambda_{ik}^{(A)}T_{jkl}^{(B)}$ is not convex, since each map can be decomposed into extremal LTOCC using product distribution $M_{ij} = \sum_{\alpha_1, \alpha_2} t_{\alpha_1}u_{\alpha_2} \Lambda_{ik}^{\alpha_2} T_{jkl}^{\alpha_1} =: \sum_{\alpha_1,\alpha_2} t_{\alpha_1}u_{\alpha_2} M_{ij,kl}^{\alpha_1,\alpha_2}$. Since the product distributions form proper subsets of all two-dimensional distributions we infer some combinations of extremal maps are not in $LTOCC$, thus it is not a convex set. 
    
    On the other hand, the introduction of shared randomness enables one to perform arbitrary convex combinations of extremal channels $M_{ij} = \sum_{\alpha_1, \alpha_2} \lambda_{\alpha_1, \alpha_2}  \Lambda_{ik}^{\alpha_1} T_{jkl}^{\alpha_2}$, thus enlarging the framework by enabling arbitrary convex combinations. This shows that $\text{LTOCC}+\text{R} \supset \text{LTOCC}$.
    \end{proof}

    Similar statements follow \textit{mutatis mundis} when considering multi-round protocols with the same number of rounds $n$ both with and without memory, when comparing variants with and without randomness (R).
    \medskip

    Finally, we present the main property of most of the introduced protocols, which is simultaneously a main restriction of their capabilities:

    \begin{restatable}{thm}{subsetSLTO}
    \label{thm:subset_of_SLTO}
    Local thermal Operations and Classical communication with randomness (LTOCC$_n$ +R) (with projective energy-incoherent measurements), form a subset of semilocal thermal operations (SLTO) (even for non energy-incoherent states).
    \end{restatable}

    \begin{proof}
        Proof of this theorem is given in Appendix~\ref{App:semilocal}. 
    \end{proof}

    We stress that the following theorem describes only protocols without memory. If we allow the use of outcomes of some measurements in multiple rounds, the class of obtained protocols might be noticeably more powerful, as presented in Appendix~\ref{App:Geo_evol}.

    \begin{cor}[Thermomajorisation for LTOCC] \label{corr:SLTO_major}
        Consider two systems $A$ and $B$ with local Gibbs states $\gamma^{(A)}$ and $\gamma^{(B)}$, respectively. Given an input state $\vb{p}_{AB}$ and output $\vb{q}_{AB}$, the transition can be achieved using LTOCC$_n$ +R only if 
        \begin{equation}
            \vb{p} \underset{\gamma_{AB}}{\succ} \vb{q}
        \end{equation}
        where the thermomajorization is taken with respect to $\gamma_{AB} = \gamma^{(A)}\otimes\gamma^{(B)}$. In particular, the Gibbs state is preserved
        \begin{equation*}
            M\qty(\gamma^{(A)}\otimes \gamma^{(B)}) = \gamma^{(A)}\otimes \gamma^{(B)}
        \end{equation*}
    \end{cor}
    We stress that the above corollary describes the implication in only one direction. The implication in the opposite direction would imply that LTOCC $=$ SLTO.
    \medskip

	Before finishing this section, we can impose an additional symmetry requirement. Consider a scenario in which Alice's and Bob's thermal states coincide $\gamma^{(A)} = \gamma^{(B)} =: \gamma$, and a parallel LTOCC protocol in which if Alice received a thermal state as an input, then the result of Bob's action would also be thermal and vice versa
	\begin{equation}
		\label{tristo_def}
		\begin{aligned}
		\forall_k\sum_{l} T_{ikl}^{(A)} \gamma_l = \gamma_i \;,~~ 
		\forall_l\sum_{k} T_{ikl}^{(A)} \gamma_k = \gamma_j \;
        \end{aligned}
	\end{equation}
	and analogously for $T_{jkl}^{(B)}$. Thus Alice and Bob could exchange their tensor without going outside of the framework of LTOCC.
	We call any stochastic tensor satisfying both properties~\eqref{tristo_def} a \textit{bithermal} tensor. The above scheme constructed from a bithermal tensor gives one-round \textit{Symmetric Thermal Operations and Classical Communication} ($\mathcal{S}$LTOCC). 

    \begin{thm}\label{thm:stric_sup_LTOCC_SLTOCC}
        Symmetric Local Thermal Operations and Classical Communication with $n$ rounds ($\mathcal{S}$LTOCC$_n$) form a strict subset of $2n$-rounds Local Thermal Operations and Classical Communication (LTOCC$_{2n}$).
    \end{thm}

    \begin{proof}
    The inclusion follows from the fact that each round $\mathcal{S}$LTOCC, by definition, can be implemented in two rounds of LTOCC. 
    The simplest proof of the proper inclusion is based on the observation that both $\mathcal{S}$LTOCC output marginals need to be Gibbs wherever the Gibbs state appears as one of the inputs. On the other hand, this is not the case for all LTOCCs, which can be presented using the identity protocol in which all (conditional) terminal operations are just identity.
    \end{proof}

    An alternative and more insightful argumentation requires the demonstration that $\mathcal{S}$LTOCC's building block -- bithermal tensor forms a proper subset of the set of thermal tensors. Here we focus once again on one round of $\mathcal{S}$LTOCC for convenience notation, however, the presented argument is general.
    To do so, we note that $T_{ijk}$ in $\mathcal{S}$LTOCC is restricted by more linear conditions than in LTOCC~\eqref{tristo_def}. 
    Furthermore, each equation defining thermal tensors $\forall_l\sum_{k} T_{ikl}^{(A)} \gamma_k = \gamma_i $ does not include the elements from different layers of the tensor (different values of $l$), whereas the second set of equalities does so. Hence, one set of equations cannot impose the other, so the additional equalities provide further restrictions.
    The proper inclusion of the bithermal tensors' set inside the thermal tensors' set for local dimension $d = 2$ is also presented in the next section.

    \begin{thm} \label{thm:stric_sup_SLTOCC}
        Symmetric Thermal Operations and Classical Communication ($\mathcal{S}$LTOCC) form a strict subset of $\mathcal{S}$LTOCC with randomness.
    \end{thm}

    \begin{proof}
        The proof follows \textit{mutatis mutandis} the line of reasoning of Theorem \ref{thm:strict_sup_LTOCC}.
    \end{proof}

    Summary of inclusion relations between variants of LTOCC model with different amounts of memory and control, as discussed in this section, is presented pictorially in Fig.~\ref{fig:inclusions}.

    \medskip
    Finally, we put forward a conjecture concerning the relation between LTOCC and SLTO in general.

    \begin{conj}\label{conj:subset_of_SLTO}
        Local thermal operations and classical operations with an arbitrary number of rounds and shared randomness yield a set of operations, which closure is equivalent to the semilocal thermal operations when we restrict both classes to energy-incoherent states,
        \begin{equation}
            \overline{\bigcup_{n=1}^\infty (\text{LTOCC}_n + \text{R})} = \text{SLTO}.
        \end{equation}
    \end{conj}
    Should the above conjecture, its stronger version without shared randomness or extension to genuinely quantum states, hold true, it would have interesting consequences, especially taking into account that considerations in \cite{Bera2021} have been carried out appealing only energy-incoherent states -- \textbf{finite-size engine operating with Carnot efficiency could be realised with a (potentially finite-round) LTOCC protocol without appealing to its quantum properties.} Furthermore, in LTOCC$_n$ +R all interactions between subsystems and by extension between thermal baths are purely classical, which is not the case for SLTO. Hence, the above result would indicate that at least for energy-incoherent states, the introduction of coherent interactions does not increase the set of possible operations. \bl{At the same time, it is worth noting that, in line with the discussion provided in Section \ref{sec:landauer_cost}, LTOCC protocols involve external costs connected to continued writing and erasure of classical memory used to track measurement results. Thus, the question of whether the energetic cost of memory impedes attaining Carnot limit merits further investigation.}

    \begin{figure*}[t!]
        \centering
        \begin{tikzpicture}[
        mybox/.style={draw, inner sep=2pt},
        mybox2/.style={draw, inner sep=5pt},
        mybox3/.style={draw, dashed, inner sep=8pt}]

    \node[mybox](box1) at (0,0) {
        \begin{tikzcd}[column sep=.5em]
        ~~~\, \text{LTOCC}_1 \arrow[r, symbol=\subset] \arrow[d, symbol=\subseteq] ~~~&~~~ \text{LTOCC}_1 + \text{R}\arrow[d, symbol=\subseteq] ~~~\, \\
        \text{} ~~&~~ \text{}  \\[-50 pt]
        \end{tikzcd}
    };

    \node[mybox2] at (0,-1.625) [minimum width=7.1cm, minimum height=5cm] (box2) {};
    \node[mybox3] at (0,-1.95) [minimum width=7.3cm, minimum height=5.9cm] (box2) {};

    \node at (0,-2.4) {
        \begin{tikzcd}[column sep=.5em]
        \text{LTOCC}_2 \arrow[r, symbol=\subset] & \text{LTOCC}_2 + \text{R} \\[-25pt]
        \arrow[d, symbol=\subseteq] & \arrow[d, symbol=\subseteq] \\
        \text{LTOCC}_2 + \text{M} \arrow[r, symbol=\subset] & \text{LTOCC}_2 + \text{R} + \text{M} \\
        \arrow[u, symbol=\subset] & \arrow[u, symbol=\subset] \\[-25pt]
        \mathcal{S}\text{LTOCC}_1 \arrow[r, symbol=\subset] & \mathcal{S}\text{LTOCC}_1 + \text{R} \\[-20 pt]
        \end{tikzcd}
    };

    \node[shift = {(4.9,0)}]{1 round};
    \node[shift = {(6.5,0)}]{$\subset$ SLTO};
    \node[shift = {(4.9,-1)},rotate = -90]{$\subseteq$};
    \node[shift = {(4.9,-2.2)}]{2 rounds};
    \node[shift = {(4.9,-2.7)},rotate = -90]{$\subseteq$};
    \node[shift = {(4.9,-3.35)},rotate = -90]{$\cdots$};
    \node[shift = {(4.9,-4)},rotate = -90]{$\subseteq$};
    \node[shift = {(4.9,-4.5)}]{n rounds};
\end{tikzpicture}
    \caption{
    \textbf{LTOCC inclusion hierarchy:}
    Each frame corresponds to a different number of rounds in the protocol. In general, variants with shared randomness (+R) contain the ones not including it (Theorems \ref{thm:strict_sup_LTOCC} and \ref{thm:stric_sup_SLTOCC}), $n$-round protocols contain all $n'$-round protocols for $n'\leq n$ (Observation \ref{obs:more_rounds_inclusion}), protocols with memory (+M) contain the ones without it (Observation \ref{obs:memory_inclusion}), and protocols without memory (including all one round protocols) form a subset of SLTO (Theorem \ref{thm:subset_of_SLTO}). 
    }
    \label{fig:inclusions}
    \end{figure*}
    \medskip

    For the sake of simplicity we introduced the entire LTOCC($+$M) framework in  restriction to the energy-incoherent states. An illustrative example of this paradigm -- implementing classical logic gates CNOT and SWAP by LTOCC is presented in Appendix \ref{app:subsec_logic_gates}. However, they are well defined for general quantum states as well, see Appendix \ref{App:semilocal}. Furtheremore in Appendix \ref{App:coch_evol} we provide partial resoults on coherences' evolutions under LTOCC($+$M).

    \section{Basic building blocks -- thermal and bithermal tensors}

    \label{sec:BTTensors}

	In this section, we shift our attention to the new building block of the presented LTOCC framework -- thermal tensor -- since the properties of thermal operations have already been widely studied, e.g.~\cite{Mazurek_2018}. 
    We start with the observation that the layers of thermal tensors are independent Gibbs-preserving stochastic operations. Therefore the space of possible thermal tensors as a whole can be summarised easily as $d$-th tensor power of the space of Gibbs-preserving matrices, which has already been studied in detail~\cite{Birkhoff_thm, Mazurek_2018, Low_temerature_cooling_maps}. Thus, solutions to certain problems, like extremal operations, are straightforward, given that they are known for Gibbs-preserving matrices.  
    
    Taking this into account, we narrow down a majority of the discussion to a mathematically interesting and unexplored subset constituting the building blocks of $\mathcal{S}$LTOCC -- bithermal tensors, with the additional requirement of symmetry, to present its wide range of properties in such a restricted case. Additionally, we note that this set is interesting in its own right due to its semblance to already explored problems -- bithermal and tristochastic tensors bear a similar relation as thermal-stochastic and bistochastic matrices. We start the discussion by considering the geometry of the bithermal tensors' set particularly to unveil the extremal operations -- a thermal equivalent of study of Birkhoff-like polytope for tristochastic tensors.
	
	Similarly to the tri-stochastic tensors, the bithermal tensors are defined by a finite set of linear (in)equalities~\eqref{tristo_def}, and thus they form a convex polytope. Moreover, to fully describe this set it is sufficient to consider only a discrete set of extremal points.
	Let us start by mentioning the limit of 
	infinite
	temperature, $\beta = 0$ -- tristochastic tensors.
	In this case, the problem of characterizing vertices was already extensively studied in~\cite{extremal_multistoch}, where they have been divided into the permutation and non-permutation extreme points.
    \medskip
	
	In the finite temperature case, $\beta > 0$,
	the situation is more convoluted, even when restricted to standard thermal operations~\cite{Lostaglio_2018}. As presented in~\cite{Mazurek_2018}, the structure of vertices for the thermal equivalent of Birkhoff polytope 
	depends on the relations between system's energies and temperature. Thus we start the discussion by presenting a simple, yet comprehensive, example of a bithermal operation on a two-level system.
	
	Note that for a two-level systems the defining relations of bithermal operations~\eqref{tristo_def}, determine all tensor values as a function of only one of them. Thus geometrically it is equivalent to a
	line segment. To describe its extremal points let us denote the energies by $E_0 < E_1$, moreover to simplify the notation we set $g_{1,0} = e^{-\beta (E_1 - E_0)}$. The extremal bithermal tensors are:
    \begin{widetext}
	
	\begin{equation}
		\label{A1}
		T_{ijk}^{(1)} = (T_{0,j,k}^{(1)}|T_{1,j,k}^{(1)})
		= \left(\begin{matrix}
			1-g_{1,0} & 1 \\
			1 & 0 \\
		\end{matrix}\right.
		\left|\begin{matrix}
			g_{1,0} & 0 \\
			0 & 1 \\
		\end{matrix}\right)~,
	\end{equation}
	
	\begin{equation}
		\label{A2}
		T_{ijk}^{(2)} = (T_{0,j,k}^{(2)}|T_{1,j,k}^{(2)}) 
		= \left(\begin{matrix}
			1-g_{1,0}(1-g_{1,0}) & 1-g_{1,0} \\
			1-g_{1,0} & 1 \\
		\end{matrix}\right.
		\left|\begin{matrix}
			g_{1,0}(1-g_{1,0}) & g_{1,0} \\
			g_{1,0} & 0 \\
		\end{matrix}\right)~.
	\end{equation}
        
    \end{widetext}
	Note that while the layers of $T^{(1)}$ are extremal thermal operations, it no longer holds for the layers of $T^{(2)}$. 
	This is because the tensor $T^{(2)}$ can be realized as $T^{(1)}$ followed by a thermal swap $S^{(\beta)}$, or as a $T^{(1)}$ pre-processed by thermal swap on one of subsystems: 
	\begin{equation}
    \label{t_swap}
        \begin{aligned}
		& T_{ijk}^{(2)} = \sum_l S_{il}^{(\beta)} T_{ljk}^{(1)} ~,~~ 
		T_{ijk}^{(2)} = \sum_l T_{ilk}^{(1)} S_{lj}^{(\beta)} ~, \\
		& T_{ijk}^{(2)} = \sum_l T_{ljl}^{(1)} S_{lk}^{(\beta)} ~,~~ 
		  \text{where }~ S^{(\beta)}=\left(
		\begin{matrix}
			1-g_{1,0} & 1\\
			g_{1,0} & 0
		\end{matrix}
		\right).
        \end{aligned}
	\end{equation}
	This interconnection between extremal tensors is standard for tri-stochastic tensors~\cite{extremal_multistoch}, however in finite temperature it results in one tensor being ``more thermalized'' than the second one and multiple applications of thermal swaps generally leads to non-extremal tensors.

    Before discussing larger local dimensions we would like to contrast the presented structure of bithermal tensors with thermal tensors for local dimension $2$. Following~\cite{Mazurek_2018} we note that in this dimension there are two extremal thermal operations identity $\id$ and thermal swap $S^{(\beta)}$~\eqref{t_swap}. Thus each of the two layers of extremal thermal tensor $T$ can be either extremal thermal operations giving four, not two, (linearly independent) extremal operations.
		
	\begin{figure*}[t]
		\centering
		\includegraphics[width=0.8\textwidth]{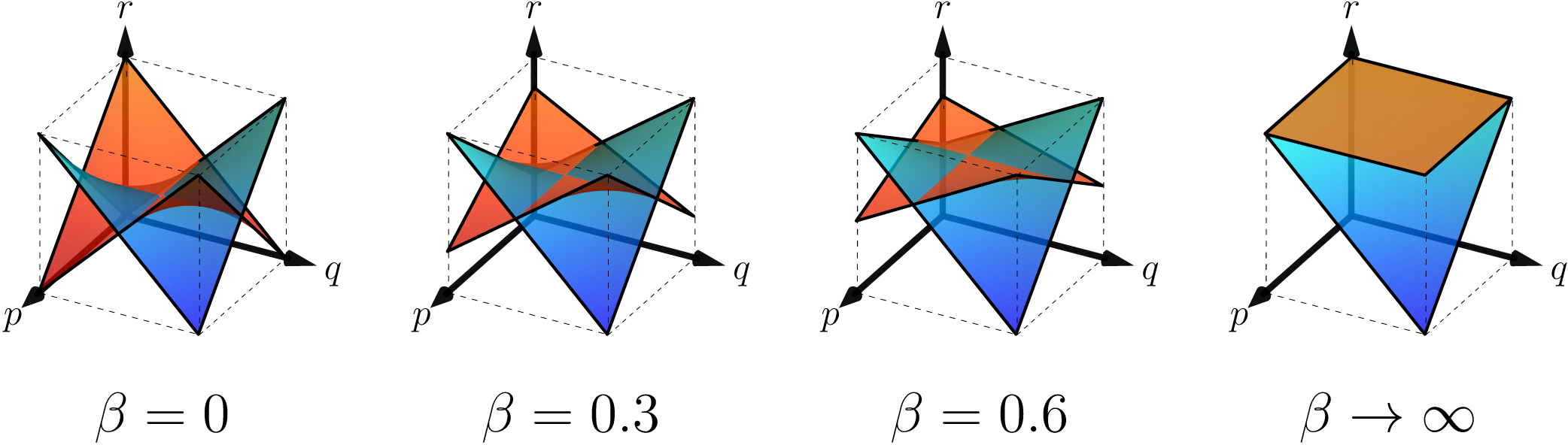}
		\caption{Attainable states $\vb{r} = T(\vb{p},\vb{q})$ in dimension $d = 2$ using LTOCC lie between the extreme surfaces generated by $T_1$ (red-orange) and $T_2$ (blue-cyan). Each distribution is characterized by its first coefficient, e.g. $\vb{p} = (p,1-p)$, and the extremal values are denoted by shaded lines.
        Although fully symmetric for $\beta = 0$, the symmetry is lost for finite temperatures, $\beta > 0$. }
		\label{fig:Attainable_states}
	\end{figure*}

    \subsection{Extremal points of thermal  ``Birkhoff'' polytope in infinite and large temperatures}

    We start this section by recalling 
    thermal and bithermal tensors in the limit of infinite temperatures $\beta \to 0$. In this case the condition of thermality~\eqref{thermal_cond} simplifies to:
    \begin{equation*}
     \forall_k   \sum_{j} T_{ijk} \eta_j = \eta_i \iff \forall_k \sum_{j} T_{ijk} = 1
    \end{equation*}
    where $\eta_i = 1/d$ is a flat state. Thus the layers of thermal tensors reduce to bi-stochastic matrices. Therefore, extremal stochastic tensors are the ones with permutation matrices as their layers since those are the extremal bi-stochastic matrices.

    Similarly, for bithermal tensors, their defining conditions~\eqref{tristo_def} in the infinite temperature limit reduces to the conditions for tristochastic tensors~\eqref{tristoch_def0}. Thus the extremal tensors in infinite temperature are, among others, permutation tensors.

	To tackle the general structure of bithermal operations' set in dimensions $d>2$ we start by considering the scenario of large temperatures as a perturbation from the infinite temperature $\beta = 0$. 
	To study the form of extremal bithermal tensors one can consider its series expansion around tri-stochastic tensors:
	\begin{equation}
		\label{A_high_T}
		T_{ijk} = \widetilde{T}_{ijk} + \beta A_{ijk}^{(1)} + O(\beta^2)~.
	\end{equation}
	In the limit of large temperature, we may expand the thermal state as
 
	Thus in the first order, the conditions for bithermality~\eqref{tristo_def} lead to:
    \begin{widetext}
	\begin{equation}
		\label{series_bounds}
		\sum_i A_{ijk}^{(1)} = 0 ,~~~~ \sum_j A_{ijk}^{(1)}  = \sum_{j} \tilde{T}_{ijk}E_j - E_i  ,~~~~ \sum_k A_{ijk}^{(1)}  = \sum_{k} \tilde{T}_{ijk}E_k - E_i  
	\end{equation}
    \end{widetext}
	Moreover, we must ensure the bithermal tensor $T_{ijk}$ is non-negative. Thus treating $A_{ijk}^{(1)}$ as tangent vector in the point $T_{ijk}$, this requirement corresponds to:
	\begin{equation}
		\label{series_borders}
		A_{ijk}^{(1)} \geq 0 ~\text{ if }~ \tilde{T}_{ijk} = 0. 
	\end{equation}
	
	The condition~\eqref{series_borders} does not restrict the values of all tangent vector $A_{ijk}^{(1)}$ elements since not all directions from $\tilde{T}_{ijk}$ are prohibited. In particular $A_{ijk}^{(1)}$ can always point in the centre of the set of bithermal tensors with arbitrary amplitude. 
	However, we are interested in the extremal bithermal tensors. Thus we focus on extremal tensor $T_{ijk}$ and extremal tangent vectors $T_{ijk}^{(1)}$, which saturate as many of inequalities~\eqref{series_borders} as possible.
	
	This approach not only lets us recover the first order of extremal tensors expansion in local dimension $d = 2$~\eqref{A1},\eqref{A2}, but also find the first order expansion in dimension $d = 3$ around an exemplar permutation tensor~\eqref{example_permutation_tensor}, which yields $67$ distinct approximants. The solutions, together with the code used to obtain them are presented in the GitHub repository~\cite{Github}. 
	
	We note the enormous growth of extremal tensors' number in finite temperature, even in such a small local dimension. This indicates that the complexity of the set of bithermal operations is immense even in comparison with standard thermal operations, for which in local dimension $d = 3$ there were no more than $3$ extremal operations corresponding to one permutation~\cite{Mazurek_2018}.

	\subsection{Bithermal tensors approaching zero temperature}
	
	The opposite extreme of zero temperature $\beta \to \infty$ turns out to be surprisingly simple. 
	As stated in~\cite{Low_temerature_cooling_maps}, for temperatures low enough, if the thermal bath has any spectral gap, effectively its entire population lies in the ground state. In such a scenario, thermal operations reduce to cooling maps defined by condition $\Lambda_{ij} = 0 $ if $i > j$. Moreover, in the Appendix~\ref{App:cold}, we present an alternative derivation of cooling maps, as a zero temperature limit $\beta\to \infty$ of thermal operations.
	
	Knowing the limit form of thermal operations at zero temperature, we can combine them into thermal and bithermal operations. 
	
	\begin{obs}[Extremal bicooling tensors]
		In zero temperature, assuming a non-degenerated Hamiltonian spectrum, the set of thermal and bithermal tensors simplify to polytopes of stochastic matrices satisfying respectively $T_{ijk}  = 0$ if $i > j$ and  
        \begin{equation}
        \label{ZetoTem}
        T_{ijk} = 0 \text{ if } i>j \text{ or } i >k
        \end{equation}
        The extremal thermal and bithermal tensors are the ones with only $0$ and $1$ entries, restricted only by the abovestated condition.
	\end{obs}

 Hereafter we will call the zero temperature version of bithermal operations the \textit{bicooling} operations.
	
	\begin{proof}
		For the thermal requirement in the zero temperature $\beta\to\infty$, we have $T_{ijk} = 0 \text{ if } i > j$ and $T_{ijk} = 0 \text{ if } i > k$, thus~\eqref{ZetoTem} follows trivialy.
		From the condition~\eqref{ZetoTem} we see that the columns of the (bi-)cooling tensor are decoupled and can be treated as independent probability vectors subject only to the condition~\eqref{ZetoTem}. 
		Thus the form of extremal tensors, as a combination of extremal columns, follows immediately.
	\end{proof}
	
	We note, that the set defined by the limiting form of bithermality conditions~\eqref{ZetoTem} might be larger than the zero-temperature limit of the set of bithermal tensors.
    However, we choose the first one as bicooling tensors, since in the discussed framework the fundamental physical operations are thermal operations. In the limit of zero temperature, one can combine them into any tensors satisfying conditions~\eqref{ZetoTem}.

	\subsection{Explicit construction of a family of extreme bithermal tensors in low temperatures}
	
	Although the general construction of extremal bithermal tensors seems unattainable, based on the results introduced in~\cite{Mazurek_2018} we may construct \bl{a particular} family of extreme bithermal tensors, which exist for every Hamiltonian when the temperature is low enough. 
	Hereafter, for the sake of clarity, we adopt the notation $g_{i,j} = e^{-\beta(E_i - E_j)}$~\cite{Mazurek_2018}. 
    In dimension $2$ our example is just an extremal tensor~\eqref{A1}
	\begin{equation}
		T_{ijk}|_{d  = 2} = (T_{i,0,k}|T_{i,1,k})_{d = 2}
		= \left(\begin{matrix}
			1-g_{1,0} & 1 \\
			g_{1,0} & 0 \\
		\end{matrix}\right.
		\left|\begin{matrix}
			1 & 0 \\
			0 & 1 \\
		\end{matrix}\right)~,
	\end{equation}
	Note, that we sliced the tensor differently so that each slice is now a thermal operation that Bob may perform.
    \begin{widetext}
    In dimension $3$ the structure is given as follows
	\begin{equation}
		\begin{aligned} 
			T_{ijk}|_{d  = 3} & = (T_{i,0,k}|T_{i,1,k}|T_{i,2,k})|_{d = 3}  \\
             & = \left(\begin{matrix}
				1-g_{1,0} - g_{2,0} & 1 & 1\\
				g_{1,0} & 0 & 0\\
				g_{2,0} & 0 & 0 \\
			\end{matrix}\right.
			\left|\begin{matrix}
				1 & 0 & 0 \\
				0 & 1-g_{2,1} & 1 \\
				0 & g_{2,1} & 0 \\
			\end{matrix}\right.
			\left|\begin{matrix}
				1 & 0 & 0 \\
				0 & 1 & 0 \\
				0 & 0 & 1 \\
			\end{matrix}\right) =\left(\begin{matrix} 
				1-g_{1,0} - g_{2,0} & \rvline & \begin{matrix} 1 & 1\\\end{matrix} \\
				\hline
				\begin{matrix}
					g_{1,0} \\ g_{2,0} \\
				\end{matrix} & \rvline &
				\bigzero
			\end{matrix} \right.
			~\left|~~\begin{matrix} 
				1 & \rvline & \begin{matrix} 0 & 0\\\end{matrix} \\
				\hline
				\begin{matrix}
					0 \\ 0 \\
				\end{matrix} & \rvline &
				T_{i,0,k}|_{d = 2}
			\end{matrix} \right.
			~\left|~~\begin{matrix} 
				1 & \rvline & \begin{matrix} 0 & 0\\\end{matrix} \\
				\hline
				\begin{matrix}
					0 \\ 0 \\
				\end{matrix} & \rvline &
				T_{i,1,k}|_{d = 2}
			\end{matrix} \right)
		\end{aligned}
	\end{equation}
	And in the higher dimensions we follow the same inductive procedure 
	\begin{equation}
		\label{example_n}
		\begin{aligned} 
			&T_{ijk}|_{d = n}= (T_{i,0,k}|\cdots|T_{i,n,k})|_{d = n} =\\ 
			& =\left(\begin{matrix} 
				1-\sum_{i = 1}^n g_{i,0} & \rvline & \begin{matrix} 1 & \cdots & 1\\\end{matrix} \\
				\hline
				\begin{matrix}
					g_{1,0} \\ \vdots \\ g_{d,0} \\
				\end{matrix} & \rvline &
				\bigzero
			\end{matrix} \right.
			~\left|~~\begin{matrix} 
				1 & \rvline & \begin{matrix} 0 &\cdots & 0\\\end{matrix} \\
				\hline
				\begin{matrix}
					0 \\ \vdots\\ 0 \\
				\end{matrix} & \rvline &
				T_{i,0,k}|_{d = n-1}
			\end{matrix} ~~\right|
			~~~\cdots~~~\left|~~\begin{matrix} 
				1 & \rvline & \begin{matrix} 0& \cdots & 0\\\end{matrix} \\
				\hline
				\begin{matrix}
					0 \\ \vdots \\ 0 \\
				\end{matrix} & \rvline &
				T_{i,n-1,k}|_{d = n-1}
			\end{matrix} \right)~.
		\end{aligned}
	\end{equation}
    \end{widetext}
	
	To guarantee the non-negativity of each element in the tensor we must demand
	\begin{equation}
		\sum_{i = 1}^n g_{i,0} \leq 1 ~,~~ \sum_{i = 2}^n g_{i,1} \leq 1 ~,~~ \cdots ,~ g_{n-1,n-2} + g_{n,n-2} \leq 1 ~.
	\end{equation}
	Each of those equations introduces some upper bound for the temperature, thus it is sufficient to satisfy the tightest of them. The above bithermal tensor is symmetric in the sense $T_{ijk} = T_{ikj}$, see fig. \ref{fig:Example_graphical}, thus the opposite way of slicing it would give the same formulas.
	
	Note that each layer of this tensor is an extreme thermal operation. It has a block structure consisting of identity and an extremal thermal operation presented in~\cite{Mazurek_2018} (see fig. \ref{fig:Example_graphical}). Since none of the tensor layers can be decomposed as a convex combination of thermal operations, the entire tensor itself cannot be either, thus it is extremal.
	
	Furthermore, each layer cannot be decomposed as a product of two other thermal operations, as demonstrated in~\cite{Mazurek_2018}, so this tensor cannot be obtained from anything ``simpler'' using local thermal pre- and post-processing.
	Finally, in the limit of zero temperature, the bithermal tensor~\eqref{example_n} simplifies to a bi-cooling tensor of least possible cooling i.e.
	\begin{equation}
		T_{ijk}|_{d = n,T = 0} = 
		\left\{\begin{matrix}
			1& \text{ if } i = \min(j,k), \\
			0& \text{ otherwise},
		\end{matrix}\right.
	\end{equation}
	thus in each column, the populations are placed at the highest possible temperature.
		
	\begin{figure}[h]
		\centering
		\includegraphics[width=1\linewidth]{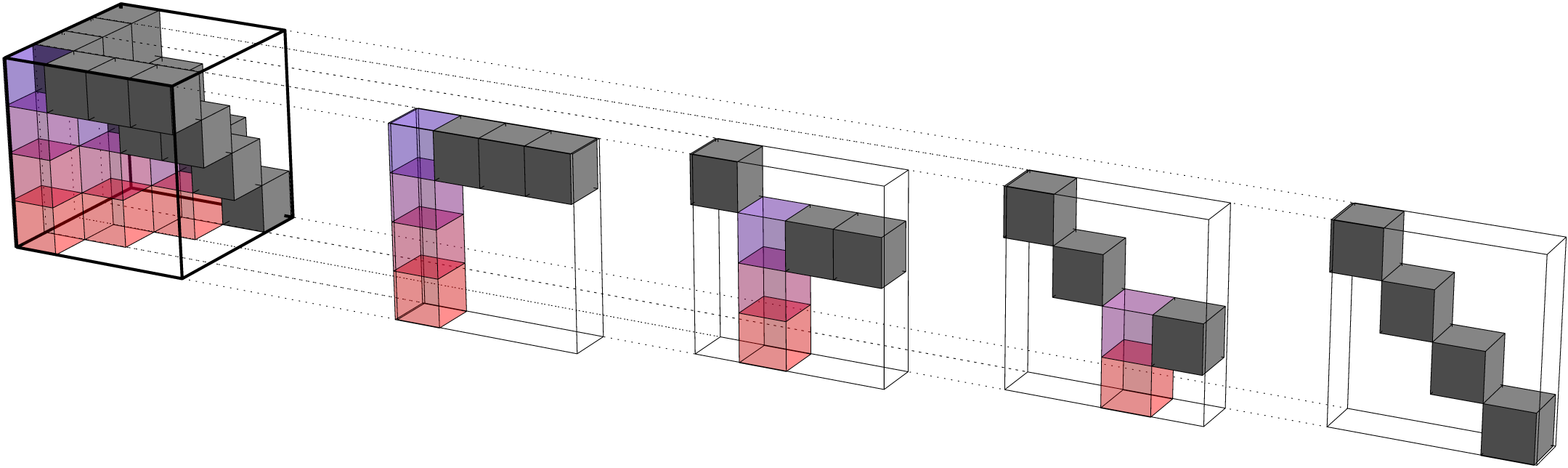}
		\caption{Graphical representation of extremal bithermal tensor~\eqref{example_n} in dimension $d = 4$ as a cube and its separation into consecutive slices. Grey cells correspond to value $1$, red-to-blue cells to other nonzero values, and empty cells to $0$.}
		\label{fig:Example_graphical}
	\end{figure}

	\subsection{Majorization-like restrictions on outputs of (bi)thermal tensors}
	
	Having discussed the general structure of the set of thermal and bithermal tensors, including methods for obtaining the extreme points, we shift our attention to the images of pairs of states under LTOCC and $\mathcal{S}$LTOCC. We start our discussion by considering local results. In other words, we will consider what Bob can obtain in his subsystem after reciving information from Alice, excluding the correlations between subsystems. Thus, we will focus on one (bi)thermal tensor, rather than the entire scheme.
	We stress that although the entire sub-section is presented from the perspective of thermal operations, the presented results hold also in the $\beta \to 0$ limit, corresponding to (tri)stochastic tensors. 
	
	It turns out that there is a strong connection between the images of (bi)thermal tensors and thermomajorization. In the following theorem in square brackets, we present its extended version for bithermal tensors.
	
	\begin{thm}[(Bi)thermality from thermomajorisation]
    \label{first_majorizations}
		The tensor $T$ is a [bi-]thermal tensor, $\forall_q$ $T({\vb{\gamma}},\vb{q}) = {\vb{\gamma}}$ [and $T({\vb{q},\vb{\gamma}}) = {\vb{\gamma}}$], if and only if for any input distributions $\vb{p},\vb{q}$ the output is a probability distribution satisfying thermomajorization relations $\vb{p} \succ_{(\beta)} T(\vb{p},\vb{q})$ [and $\vb{q} \succ_{(\beta)} T(\vb{p},\vb{q})$].
	\end{thm}
	
	\begin{proof}
		Let $T$ be a thermal tensor then for any distribution $\vb{q}$, $T(\cdot~,\vb{q})$ is a thermal operation, thus the implication in one direction follows.
		
		To prove the implication in the opposite direction, we note that any layer of a tensor $B^{(j)} = T(\cdot,e_j)$ is a matrix such that $\vb{p} \succ_{(\beta)} B^{(j)}(\vb{p})$. Thus each layer $B^{(j)}$ preserves a Gibbs state. Moreover, its non-negativity and normalization can be proven analogically as in~\cite{Majorization_book} (theorem A.4).
		Since each layer $B^{(j)}$ is a thermal operation, $T$ must be a thermal tensor.
		
		In the case of bithermal tensors, one only needs to repeat the above arguments replacing the first with the second input state.
	\end{proof}
	
	As one can see, the problem of majorization between inputs and outputs is very similar to the scenario with thermal operations.
    For the thermal tensors, one can also provide an analogue of Theorem \ref{thm_HLPgeneralisation}: 

	\begin{thm}
		\label{thm_HLPgeneralisation_tri}
		There exists a thermal tensor $T$, $T(\vb{\gamma},\vb{q}) = \vb{\gamma}$, mapping $(\vb{p}, \vb{q})$ to $\vb{r}$ if and only if $\vb{p} \succ_{\beta} \vb{r}$.
	\end{thm}

    \begin{proof}
    The implication in the first direction follows directly from the above Theorem \ref{first_majorizations}. The implication in the other follows from the fact that if $\vb{p} \succ_{\beta} \vb{r}$, then, by the Theorem \ref{thm_HLPgeneralisation}, there exists a Gibbs-preserving stochastic matrix $\Lambda$ such that $\Lambda(\vb{p}) = \vb{r}$. Thus we may construct a thermal tensor $T$ with all layers equal to $\Lambda$: $T_{ikl} = \Lambda_{ik}$ to obtain: $T(\vb{p},\vb{q})_i = \sum_{kl}\Lambda_{ik} p_k q_l = \sum_k \Lambda_{ik} p_k = r_i$, which ends the proof. 
    \end{proof}

    On the other hand, the problem of finding a \textit{bithermal} tensor, that may transform a pair of distributions into a desired one is much more complicated. The majorization of possible output by inputs is only a necessary condition for possible mapping.
	
	\begin{lem}
		Let $\vb{p},\vb{q},\vb{r}$ be three probability distributions, then there exists a bithermal tensor $T$ mapping $\vb{p},\vb{q}$ into $\vb{r}$: $\vb{r} = T(\vb{p},\vb{q})$ if and only if $\vb{r}$ is majorized by a convex combination of majorization curves of $T^{(i)}(\vb{p},\vb{q})$ for all extremal bithermal tensors $T^{(i)}$.
	\end{lem}
	
	\begin{proof}
		Suppose some combination of majorization curves for $T^{(i)}(\vb{p},\vb{q})$ majorizes $\vb{r}$. In that case, one can perform a convex combination with the same weights between appropriate extremal tensors $T^{(i)}$, obtaining tensor $\tilde{T}$ such that $\tilde{T}(\vb{p},\vb{q}) \succ_{(\beta)} \vb{r}$. Thus additional thermal post-processing is sufficient to construct the desired tensor.
		Note that such post-processing can be included in bithermal tensor without changing its defining properties.
		
		On the other hand, if there is no convex combination of majorization curves that majorizes $\vb{r}$, there is also no combination of extremal tensors which would lead $\vb{p},\vb{q}$ into state distribution majorizing $\vb{r}$ - thus there is no desired tensor. 
	\end{proof}
	
	From the above lemma 
    , we can deduce the following observations, which simplify the discrimination of possible transforming bithermal tensors.
	
	\begin{obs}
		Let $l$ be a curve obtained by taking a maximum over majorization curves for $T^{(i)}(\vb{p},\vb{q})$ at each point, with $T^{(i)}$ being extremal bithermal tensors. If $l$ does not majorize the curve corresponding to $\vb{r}$, then there exists no  bithermal  tensor $T$, such that $T(\vb{p},\vb{q}) = \vb{r}$.
	\end{obs}
	
	\begin{obs}\label{obs_extremal}
		(Numerical) For tristochastic tensors, it is sufficient to consider only extremal permutation tensors. 
	\end{obs}
	
	Note that Observation \ref{obs_extremal} is only conjectured, based on extensive numerical experiments. However, the heuristic argument in its favour can be given as follows. The only tensors that map all pairs of basis vectors $e_i, e_j$ into basis vectors are permutation tensors. Thus for any pair of input basis vectors its image under any permutation tensor majorizes its image under any tristochastic tensor. Since all finite probability distributions are convex combinations of basis vectors, we expect that the dominance of permutation tensors holds in the general case.  
	
	\begin{figure}[h]
		\label{fig:architecture}
		\centering
        \includegraphics[width=1.1\columnwidth]{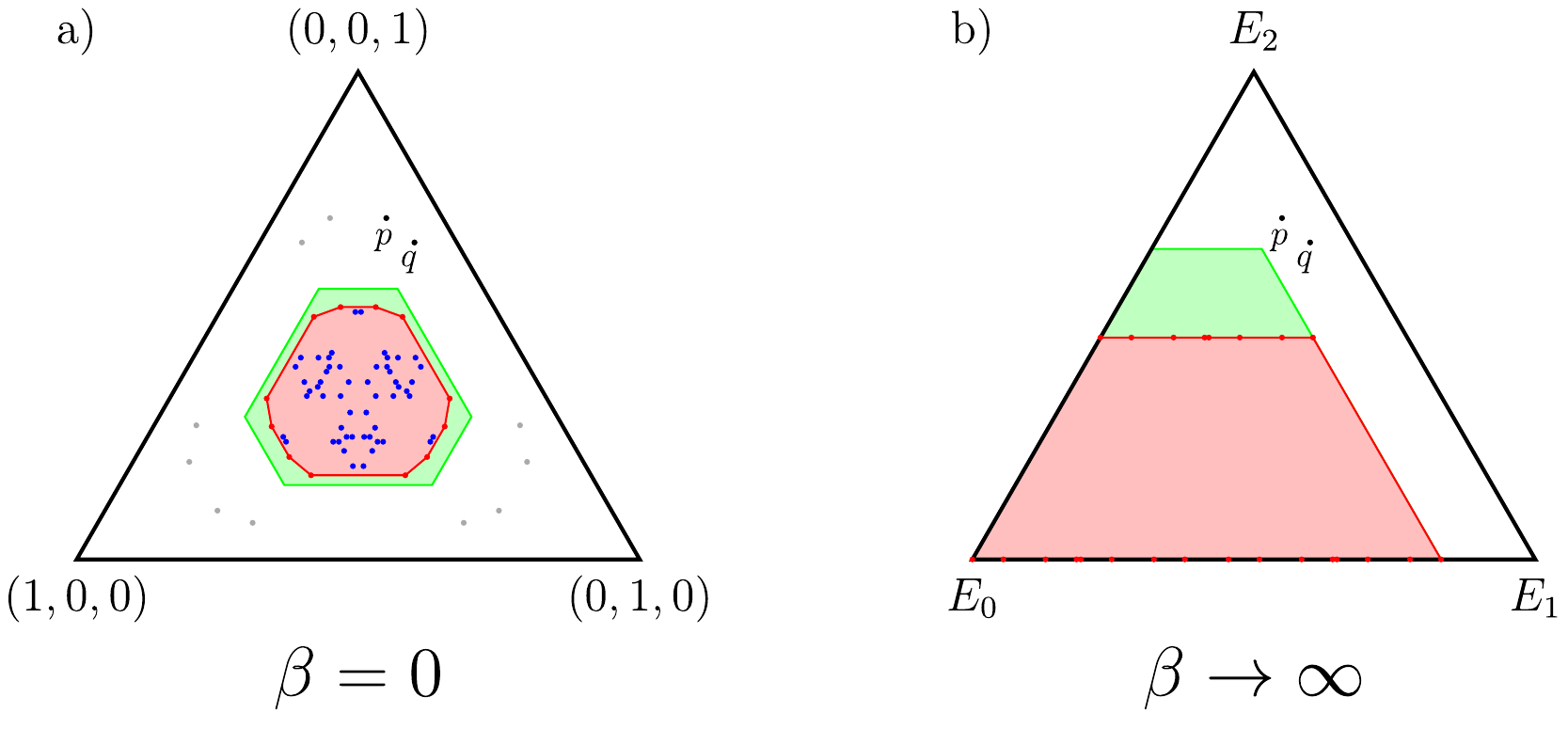}
		\caption{The allowed region for ourput state $\vb{r} = T(\vb{p},\vb{q})$ in local dimention $d = 3$. The two input states are denoted by black dots. All images of $T(\vb{p},\vb{q})$ are convex combinations of extremal images, denoted by coloured dots, thus they lie inside the red border. The green region is the upper bound on states accessible via bithermal operation according to Theorem \ref{thm_allowed}. (a) In infinite temperature $\beta = 0$ images via permutation tensors and other extremal tensors are represented by red and blue dots, respectively. The grey dots correspond to the images of $\vb{p}$ and $\vb{q}$ under permutation pre-processing, which does not affect the allowed region. (b) In zero temperature, $\beta \rightarrow \infty$ images via extremal tensors are denoted in red.}
	\end{figure}
	
	Fortunately, there also exists a stronger necessary condition for the existence of bithermal tensor $ T(\vb{p},\vb{q})\ = \vb{r}$ for given distributions $\vb{p},\vb{q},\vb{r}$.
	To present it let us introduce the following notion.
	
	For any distribution $\vb{p}$ let us define its projection $\tilde{\vb{p}}$ onto the boundary of the probability simplex, which we will refer to as boundary distribution of $\vb{p}$,
	\begin{equation}
		\tilde{\vb{p}} = \frac{1}{1-w_p}\vb{p} - \frac{w_p}{1-w_p} {\gamma},
	\end{equation}
	with $w_p = \min_i \frac{p_i}{{\gamma}_i}$; note that $\tilde{p}$ has at least one zero entry. 
	
	\begin{thm}[Enhanced thermomajorization for bithermal tensors]
		\label{thm_allowed}
		Let $\vb{p},\vb{q},\vb{r}$ be three probability distributions, and $\tilde{\vb{p}}$, $\tilde{\vb{q}}$ boundary distributions of $\vb{p},\vb{q}$. 
		Moreover, let us define
		\begin{equation}
			\overline{\vb{r}} =  \frac{1}{(1 - w_p)(1 - w_q)} \vb{r} -   \frac{1 - (1 - w_p)(1 - w_q)}{(1 - w_p)(1 - w_q)} \gamma
		\end{equation}
		Then the necessary condition for $\vb{r}$ to be achivable by a bithermal tensor from $\vb{p},\vb{q}$ is not only $\vb{p} \succ_{\beta} \vb{r}$ and $\vb{q} \succ_{\beta} \vb{r}$ but also $\tilde{\vb{p}} \succ_{\beta} \overline{\vb{r}}$ and $\tilde{\vb{q}} \succ_{\beta} \overline{\vb{r}}$.
	\end{thm}
	
	\begin{proof}
		Assume that there exists some bithermal tensor such that $T(\vb{p},\vb{q}) = \vb{r}$. Then by the definition of boundary distributions:
		\begin{equation}
			\begin{aligned}
				&T(\vb{p},\vb{q}) = T\left((1 - w_p)\tilde{\vb{p}} + w_p {\gamma},(1 - w_q)\tilde{\vb{q}} + w_q {\gamma}\right) \\
				& ~~~~=  (1 - w_p)(1 - w_q)T(\tilde{\vb{p}},\tilde{\vb{q}}) + \left[1 - (1 - w_p)(1 - w_q) \right]{\gamma}
			\end{aligned}
		\end{equation}
		Therefore the image of $\tilde{\vb{p}},\tilde{\vb{q}}$ is exactly $\overline{\vb{r}}$, hence it must be a probability distribution and  $\tilde{\vb{p}} \succ_{\beta} \overline{\vb{r}}$, $\tilde{\vb{q}} \succ_{\beta} \overline{\vb{r}}$.
	\end{proof}
	
	Thus we effectively improve the requirements proceeding from thermomajorization relations, leveraging the fact, that the allowed operations are bithermal.

    \section{Nonlocal correlations under LTOCC}
    \label{sec:nonloc_corr}



    In the following two sections, we ask a natural questions concerning relation between LTOCC and nonlocality. First, we demonstrate that starting from initially uncorrelated energy-incoherent systems in product state one may achieve significant level of correlations between the systems, in particular demosntrating the power of LTOCC with memory. Next, we move to present two aspects of LTOCC in the CHSH scenario as an example of test of Bell nonlocality. We show that on one hand it is generically impossible to detect breaking of CHSH inequality in a single-copy scenario due to restrictions imposed by local thermal operations; on the other hand, when multiple copies are present, the upper bound achievable by LTOCC is below Tsirelson bound, and tends to it in the limit of infinite copies, leading to possiblity of detecting athermal resource usage in finite-copy setting.
    \medskip

    \subsection{Generating correlations from product input states}
    \label{sec:sltocc_sepstates}

    First, we will focus on the ability of one-round and two-round LTOCC between two parties to demonstrate its ability to generate nonlocal correlations when acting on product energy-incoherent states.
	
	First and foremost we stress that neither LTOCC+M nor $\mathcal{S}$LTOCC are (semilocal) thermal operations, contrary to the protocols without memory. Both of these protocols can induce correlations even on Gibbs ensemble $\gamma = \gamma^{(A)} \otimes\gamma^{(B)}$, which can be most clearly seen in the case of extremal operations in zero or infinite temperatures.
    However, before discussing these extreme scenarios let us first consider LTOCC without memory in a wider range of temperatures.

    \begin{obs}
    Let $T_{jkl}$ be a thermal tensor, $\sum_{l}T_{jkl}\gamma_l = \gamma_j$, having an identity [or permutation matrix in infinite temperature case] as at least one layer $T_{jkl_0} = \delta_{j,k}$ [or $T_{jkl_0} = \delta_{j,\sigma(k)}$ for some permutation $\sigma$]. Then the pair of states $\vb{p}, \vb{q}$, second of which is a sharp state $q_l = \delta_{l,l_0}$ can be mapped into the maximally correlated state in one-round LTOCC.
    \end{obs}

    \begin{proof}
    Let us take $\Lambda_{ik} = \delta_{i,k}$ (or $\Lambda_{ik} = \delta_{i,\sigma(k)}$ in infinite temperature) and consider one-round LTOCC $M_{ij,kl} =\Lambda_{ik}T_{jkl}$ acting on uncorelaterd distribution $\vb{p}\otimes\vb{q}$. Then the output distribution $\pi_{ij}$ is given by
    \begin{equation*}
    \begin{aligned}
    \pi_{ij} & = \sum_{kl} M_{ij,kl} p_k q_l = \sum_{kl} \Lambda_{ik}T_{jkl} p_k \delta_{l,l_0} = \sum_{k} \delta_{i,k} \delta_{j,k} p_k \\
    & = \delta_{i,j}p_i~, \
    \end{aligned}
    \end{equation*}
    and analogously for the extended case in the infinite temperature.
	Since output distribution is diagonal, by measuring one output, one knows about the result of the second output as well. Formally, while calculating conditional entropy between outputs one obtains
		\begin{equation}
        \begin{aligned}
		& H(A|B) = H(B|A) = -\sum_{ij} \pi_{ij} \ln\left(\frac{\pi_{ij}}{p_i}\right) \\
        &= -\sum_{ij} \delta_{i,j} p_{i} \ln\left(\frac{p_i}{p_i}\right) = -\sum_{i}p_i \ln(1) = 0.
        \end{aligned}
		\end{equation}
    \end{proof}

   Notice that the requirement for thermal tensor demanded above can be guaranteed not only in zero or infinite temperature but is satisfied, for example, by construction~\eqref{example_n}.
   Additional discussion and examples of correlating properties of LTOCC in finite temperature are presented in Appendix~\ref{App:Geo_evol}.

   The correlation-generating power of LTOCC increases dramatically if we allow retention of memory between rounds.
   Suppose we allow ourselves to use one round of parallel LTOCC, especially $\mathcal{S}$LTOCC, which is a subset of two rounds of LTOCC with memory. In that case the above observation can be extended to encompass non-sharp states, and even Gibbs state.

	\begin{obs}
		Consider a $\mathcal{S}LTOCC$ operation consisting of two copies of the same extremal (bi)thermal tensor  $M_{ij,kl} = T_{ikl}T_{jkl}$, which entries are either $0$ or $1$. Then for any product input state $\vb{p}\otimes \vb{q}$ the output state $\vb{\pi} = M(\vb{p}\otimes \vb{q})$ is maximally correlated.
	\end{obs}
	
	\begin{proof}
		Let us denote $T(\vb{p},\vb{q}) = \vb{r}$, then, since in each column ($\forall_{j,k}$) of $T_{ijk}$ there is only a single entry equal to $1$ and the rest of them are $0$, we get
		\begin{equation}
        \begin{aligned}
		\pi_{ij} & = \sum_{kl} M_{ij,lk}p_kq_l = \sum_{kl}T_{ikl}T_{jkl}p_k q_l = \sum_{kl}\delta_{i,j}T_{ijk}p_k q_l \\
        & = \delta_{ij} r_i~.
        \end{aligned}
		\end{equation}
		Thus the output distribution is once again diagonal.
	\end{proof}

    The above requirements for thermal tensors are satisfied by extremal bithermal tensors in the limit of zero temperature $\beta\to \infty$ and a subset of extremal tensors in infinite temperature $\beta \to 0$ known as permutation tensors~\cite{Birkhoff_tensors}. The above Observation points to the fact that LTOCC with long-term memory is a significantly more powerful framework than LTOCC without memory. 


    \subsection{Thermally restricted Bell nonlocality}
    \label{sec:bell_LTOCC}
    
    Finally, let us consider a standard bipartite Bell scenario as the setting for quantum nonlocal correlations. Alice and Bob, each with a set of local measurement settings, are to determine whether a shared state $\ket{\Psi_{AB}}$ can be modelled by local hidden variable model or not. \rd{In order to keep the discussion specific, we} will focus on 
    the CHSH scenario, given in terms of correlators as \cite{CHSH1969, Cirelson1980}
    \begin{widetext}
        \begin{equation} \label{eq:CHSH_Tsirelson}
            \abs{E(a_0,b_0) + E(a_0,b_1) - E(a_1,b_0) + E(a_1,b_1)} \equiv S \leq \begin{cases}
                2 & \text{Classical (CHSH),} \\
                2\sqrt{2} & \text{Quantum (Tsirelson).}
            \end{cases}
        \end{equation}
    \end{widetext}
    where $E(a_i, b_j)$ are expectation values corresponding to measuremnt settings $a_i, b_j$ for Alice and Bob, respectively. The measurement outcomes are further restricted to $\pm 1$.
    
    Both classical and quantum bounds hold for arbitrary dimension $d$, but it is instructive to note that for a two-qubit canonical EPR state $\ket{\Psi_{+2}} = (\ket{00}+\ket{11})/\sqrt{2}$ the Tsirelson bound is saturated by taking setting $E(a_i,b_j) = \ev{A_iB_j}{\Psi_+}$ with $A_0 = Z, A_1 = X, B_0 = (Z+X)/\sqrt{2}, B_1 = (Z-X)/\sqrt{2}$. Notice that all but $A_0$ observable require incoherent measurements to be implemented. This leads to the following general observation \rd{applicable to any Bell inequality}.
    \begin{obs}
        Consider a pair of systems of local dimension $d$ with non-degenerate local Hamiltonians, capable of implementing LTOCC. In this setting it is impossible to break classical bound of inequality \eqref{eq:CHSH_Tsirelson}.
    \end{obs}
    \begin{proof}
        The set of free measurements of LTOCC consists of incoherent measurements, which are compatible. At the same time Bell nonlocality can be detected if and only if the measurements are incompatible \cite{Wolf2009}, leading to contradiction.
    \end{proof}
    Thus, in the single-copy energy-nondegnerate scenario it is impossible to detect breaking of 
    \rd{any Bell nonlocality when restricted to LTOCC.}

    Let us however consider, again for illustratory purpose, a scenario where Alice and Bob share a two-copy state $\ket{\Psi_{+2}}^{\otimes 2} = \ket{\Psi_{+4}}$. Since both subsystems are identical, it is natural to take $H^{\oplus 2} \equiv H_{AB}\otimes\mathbb{I} + \mathbb{I}\otimes H_{AB}$, with tensor product between copies. Let us set 
    \begin{equation*}
    A'_i = \mathbb{I}_1\oplus A_i \oplus \mathbb{I}_1  = \mqty(1 & 0 &0\\ 0 & A_i & 0 \\  0 & 0 & 1) 
    \end{equation*}
     as observables implementable by two-outcome measurement local to Alice, and similarly $B'_j$ for Bob. Direct calculation shows that
    \begin{equation}
        \abs{\ev{
        A'_0B'_0 +
        A'_0B'_1 -
        A'_1B'_0 +
        A'_1B'_1}{\Psi_{+4}}} = 1+\sqrt{2}.
    \end{equation}
    These observables can be realised by incoherent measurements due to the energy-degenerate subspace $\operatorname{span}(\{\ket{01},\ket{10}\})$. Furthermore, post-selection with respect to this subspace saturates Tsirelson bound. It follows that
    \begin{equation}
        \abs{\ev{
        A'_0B'_0 +
        A'_0B'_1 -
        A'_1B'_0 +
        A'_1B'_1}{\Psi_{+4}}} \leq 1+\sqrt{2}.
    \end{equation}
    for any $A'_i$ and $B'_j$ realisable with incoherent measurement and local thermal operations. This leads to a general theorem.
    \begin{thm} \label{thm:LTO_bound_CHSH}
        Consider Alice and Bob sharing an $n$-copies of a $d$-dimensional maximally entangled state state $\ket{\Psi_{+d}}^{\otimes n} = \ket{\Psi_{+d^n}}\equiv\ket{\Psi}$, with non-degenerate Hamiltonian $H_{AB}$ for each copy, such that both parties have the same energy degenerated subspaces for $H_A^{\oplus n}$ and $H_B^{\oplus n}$. Then the expectation of CHSH scenario is bounded as
        \begin{equation}\label{eq:LTO_bound_CHSH}
            \abs{\ev{
        A_0B_0 + A_0B_1 - A_1B_0 + A_1B_1}{\Psi}}   \leq 2\frac{D_{\text{ndeg}} + \sqrt{2}D_{\text{deg}}}{d^n}
        \end{equation}
        with $D_{deg}$ being the sum of dimensions of energy-degenerated subspaces and $D_{ndeg}$ sum of dimensions of non-degenerated subspaces.
        In particular for Hamiltonian with equal energy spacing $D_{\text{ndeg}} = 2,\; D_{\text{deg}} = d^n - 2$ and for Hamiltonian with rationally-independend energy spacing $D_{\text{ndeg}} = d,\; D_{\text{deg}} = d^n - d$.
        
    \end{thm}
    \begin{proof}
        The $n$-copy non-degenerate Hamiltonian provides splitting of the local state space $\mathcal{H}_{d}^{\otimes n} = \mathcal{H}_{\text{ndeg}} \oplus \mathcal{H}_{\text{deg}}$ into a non-degenerate part $\mathcal{H}_{\text{ndeg}} \equiv \text{span}(\qty{\ket{\vb{b}_i}}_{l = 1}^{D_{\text{ndeg}}})$, where $\vb{b}_i$ are strings describing appropriate Hamiltonian eigenvectors, and degenerate spaces $\mathcal{H}_{\text{deg}} \equiv \bigoplus_k \mathcal{H}_{d_k}$ of dimension $d_k$ each.

        \begin{figure*}
            \centering
            \includegraphics[width=.9\linewidth]{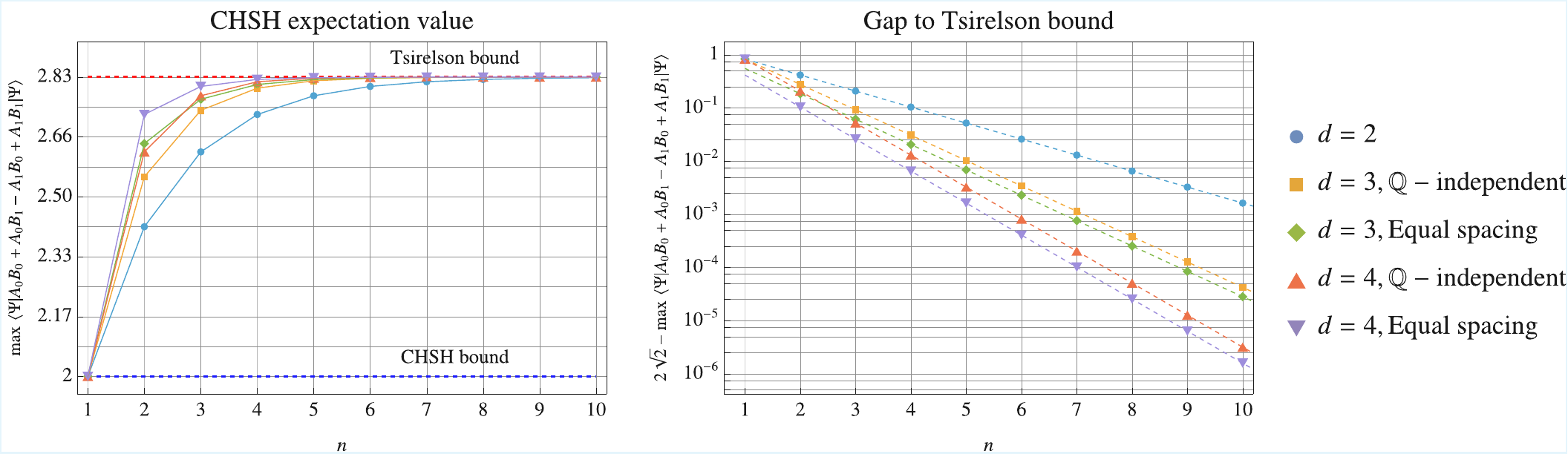}
            \caption{\bl{Comparison of upper bounds for expectation values in CHSH scenario \eqref{eq:thermal_CHSH_bounds} based on local hidden variable models (CHSH bound), local operations with classical communication (Tsirelson bounds) and LTOCC with $n$ copies of a joint quantum state $\ket{\Psi}\in\mathcal{H}_d^{\otimes 2}$. Local dimensions $d = 2,3,4$ are considered both with equal and rationally-independent energy spacing. Overall convergence rate of $\mathcal{O}(d^{-n})$ does not depend on the local energy level structure, with a multiplicative improvement for the equally spaced energy levels.}}
            \label{fig:Bell_plot}
        \end{figure*}
        
        The presented value is attainable by setting $E(a_i,b_j) = \ev{A_i B_j}{\Psi}$ with observables given by
        \begin{align} \label{eq:direct_sum_obs}
            A_i = \Pi_{\text{ndeg}}\oplus\qty(\bigoplus_k A_{i,k})
        \end{align}
        where $\Pi_{\text{ndeg}}$ is the projector onto the non-degenerate subsapce $\mathcal{H}_{\text{ndeg}}$ with $A_{i,k}$ saturating Tsirelson bound \eqref{eq:CHSH_Tsirelson} for $\ket{\Psi_{+d_k}}$ state; the operators $B_j$ are constructed by analogy. $A_i$ and $B_j$ can be realised with local thermal operations due to degeneracy of the $n$-copy Hamiltonian in each of the subspaces $\mathcal{H}_{d_k}$. The optimality of the bound is shown by considering postselection with respect to the degenerate subspaces $\mathcal{H}_{\text{deg}}$, where the same observables saturate quantum bound \eqref{eq:CHSH_Tsirelson}.

        Finally we show the values of $D_{\text{ndeg}}$ and $D_{\text{deg}}$. When the Hamiltonian has equal spacing, then the only nondegenerated subpsaces are $\text{span}(|1\ra^{\otimes n})$ and  $\text{span}(|d\ra^{\otimes n})$. For any other state one can find at lest one more, with the same energy by lovering the level in one system, and raising the level in other one.
        On the other hand, when the energy spacing is rationally independend, then the degenerated subspaces are
        $\qty{\operatorname{span}(\ket{i\hdots i})}_{i=1}^d$ and for any other state one can find state with the same energy by permuting systems.
    \end{proof}

    \begin{cor}
        The same bound \eqref{eq:LTO_bound_CHSH} holds for arbitrary $n$-copy entangled state $\rho^{\otimes n}$.
    \end{cor}
    \begin{proof}
        The proof of Theorem \ref{thm:LTO_bound_CHSH} rests on the fact that $\ket{\Psi_{+d^n}}$ is proportional to the maximally entangled state when restricted to any of the degenerate subspaces $\mathcal{H}_{\text{deg}}$. The same is not the case for arbitrary entangled states, and thus achievable value with operators of the form \eqref{eq:direct_sum_obs}, enforced by the structure of LTO, is strictly lower than \eqref{eq:LTO_bound_CHSH}.
    \end{proof}


    Interestingly, the statement extends to local thermal operations aided by both communication and memory.


    \begin{restatable}{thm}{CHSHLTOCCbound}[Extension of Theorem \ref{thm:LTO_bound_CHSH}]\label{thm:LTOCC_ent_dec1}
        Bound \eqref{eq:LTO_bound_CHSH} on the CHSH operator holds also for local thermal operations extended by communication and memory [LTOCC(+M)]
    \end{restatable}

    \proof{For the clarity of the discussion the proof of this Theorem, toghether with necesery Lemmas, is presented in the Appendix \ref{App:coch_evol}.}

    Theorem \ref{thm:LTOCC_ent_dec1} shows three distinct critical values of the CHSH scenario:
    
    \begin{widetext}
        \begin{equation} \label{eq:thermal_CHSH_bounds}
            \ev{A_0B_0 + A_0B_1 - A_1B_0 + A_1B_1}{\Psi} \leq \begin{cases}
                2 & \text{Classical},\\
                2\frac{D_{\text{ndeg}} + \sqrt{2}D_{\text{deg}}}{d^n} & \text{Thermal (LTOCC+M)},\\
                2\sqrt{2} & \text{Quantum}.
            \end{cases}
        \end{equation}
    \end{widetext}

    This leads to a proposal of a new thermally-restricted CHSH (ThCHSH) scenario, in which Alice and Bob are tasked with maximising value of CHSH correlator using an $n$-copy entangled state. When no memory is involved, the central evaluating party, Charlie, can verify that Alice and Bob have access to non-thermal operations by verifying breaking of the bound \eqref{eq:LTO_bound_CHSH}. The scenario can even be extended by memory shared between Alice and Bob, which can be used to communicate results of intermediate measurements under LTOCC+M paradigm; inspection of memory state by Charlie can be used to verify that parties have not communicated the measurement settings. While being a proof-of-principle rather than loophole-free robust scenario, it provides a way to certify use of resources via Bell-like experiments.

    \bl{Figure \ref{fig:Bell_plot} demonstrates the generic properties of convergence of the value of three bounds presented in eq. \eqref{eq:thermal_CHSH_bounds}. Although specific values depend on the strucutre of energy gaps, showing higher violation when the spectrum is $\mathbb{Q}$-dependent, the general convergence behaviour is dependent primarily on the local dimensions, with behaviour of convergence to Tsirelson bound as $\mathcal{O}(d^{-n})$, with multiplicative constant depending explicitly on the number of energy-nondegenerate subspaces. }

    \rd{Finally, let us recall that we focused on CHSH mainly for the sake of concreteness. The arguments corresponding to energy-degenerate subspaces are, however, fully general, and thus we put forward the following general statement
    \begin{thm}
        Consider an arbitrary bipartite Bell scenario with an expectation value $S$ bound by classical and quantum bounds $\mathcal{B}_{\text{cl}}$ and $\mathcal{B}_{\text{q}}$. A realisation using $n$ copies of $d$-dimensional state under the same restrictions as in Theorem \ref{thm:LTO_bound_CHSH}
        \begin{equation}
            S \leq \frac{\mathcal{B}_{\text{cl}} D_{\text{ndeg}} + \mathcal{B}_{\text{q}}D_{\text{deg}}}{d^n}.
        \end{equation}
    \end{thm}}
	
	\section{Summary and Outlook}\label{Sec:out}

    In this work, we introduce a novel framework where parties operate under the distant laboratories paradigm, subject to thermodynamic constraints, coining the notion of local thermal operations and classical communication (LTOCC). We begin by presenting a hierarchy of classes of allowed operations, categorized by the level of control and the complexity of communication between the parties, and present the inclusion relations between them, including a special case of parallel LTOCC and its symmetric variant $\mathcal{S}$LTOCC. In particular, we show that LTOCC with shared randomness form a subset of semilocal thermal operations (SLTO), pointing to a possibility of realising finite-size thermal engines operating at Carnot efficiency, as described in \cite{Bera2021}, within the distant laboratory setting subject to classical communication restrictions. To do so, we provide new results for SLTO, presenting, among others, that the set of such operations is convex and closed under composition.
    
    Since thermal tensors, encoding conditional thermal operations when restricted to the energy-incoherent states, constitute basic building blocks of the LTOCC framework, the second part of the work is devoted to the study of their properties. In particular, a major part of our considerations is concerned with the bithermal tensors, which constitute a symmetric subset of thermal tensors.
    The structure of the full set of such tensors can be understood as a thermal extension of the Birkhoff polytope of multistochastic tensors, and as such is an interesting mathematical problem in and of itself.
    

    Finally, we consider the question of nonlocal correlations under LTOCC, natural in the context of the distant laboratories paradigm. We first show ability of LTOCC to generate correlations when input states are energy-incoherent, highlighting significant capability when extended by memory. Then we move on to the Bell nonlocality in the context of thermally restricted operations, for which we show a clear gap between thermally-restricted and maximal breaking of CHSH inequality. This serves to demonstrate possibility of certifying resource utility in the context of Bell nonlocality.

    The results presented here lay the foundation for a new framework that combines two established concepts: the distant laboratories setting and thermodynamic limitations on operations. Given the novelty of this framework, several key questions remain open for future research. The first is a deeper exploration of multi-round scenarios, which may allow for a broader range of accessible state transitions, and formulation of stronger relations between protocols with different numbers of rounds.
    \bl{Next topic concerns characterization of LTOCC$+$M operations, which may distort Gibbs state. Contrary to LTOCC without memory, these protocols cannot be implemented without persistent classical registers, raising the question about the interplay between Landauer cost of erasing classical memory and athermality of quantum system created by corresponding operations.}
    
    Another important question is the introduction of coherences into the framework, with first steps presented in an example of detection gap for CHSH inequality. This is in general a challenging problem, as coherent thermal operations are difficult to describe in closed form even for a single qubit~\cite{Lostaglio2015QCoherence, Korzekwa2017}. The distinction between thermal operations and Gibbs-preserving operations, which emerge on the quantum level, further complicates this analysis. However, taking into account the conjecture that LTOCC with shared randomness and SLTO are equal when restricted to energy-incoherent processes, extension to coherent states in future work may provide tools necessary to identify the gap between classical and genuinely quantum effects in thermodynamics of finite-size systems interacting with multiple thermal baths.

    Moreover, in the face of recent developments at the frontier of quantum thermodynamics and nonlocality researchers within the community have begun to investigate interplays between characteristic features of both subfields, such as possibility of generating entanglement using out-of-equilibrium states~\cite{oliveira2024entanglement}, detection of nonlocal properties of states using their thermodynamic properties~\cite{Oliveira2024heatEntanglementWit} or ever-growing body of literature concerned with steady-state entanglement from autonomous machines operating in thermal framework~\cite{BohrBrask2015, khandelwal2024maximal}. In all of the above cases, we believe that considering the problem within the LTOCC framework would further understanding of both thermodynamical and informational tradeoffs incurred throughout the implementation of the protocol.
	
	Finally, the introduced framework could be applied to quantum communication protocols, potentially revealing new thermodynamic trade-offs, such as those between system temperature and achievable key rates or thermodynamical entropy of states and informational entropy of shared randomness. Such analysis would allow for a better understanding of thermally optimal operations in the context of protocols such as entanglement generation, cooling or work extraction. However, this would require either shift to multicopy scenarios, similar to the one presented for the CHSH inequality breaking, or studying measurements that are not restricted to the energy eigenbasis, raising questions about the inherent thermodynamic costs. These trade-offs between information-theoretic gains, like key rates, and thermodynamic costs, such as work required, represent another avenue for future research.
	
	All of these questions, while significant, fall outside the scope of this initial presentation of the LTOCC framework. They should be viewed as a roadmap for future studies, aimed at expanding and refining this new theoretical approach, and some of them will be explored in a future manuscript \cite{Czartowski2025QuantumLTOCC}.

\begin{acknowledgments}

The authors thank A. de Oliveira Jr., Jeongrak Son, Kamil Korzekwa and Karol {\.Z}yczkowski for useful discussions and comments concerning the paper.  J. Cz. would like to thank Ray Ganardi for useful comments concerning Bell nonlocality in the presence of thermal operations. Furthermore, we thank A. de Oliveira Jr. for kindly providing us with the artistic impression in Fig.~\ref{fig:Famework_sheme}. 
    
RB acknowledges support by the National Science Centre, Poland, under the contract number 2021/03/Y/ST2/00193
within the QuantERA II Programme that has received funding from the European Union’s Horizon 2020 research and innovation programme under Grant Agreement No 101017733.
JCz is supported by the start-up grant of the
Nanyang Assistant Professorship at the Nanyang Technological University in Singapore, awarded to Nelly Ng.
\end{acknowledgments}

\appendix

	\section{Cooling maps}\label{App:cold_original}
	
	In this Appendix, we discuss the theoretical framework designed to grasp thermal operations in low temperatures: \textit{cooling maps}~\cite{Low_temerature_cooling_maps}. Although the original construction was concerned with the general scenario of mixed quantum states, we let ourselves limit the discussion to diagonal states, similar to the bulk of the paper.
	
	Consider a quantum system together with a thermal bath with a nonzero spectral gap. If the temperature tends to zero, the contribution to the bath's Gibbis state from excited states decreases exponentially. Thus for low enough temperatures, one may approximate
	\begin{equation}
		\vb{\gamma} \approx |E_0\ra\la E_0|~.
	\end{equation}
	Therefore the energy-preserving interaction with such a bath~\eqref{E_prev} can only deexcite the system of interest. Indeed, if we consider a diagonal input state with populations $\vb{p} = \operatorname{diag}_H(\rho)$, the action of the cooling map does not create coherences, as in the case of thermal operations, and modify the populations according to~\cite{Low_temerature_cooling_maps}
	\begin{equation}
		\vb{q} = P\, \vb{p}
	\end{equation}
	where $P$ is upper-triangular (UT) stochastic matrix: $P_{j,k} = 0$ if $j > k$.
	
	Similarly as for bistrochastic matrices and thermal operations, one can define a suitable version of majorization also for cooling maps. However this time the deexcitation properties of cooling maps naturally fix the order of basis states, so there are no rearrangements.
	
	\begin{defi}[\cite{Low_temerature_cooling_maps}]\label{def_TO_Majorisation}
		Given two $d$-dimensional probability distributions $\vb{p}$ and $\vb{q}$, we say that $\vb{p}$ \emph{UT-majorises} $\vb{q}$, and denote it by $\vb{p} \succ_{UT} \vb{q}$, if and only if the following condition holds:
		\begin{equation}
			\label{eq_T0_majorisation}
			\sum_{i=1}^k p_i\geq\sum_{i=1}^k q_i\quad \text{for all} \quad  k\in\{1\dots d\}.
		\end{equation}
	\end{defi}
	
	This definition lets us state the counterpart of Theorems \ref{thm_HLP} and \ref{thm_HLPgeneralisation} in zero temperature:
	
	\begin{thm}[Theorem~1 of Ref.~\cite{Low_temerature_cooling_maps}]
		There exists a cooling map, corresponding to upper triangular matrix $P$, $P \vb{\gamma}=\vb{\gamma}$, mapping $\vb{p}$ to $\vb{q}$ if and only if $\vb{p} \succ_{UT} \vb{q}$.
	\end{thm}

	\section{Alternative derivation of cooling maps}\label{App:cold}
	
	In this Appendix, we present an alternative derivation of cooling maps in the case of energy-incoherent states as a thermal operation in the limit of zero temperature $\beta \to \infty$.
	
	\begin{thm}
		\label{zero_temp_bound}
		In the zero temperature limit $\beta\to\infty$, assuming nondegenerated Hamiltonian spectrum, the stochastic matrix $\Lambda$ corresponds to thermal operation if and only if $\Lambda_{i,j} = 0$ for $i > j$.  
		Furthermore, the extremal operations correspond to matrices with $0$ and $1$ entries, restricted only by the condition $\Lambda_{ij} = 0$ if $i>j$.
	\end{thm}
	
	\begin{proof}
		We start the proof by discussing the necessity of the above condition, to later focus on its sufficiency.
		Let us divide the populations of the Gibbs ensemble into two sets. The first are those corresponding to energies smaller of equal $E_m$ and the other correspond to larger energies.
		Since the operation $\Lambda$ is thermal, Gibbs ensemble is an invariant state, $\Lambda{\gamma} = {\gamma}$, so the flow of populations between those two sets, described by $\Gamma$ elements must balance, which can be shown explicitly,
		\begin{equation}
			\begin{aligned}
				& \sum_{i > m} \sum_{j \leq m} \Lambda_{ij} {\gamma}_j = \sum_{i>m}\qty(\sum_j \Lambda_{ij}{\gamma}_j - \sum_{j > m} \Lambda_{ij}{\gamma}_j) = \\
				& = \sum_{i}\qty({\gamma}_i - \sum_{j > m} \Lambda_{ij}{\gamma}_j)  - \sum_{i \leq m}\qty({\gamma}_i - \sum_{j > m} \Lambda_{ij}{\gamma}_j) = \\
                & 1 - \sum_{j >m}{\gamma}_j - \sum_{i \leq m} {\gamma}_i + \sum_{i \leq m} \sum_{j < m} \Lambda_{ij} {\gamma}_j  \\
				& =\sum_{i \leq m} \sum_{j < m} \Lambda_{ij} {\gamma}_j~.
			\end{aligned}
		\end{equation}
		To obtain the bounds on $\Lambda_{ij}$ elements, we may extrapolate the populations of Gibbs states using two inequalities: ${\gamma}_i \geq {\gamma}_m = e^{- \beta E_m}$ for $i \leq m$, and ${\gamma}_i \leq {\gamma}_{m+1} = e^{- \beta E_{m+1}}$ for $i >m$. Thus the above equality can be changed into
		\begin{equation}
			\begin{aligned}
				\sum_{i > m} \sum_{j \leq m} \Lambda_{ij} e^{-\beta E_{m}} &\leq \sum_{i \leq m} \sum_{j > m} \Lambda_{ij} e^{-\beta E_{m+1}} \\
				\sum_{i > m} \sum_{j \leq m} \Lambda_{ij}  &\leq \qty(\sum_{i \leq m} \sum_{j > m} \Lambda_{ij}) e^{-\beta (E_{m+1} - E_m)}. \\
			\end{aligned} 
		\end{equation}
		Since the right-hand side of the equation tends to zero and all terms on the left-hand side are positive we obtain the desired result.

		It turns out, that, the condition from this Lemma is also sufficient. This can be seen by considering the channel, which corresponds to a detailed balance of Gibbs state, ie $\Lambda_{ij}{\gamma}_j = \Lambda_{ji}{\gamma}_i$, without any summation involved. Each such stochastic channel is by its construction also thermal. However, in the limit of zero temperature $\beta = \infty$, the condition of detailed balance $\Lambda_{ij} =  \Lambda_{ji} e^{-\beta (E_i - E_j)}$, is always satisfied given $\Lambda_{ij} = 0$ if $i > j$. Hence the freedom of all nonvanishing $\Lambda$ elements is restricted only by stochasticity.
		
		Because the condition $\Lambda_{ij} = 0$ if $i > j$ does not couple the elements of $\Lambda_{ij}$, we can treat each column of $\Lambda$ separately as a probability vector which immediately leads to the extremal points.
	\end{proof}

    \section{Detailed discussion of one- and two-round LTOCC}\label{App:Geo_evol}

    In this Appendix, we will consider the evolution of states under LTOCC for one- and two-round protocols. In what follows, when it is important, we will consider Alice and Bob having at their disposal systems of equal dimension, $d_A = d_B = d$ and the local Hamiltonians $H^{(A)} = H^{(B)} = H$, but different inverse temperatures $\beta^{(A)}$ and $\beta^{(B)}$. We will assume that they initialize their systems independently, thus starting in a state $\vb{r}^{(0)} = \vb{p}^{(0)}\otimes\vb{q}^{(0)}$, where the superscript of the global state will refer to the number of rounds. For reference we will be using $d=3$, so that we can represent marginal states of Alice and Bob on their respective 2-dimensional probability simplices, as shown in Fig.~\ref{fig:geo_0_steps}. This appendix will be using slightly different notation in comparison to the remainder of the manuscript.

    In the first round of the protocol (thus realising LTOCC$_1$), we consider Alice measuring her state and forwarding the obtained information to Bob, In turn, Bob can act on his system with a collection of thermal maps $\Xi^{(1,i)}$ conditioned on the information obtained from Alice. Finally, Alice post-processes her system using thermal operation $\Lambda^{(1)}$. In order to properly set the stage, we will write down the first round of the protocol in fully quantum form, starting from a state $\operatorname{diag}(\vb{p}^{(0)}) \otimes \operatorname{diag}(\vb{q}^{(0)}) \equiv \rho\otimes\sigma$:
    {
    \allowdisplaybreaks
    \begin{subequations}
        \begin{align}
            &\text{Measurement:} \; \rho\otimes\sigma  \mapsto \qty(\sum_i  \Tr(K_i \rho K_i^\dagger) K_iK_i^\dagger)\otimes \sigma \\ 
            & \hspace{2 cm} = \sum_i p_i \op{E_i}\otimes\sigma\otimes\underbrace{\op{i}_{m}}_{\text{Memory}}\\
            & \text{Conditioning:} \; 
             \mapsto
            \sum_i p_i \op{E_i}\otimes\qty(\sum_j \widetilde{K}_{ij}\sigma\widetilde{K}_{ij}^\dagger)\otimes\op{i}_{m} \\
            & \hspace{2 cm} \equiv \sum_i p_i \op{E_i}\otimes\Xi^{(1,i)}\qty(\sigma)\otimes\op{i}_{m} \\
            &\text{Postprocessing:} \; \mapsto \sum_i p_i \Lambda^{(1)}(\op{E_i})\otimes\Xi^{(1,i)}\qty(\sigma)\otimes\op{i}_{m} \\
            &\text{Memory erasure:} \; \mapsto \sum_i p_i \Lambda^{(1)}(\op{E_i})\otimes\Xi^{(1,i)}\qty(\sigma)
        \end{align}
    \end{subequations}
    }
    where we set measurement operators $K_i = \qty(\ket{E_i}\otimes\ket{i}_m)\bra{E_i}$ to reflect the energy-incoherent projective measurement. The Krauss operators $\tilde{K}_{ij}$ correspond to Bob's thermal operations $\Xi^{(1,i)}$. 

\begin{figure}[b]
    \includegraphics[width=1\linewidth]{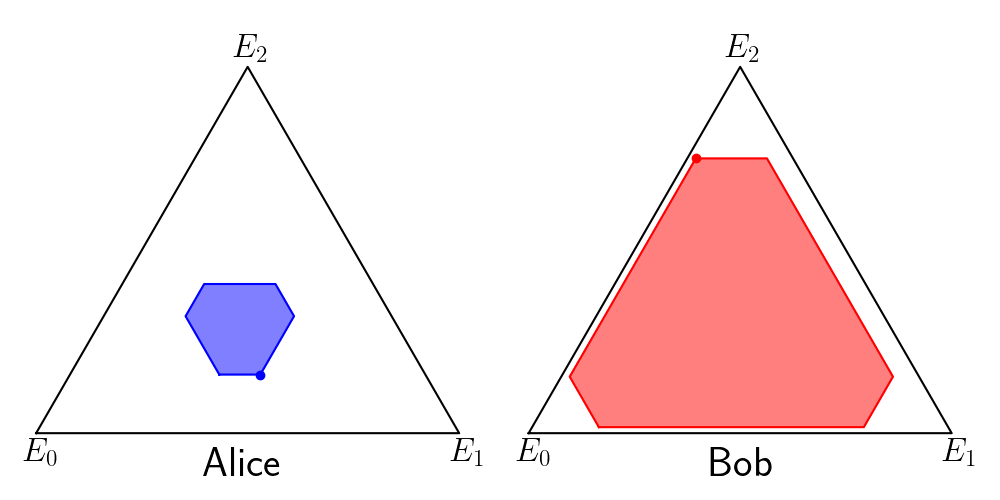}
        \caption{Exemplary starting point $\vb{r}^{(0)} = \vb{p}^{(0)}\otimes \vb{q}^{(0)}$ with Alice's state $\vb{p}^{(0)} = (0.39, 0.45, 0.16)$ (blue point) and Bob's state $\vb{q}^{(0)} = (0.23,0.02,0.75)$ (red point), subject to temperatures $\beta^{(A)} = 0.1,\,\beta^{(B)} = 0.2$ respectively and energies of both subsystems $\vb{E} = (0,1,2)$. The initial state of the system is given as a product, which justifies the representation in terms of marginals. The solid line represents the boundaries of the future thermal cone under standard local thermal operations. 
        }
        \label{fig:geo_0_steps}
    \end{figure}

    There are two things to note from this sequence of transformations. First, is that the sequence does not require $\sigma$ or $\rho$ to be diagonal on the outset, and hence the entirety of considerations applies also to protocols with coherent inputs. This leads to perhaps a more interesting observation, that after one round of such protocol Alice's reduction is fully classical, but Bob's reduction may retain coherence, thus giving him access to an in-principle coherent quantum ensemble of states. We will expand upon this notion in a future manuscript \cite{Czartowski2025QuantumLTOCC}.
    
    With the above general remarks in place, we may return to the energy-incoherent setting, which leads to a joint state in the form
    \begin{equation}
    \begin{aligned}
    \vb{r}_{\text{LTOCC}_1} & = \sum_{i=1}^d p^{(0)}_i \Lambda^{(1)}(\vb{s}(i))\otimes\Xi^{(1,i)}(\vb{q}^{(0)}) \\
    & = \sum_{i=1}^d p_i \vb{p}^{(1,i)}\otimes\vb{q}^{(1,i)}
    \end{aligned}
    \end{equation}
    where $\vb{s}(i)$ is a sharp state with $s_j(i) = \delta_{ij}$.
    Essentially, at this stage Alice and Bob have access to an ensemble of states. From the standpoint of Alice her reduced state is simply given by
    \begin{equation}
        \vb{p}^{(1)} = \sum_i p_i \vb{p}^{(1,i)} = \Lambda^{(1)}\qty(\vb{p}^{(0)}).
    \end{equation}
    This, in fact, corresponds to a situation when Alice sends out a bit without retaining it, and thus loses access to it directly after measurement. On the other hand, Bob has effectively an ensemble of states $\qty{p_i,\Xi^{(1,i)}(\vb{q})}$. However, due to erasure of memory directly after its use, the marginal state is given by
    \begin{equation}
        \vb{q}^{(1)} = \qty[\sum_i p_i \Xi^{(1,i)}]\qty(\vb{q}^{(0)}).
    \end{equation} 
    Despite the complexity of the set of thermal operations, as highlighted in \cite{Mazurek_2018}, due to thermomajorization relations it generates no more than $d!$ extreme image points for any input state $\vb{q}$ on Bob's side. Combined with $d$ levels of Alice's state, we find that the resulting joint state belongs to a convex hull of at most $(d!)^d$ extreme points. 
    Furthermore, it is easy to show that extremal points of Bob's marginal distribution are achieved by fixing his operation to be extremal and independent from Alice's measurement result, $\Xi^{(1,i)} = \Xi^{(1)}$. However, this would restrict the state to the product subset, $\vb{r}_{\text{LTOCC}_1}\in\Delta_{d-1}\otimes\Delta_{d-1}\subset\Delta_{d^2-1}$. Departure from product form can be quantified by using a standard measure of information in the joint state, which cannot be found in either of the systems locally -- mutual information, given by
    \begin{equation}
    \begin{aligned}
    I(A:B) &= \sum_{i,j=1}^d r_{ij}\Big(\log\big(r_{ij}\big) - \log\big(\sum_ir_{ij}\big) - \log\big(\sum_{j}r_{ij}\big)\Big) \\
    & \leq \log(d)
    \end{aligned}
    \end{equation}
    with the upper bound achieved when $r_{ij} = \delta_{ij} d^{-1}$. This figure of merit has been used in  fig.~\ref{fig:geo_1_step} to visualize the effect of the correlation generated by the use of classical communication. Additionally, we used the conditional entropies

    \begin{align}
        H(B|A) = - \sum_{i,j=1}^d r_{ij}\qty(\log(r_{ij}) - \log(\sum_jr_{ij})), \\
        H(A|B) = - \sum_{i,j=1}^d r_{ij}\qty(\log(r_{ij}) - \log(\sum_ir_{ij})), 
    \end{align}
    which are depicted in center and rightmost panels of fig. \ref{fig:geo_1_step}. Note that Alice's marginal has not been affected by the conditional operation, and thus $\sum_{ij}r_{ij}\log(\sum_jr_{ij}) = \text{const.}$ This makes mutual information and entropy conditioned on Bob's system equal up to an additive constant and a sign. By the same token, a similar connection exists between the entropy of the state $S(\vb{r})$ and the entropy conditioned on Alice's system. 

    First, we highlight that these entropic quantities are not functions of Alice's or Bob's reductions, but of the joint state. This is reflected in Fig~\ref{fig:geo_1_step} by plotting maximal numerically achieved values for mutual information, as opposed to minimal. This is due to the fact that the minimum of mutual information is trivially equal to zero, as each point of Bob's marginal simplex is achievable by setting his operation to be a constant thermal operation, $\Xi^{(1,i)} = \Xi^{(1)}$, and thus retaining product form of the joint state. On the similar grounds, we plot only minimal numerical values of both conditional entropies.

    \begin{figure*}[t]
        \centering
        \includegraphics[width=0.7\textwidth]{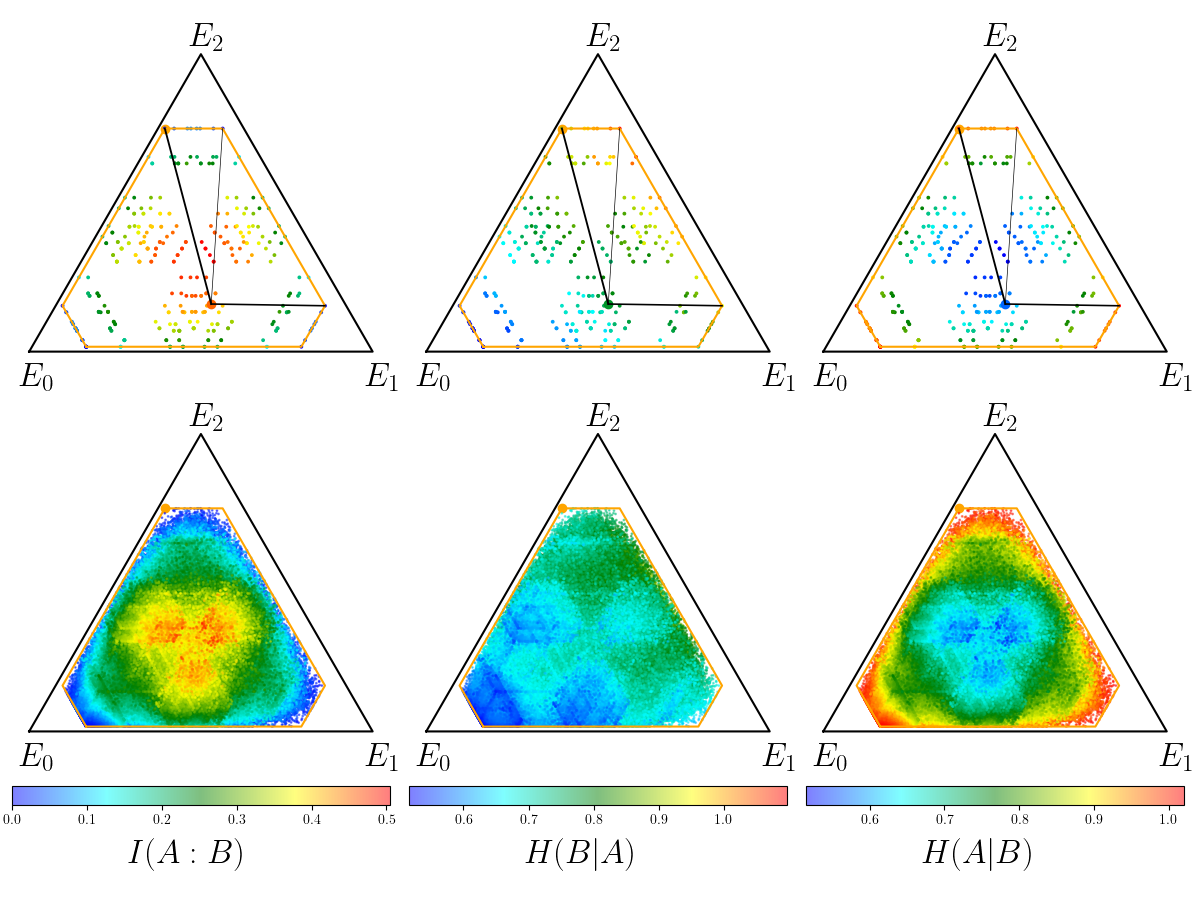}
        \caption{After a single round of LTOCC protocol we can investigate the informational quantities such as mutual information $I(A:B)$ (left panel) and conditional entropies $H(B|A)$ (center) and $H(A|B)$ (right). Top plots demonstrate achievable values when restricted to combinations of only extreme operations on Bob's side, while bottom plots demonstrate maximal (minimal) achievable values of mutual information (conditional entropy) in the situation when Bob has access to completely arbitrary local thermal operations on his side. 
        }
        \label{fig:geo_1_step}
    \end{figure*}

    At this stage, access to memory does not play a significant role, but for a full picture, before erasing the memory state, the overall state of the system is given as

    \begin{equation}
        \vb{r}^{(m)}_{\text{LTOCC}_1} = \sum_{i=1}^d p^{(0)}_i \underbrace{\vb{p}^{(1,i)}\otimes\vb{q}^{(1,i)}}_{\text{Alice and Bob}}\otimes\underbrace{\vb{s}(i)}_{\text{Memory}}
    \end{equation}

    The two-round protocol depends on whether we erase the memory after the first round, or not. Starting with the no-memory scenario, it is simplest to understand when we rewrite the probability $\vb{r}^{(1)}$ with respect to the sharp states of Bob, which is always possible,
    \begin{equation}
        \vb{r}_{\text{LTOCC}_1} = \sum_{j=1}^d q^{(1)}_j \tilde{\vb{p}}^{(1,j)} \otimes \vb{s}(j)
    \end{equation}
    where $\tilde{p}^{(1,j)} = r^{(1)}_{ij}/(\sum_i r^{(1)}_{ij})$ is just Alice's distribution conditional on Bob measuring $j$. At this point, Bob measures and sends the result to Alice, after which Alice implements a family of conditioned operations, while Bob post-processes his state to obtain a state

    \begin{equation}
        \vb{r}^{(m)}_{\text{LTOCC}_2} = \sum_{j=1}^d q^{(1)}_j \underbrace{\Xi^{(2,j)}\qty(\tilde{\vb{p}}^{(1,j)} )\otimes \Lambda^{(2)}\qty(\vb{s}(j))}_{\text{Alice and Bob}}\otimes \underbrace{\vb{s}(j)}_{\text{Memory}}.
    \end{equation}

    The rest of the protocol proceeds in the same manner as before. Note that it is easy to reason heuristically why Corollary~\ref{corr:SLTO_major} works -- conditioned processing of Bob's state is done using local thermal operations, and thus from his perspective, it obeys the thermomajorization $\prec_{\gamma^{(B)}}$. In the post-processing step Alice uses only local thermal operation, and thus clearly obeys thermomajorization $\prec_{\gamma^{(A)}}$ locally. Taking this together with the fact that baths are not correlated, it is reasonable to expect that the joint evolution will obey thermomajorization with respect to $\prec_{\gamma^{(A)}\otimes\gamma^{(B)}}$. This is only signalised here, and actual proofs will follow in Appendix~\ref{App:semilocal}.
    
    It is far more interesting to see what happens when Bob measures his own system and uses memory from the previous round. In this case, the post-measurement state is in fact given by
    \begin{equation}
        \sum_{i,j=1}^d p^{(0)}_i q^{(1,i)}_j\underbrace{\vb{p}^{(1,i)}\otimes\vb{s}(j)}_{\text{Alice and Bob}}\otimes\underbrace{\vb{s}(i)\otimes\vb{s}(j)}_{\text{Memory}}
    \end{equation}

    Using this and assuming that now Alice can condition the operations on both retained memories, while for Bob it is sufficient to condition it on the first memory, we obtain a post-operation state as
    \begin{widetext}        
    \begin{align}
        \vb{r}^{(2,m)}_{\text{LTOCC}_2+\text{M}} & = \sum_{i,j=1}^d p^{(0)}_i q^{(1,i)}_j\underbrace{\Xi^{(2,ij)}(\vb{p}^{(1,i)})\otimes\Xi^{(2,i)}(\vb{s}(j))}_{\text{Alice and Bob}}\otimes\underbrace{\vb{s}(i)\otimes\vb{s}(j)}_{\text{Memory}} \\
        & = \sum_{i,j=1}^d p^{(0)}_i q^{(1,i)}_j\Xi^{(2,ij)}(\Lambda^{(1)}(\vb{s}(i)))\otimes\Xi^{(2,i)}(\vb{s}(j))\otimes\vb{s}(i)\otimes\vb{s}(j).
    \end{align}
    \end{widetext}

    First immediate observation is that results with and without memory are not equal in principle, $\vb{r}^{(2,m)}_{\text{LTOCC}_2+\text{M}}\neq \vb{r}^{(m)}_{\text{LTOCC}_2}$. What is more telling about the power of storing correlations is that now both Alice and Bob are implementing thermal operations directly on sharp states, and can coordinate the operations using shared memory in a similar manner to shared randomness. One simple way to leverage this is to set Alice's first postprocessing to be trivial $\Xi^{(1)} = \mathbb{I}$, and then choose the remaining operations $\Xi^{(2,ij)}\otimes\Xi`^{(2,i)}$ so that they minimize energy of any input state $\vb{s}(i)\otimes\vb{s}(j)$, thus overcoming the limitations of thermal operations.

    \subsection{Remarks on relation with resource non-generating operations}

    As a final remark, we would like to discuss the relation of LTOCC without memory to the framework of fusing local resource theories as introduced in \cite{son2024robust}. There, the authors considered possible ways to fuse local resource theories into composite resource theories. In particular, they define a concept of minimal and maximal compositions. Given local sets of free states $\mathcal{F}^{(A)}$ and $\mathcal{F}^{(B)}$, they were defined respectively as
    \begin{widetext}
        \begin{align}
            \mathcal{F}^{(A)}\otimes_{\text{min}} \mathcal{F}^{(B)} & = \operatorname{conv}\qty{
            \rho^{(A)}\otimes\rho^{(B)}|
            \rho^{(A)}\in\mathcal{F}^{(A)},\,
            \rho^{(B)}\in\mathcal{F}^{(B)}}, \\
            \mathcal{F}^{(A)}\otimes_{\text{max}} \mathcal{F}^{(B)} & = \qty{\rho^{(AB)}|\Tr_A\rho^{(AB)}\in\mathcal{F}^{(B)},\,\Tr_B\rho^{(AB)}\in\mathcal{F}^{(A)}}.
        \end{align}
    \end{widetext}
    Next, let us assume that free states of Alice to be arbitrary density operators,  $\mathcal{F}^{(A)}_{\text{LO}} = \mathcal{D}\qty(\mathcal{H}_d)$, and Bob to be restricted to his local thermal equilibrium, $\mathcal{F}^{(B)}_{\text{TO}} = \qty{\gamma^{(B)}}$. Now, a simple inclusion holds
    \begin{align}
        \text{LTOCC}_{1,A\rightarrow B} \subset &\,\text{RNG}\qty(\mathcal{F}^{(A)}_{\text{LO}}\otimes_{\text{min}} \mathcal{F}^{(B)}_{\text{TO}}),\\
        \text{LTOCC}_{1,A\leftarrow B} \subset &\,\text{RNG}\qty(\mathcal{F}^{(A)}_{\text{LO}}\otimes_{\text{max}} \mathcal{F}^{(B)}_{\text{TO}}).
    \end{align}
    where RNG denotes resource non-generating operations.

    The inclusions are easily understood -- no matter what are the results of measurements on Alice's side, if Bob is restricted to thermal operations on his side, he will not change the local Gibbs state, and thus remain uncorrelated with Alice. On the other hand, if it is Bob who is performing the measurements, which are assumed to be energy-incoherent, he will preserve his marginal as the Gibbs state. 
    while generating possibly nontrivial correlations with Alice. In fact, these relations hold even when extended to a fully quantum regime.

    On similar grounds, one can argue that a general inclusion of the form
    \begin{equation}
        \text{LTOCC+R}\subset \,\text{RNG}\qty(\mathcal{F}^{(A)}_{\text{TO}}\otimes_{\text{min}} \mathcal{F}^{(B)}_{\text{TO}})
    \end{equation}
    holds. 

    \begin{widetext}
    \section{LTOCC without memory as a subset of semilocal thermal Operations.}\label{App:semilocal}

    In this Appendix, we discuss the properties of LTOCC protocols without memory, while processing general quantum states. In particular, we prove that all LTOCC protocols form a subset of semi-local thermal operations (SLTO) as defined in \cite{Bera2021}. 
    However, before moving on to the proof of the relation between LTOCC and SLTO we need to demonstrate two crucial properties of SLTO  -- namely, closedness and convexity of SLTO -- which will serve the role of lemmas for the final inclusion property. Additionally, unless stated explicitly, we work with thermal baths satisfying standard properties as put forward in \cite{horodecki2013fundamental, Bera2021} 

    \SLTOclosed*

    \begin{proof}
    Let $\qty(\mathcal{B}_1^{(A)},\,\mathcal{B}_1^{(B)})$, $\qty(\mathcal{B}_2^{(A)},\,\mathcal{B}_2^{(B)})$ be pairs of local thermal baths for Alice and Bob systems employed to implement $\Lambda_1^{(AB)}$ and $\Lambda_2^{(AB)}$, respectively. Furthermore let $H_1^{\mathcal{B}^{(A)}}$, $H_1^{\mathcal{B}^{(B)}}$, $H_2^{\mathcal{B}^{(A)}}$, $H_2^{\mathcal{B}^{(B)}}$ be corresponding Hamiltonians, and $U_1$, $U_2$ be unitary operations corresponding to the aforementioned semilocal thermal operations.

    To construct composition $\Lambda_1^{(AB)}\circ\Lambda_2^{(AB)}$ we define new baths by composing the baths from both processes: $\mathcal{B}^{(A)} := \mathcal{B}_1^{(A)}\otimes \mathcal{B}_2^{(A)}$, $\mathcal{B}^{(B)} := \mathcal{B}_1^{(B)}\otimes \mathcal{B}_2^{(B)}$ together with composite Hamiltonians $H^{\mathcal{B}^{(A)}} = H_1^{\mathcal{B}^{(A)}} + H_2^{\mathcal{B}^{(A)}}$, $H^{\mathcal{B}^{(B)}} = H_1^{\mathcal{B}^{(B)}} + H_2^{\mathcal{B}^{(B)}}$, and define joint unitary operation 
    \begin{equation*}
        U = (\id_{\mathcal{B}_1^{(A)} \mathcal{B}_1^{(B)}} \otimes U_2)(U_1 \otimes\id_{\mathcal{B}_2^{(A)} \mathcal{B}_2^{(B)}}).
    \end{equation*}
    From this, it follows that the global operation after tracing out thermal baths is given by
    \begin{equation*}
        \Lambda^{(AB)}(\rho^{(AB)}) = \Tr_{\mathcal{B}^{(A)},\mathcal{B}^{(B)}}\left[U(\gamma^{\mathcal{B}^{(A)}} \otimes \gamma^{\mathcal{B}^{(B)}} \otimes \rho^{(AB)})U^\dagger \right]~.
    \end{equation*}

    To  show that $\Lambda^{(AB)} $ is indeed a composition of $\Lambda_1^{(AB)}$ and $\Lambda_2^{(AB)}$ it is sufficient to decompose $U$ and partial traces:
    \begin{equation*}
    \begin{aligned}
    & \Lambda^{(AB)} = \Tr_{\mathcal{B}^{(A)},\mathcal{B}^{(B)}}\left[U(\gamma^{\mathcal{B}^{(A)}} \otimes \gamma^{\mathcal{B}^{(B)}} \otimes \rho^{(AB)})U^\dagger \right]= \\
    & = \Tr_{\mathcal{B}_2^{(A)} \mathcal{B}_2^{(B)}}\Big[\Tr_{\mathcal{B}_1^{(A)},\mathcal{B}_1^{(B)}}\Big[ (U_2  \otimes \id_{\mathcal{B}_1^{(A)} \mathcal{B}_1^{(B)}} )(U_1 \otimes \id_{\mathcal{B}_2^{(A)} \mathcal{B}_2^{(B)}}) \\
    &\hspace{6 cm} (\gamma^{\mathcal{B}^{(A)}} \otimes \gamma^{\mathcal{B}^{(B)}} \otimes \rho^{(AB)}) (U_1 \otimes\id_{\mathcal{B}_2^{(A)} \mathcal{B}_2^{(B)}})^\dagger(U_2 \otimes \id_{\mathcal{B}_1^{(A)} \mathcal{B}_1^{(B)}} )^{\dagger} \Big]\Big] \\
    & = \Tr_{\mathcal{B}_2^{(A)} \mathcal{B}_2^{(B)}}\Big[U_2 \Tr_{\mathcal{B}_1^{(A)},\mathcal{B}_1^{(B)}}\Big[ (U_1 \otimes \id_{\mathcal{B}_2^{(A)} \mathcal{B}_2^{(B)}})\\
    &\hspace{7 cm} (\gamma^{\mathcal{B}^{(A)}} \otimes \gamma^{\mathcal{B}^{(B)}} \otimes \rho^{(AB)}) (U_1 \otimes \id_{\mathcal{B}_2^{(A)} \mathcal{B}_2^{(B)}})^\dagger \Big]U_2^\dagger\Big] = \\
    & =  \Tr_{\mathcal{B}_2^{(A)} \mathcal{B}_2^{(B)}}\Big[U_2(\gamma_2^{\mathcal{B}^{(A)}} \otimes \gamma_2^{\mathcal{B}^{(B)}} \otimes \Lambda_1^{(AB)}(\rho^{AB}))U_2^\dagger\Big] = \\
    & = \Lambda_2^{(AB)}(\Lambda_1^{(AB)}(\rho^{(AB)})).
    \end{aligned}
    \end{equation*}

    Next, to demonstrate that such an operation still conserves global energy we consider first of the commutation relations \eqref{SLTO_constr}:
    \begin{equation*}
    \begin{aligned}
    &  [U,H^{(A)} + H^{(B)} + H^{\mathcal{B}^{(A)}} + H^{\mathcal{B}^{(B)}}] = \\
    & (U_2 \otimes  \id_{\mathcal{B}_1^{(A)} \mathcal{B}_1^{(B)}} ) [(U_1 \otimes \id_{\mathcal{B}_2^{(A)} \mathcal{B}_2^{(B)}} ), H^{(A)} + H^{(B)} + H_1^{\mathcal{B}^{(A)}} +  H_2^{\mathcal{B}^{(A)}} + H_1^{\mathcal{B}^{(B)}} + H_2^{\mathcal{B}^{(B)}}] + \\
    & +  [(U_2 \otimes \id_{\mathcal{B}_1^{(A)} \mathcal{B}_1^{(B)}}  ), H^{(A)} + H^{(B)} + H_1^{\mathcal{B}^{(A)}} +  H_2^{\mathcal{B}^{(A)}} + H_1^{\mathcal{B}^{(B)}} + H_2^{\mathcal{B}^{(B)}}](U_1 \otimes  \id_{\mathcal{B}_2^{(A)} \mathcal{B}_2^{(B)}} ) = \\
    &  = 0.
    \end{aligned}
    \end{equation*}
    Since both unitaries commuted with appropriate Hamiltonians, 
    the commutation is preserved under extension by identity.
    The second commutation relation related to the conservation of energies weighed by inverse temperatures is satisfied by the same token.
    \end{proof}

    \SLTOconvex*

    \begin{proof}
    Let us define baths, Hamiltonians and unitaries corresponding to $\Lambda_1^{(AB)}$ and $\Lambda_2^{(AB)}$ as in the first paragraph of the proof of Theorem \ref{SLTO_closed}. To construct the semilocal thermal operation corresponding to a convex combination of $\Lambda_1^{(AB)}$ and $\Lambda_2^{(AB)}$ we define new thermal baths $\mathcal{B}^{(A)} = \mathcal{B}_1^{(A)} \oplus \mathcal{B}_2^{(A)}$ and $\mathcal{B}^{(B)} = \mathcal{B}_1^{(B)} \otimes \mathcal{B}_2^{(B)}$, with Hamiltonians:
    \begin{equation*}
    \begin{aligned}
    & H^{\mathcal{B}^{(A)}} = (H_1^{\mathcal{B}^{(A)}} + \Delta_1^{\mathcal{B}^{(A)}} \id_{\mathcal{B}_1^{(A)}})\oplus(H_2^{\mathcal{B}^{(A)}} + \Delta_2^{\mathcal{B}^{(A)}} \id_{\mathcal{B}_2^{(A)}}) \\
    & H^{\mathcal{B}^{(B)}} = H_1^{\mathcal{B}^{(B)}} \otimes \id_{\mathcal{B}_2^{(B)}} + H_2^{\mathcal{B}^{(B)}}\otimes \id_{\mathcal{B}_1^{(B)}} \\
    \end{aligned}
    \end{equation*}
    where the energy shifts $\Delta_{x}^{\mathcal{B}^{(A)}}$ are defined as
    \begin{equation}
    \Delta_1^{\mathcal{B}^{(A)}} = \frac{-1}{\beta^{(A)}}\log\left( \frac{\alpha}{Z_{{\mathcal{B}_1^{(A)}}}}\right)~,~~~\Delta_2^{\mathcal{B}^{(A)}} = \frac{-1}{\beta^{(A)}}\log\left( \frac{1-\alpha}{Z_{{\mathcal{B}_2^{(A)}}}}\right)
    \end{equation}
    with $Z_{{\mathcal{B}_x^{(A)}}}$ being a partition function for a given bath.

    This choice of Hamiltonians allows us to express
    thermal state of Alice's bath as a direct sum of two thermal states corresponding to two channels, with weights $\alpha$ and $1-\alpha$ e.g.:
    \begin{equation*}
    \begin{aligned}
        \frac{1}{Z_{\mathcal{B}^{(A)}}} e^{- \beta^{(A)}H^{B^{(A)}}} = \frac{1}{Z_{\mathcal{B}^{(A)}}} \left(e^{- \beta^{(A)} H_1^{\mathcal{B}^{(A)}}} \frac{\alpha}{Z_{\mathcal{B}_1^{(A)}}} \right) \oplus \left(e^{- \beta^{(A)}H_2^{\mathcal{B}^{(A)}}} \frac{1-\alpha}{Z_{\mathcal{B}_2^{(A)}}} \right) = (\alpha \gamma_1^{\mathcal{B}^{A}})\oplus((1- \alpha) \gamma_2^{\mathcal{B}^{A}})
    \end{aligned}
    \end{equation*}
    where $Z_{\mathcal{B}^{(A)}} = 1$ is a partition function of combined baths. On the other hand, Bob's bath's thermal state is of a product form $\gamma^{\mathcal{B}^{(B)}} = \gamma_1^{\mathcal{B}^{(B)}} \otimes \gamma_2^{\mathcal{B}^{(B)}}$.

    Finally we define the joint unitary as $U = (U_1 \otimes  \id_{\mathcal{B}_2^{(B)}})\oplus (U_2 \otimes \id_{\mathcal{B}_1^{(B)}})$ where the direct sum should be understood as between spaces $\mathcal{B}_1^{(A)}\otimes \mathcal{B}^{(B)}\otimes S^{(A)}\otimes S^{(B)}$ and $\mathcal{B}_2^{(A)}\otimes \mathcal{B}^{(B)}\otimes S^{(A)}\otimes S^{(B)}$. We define channel $\Lambda^{(AB)}$ is a standard way:
    \begin{equation*}
    \begin{aligned}
    & \Lambda^{(AB)}(\rho^{(AB)}) := \Tr_{\mathcal{B}^{(A)},\mathcal{B}^{(B)}}\left[U\left(\left((\alpha\,\gamma_1^{\mathcal{B}^{(A)}})\oplus( (1-\alpha) \gamma_2^{\mathcal{B}^{(A)}})\right)\otimes (\gamma_1^{\mathcal{B}^{(B)}} \otimes \gamma_2^{\mathcal{B}^{(B)}})\otimes \rho^{(AB)}\right)U^\dagger \right] = \\
    & = \alpha \Tr_{\mathcal{B}_1^{(A)}, \mathcal{B}^{(B)}} \left[(U_1 \otimes  \id_{\mathcal{B}_2^{(B)}}))\left(\gamma_1^{\mathcal{B}^{(A)}}\otimes (\gamma_1^{\mathcal{B}^{(B)}} \otimes \gamma_2^{\mathcal{B}^{(B)}})\otimes \rho^{(AB)}\right)(U_1 \otimes  \id_{\mathcal{B}_2^{(B)}}))^\dagger \right] + \\
    & \hspace{0.3 cm} + (1-\alpha) \Tr_{\mathcal{B}_2^{(A)}, \mathcal{B}^{(B)}} \left[(U_2 \otimes  \id_{\mathcal{B}_1^{(B)}}))\left(\gamma_2^{\mathcal{B}^{(A)}}\otimes (\gamma_1^{\mathcal{B}^{(B)}} \otimes \gamma_2^{\mathcal{B}^{(B)}})\otimes \rho^{(AB)}\right)(U_2 \otimes  \id_{\mathcal{B}_1^{(B)}}))^\dagger \right] =\\
    & = \alpha \Tr_{\mathcal{B}_1^{(A)}, \mathcal{B}_1^{(B)}} \left[U_1\left(\gamma_1^{\mathcal{B}^{(A)}}\otimes \gamma_1^{\mathcal{B}^{(B)}} \otimes \rho^{(AB)}\right)U_1^\dagger \right] +  \\
    & \hspace{0.3 cm} + (1-\alpha) \Tr_{\mathcal{B}_2^{(A)}, \mathcal{B}_2^{(B)}} \left[U_2\left(\gamma_2^{\mathcal{B}^{(A)}}\otimes \gamma_2^{\mathcal{B}^{(B)}} \otimes \rho^{(AB)}\right)U_2^\dagger \right] \\
    & = \alpha \Lambda_1^{(AB)}(\rho^{(AB)}) +( 1-\alpha) \Lambda_2^{(AB)}(\rho^{(AB)})
    \end{aligned}
    \end{equation*}

    For completeness of the proof, it is necessary to check if the newly defined unitary $U$ satisfies the commutation relations \eqref{SLTO_constr}. Notice that the unitary $U = (U_1 \otimes  \id_{\mathcal{B}_2^{(B)}})\oplus (U_2 \otimes \id_{\mathcal{B}_1^{(B)}})$ is block diagonal with respect to blocks $\mathcal{B}_1^{(A)}\otimes \mathcal{B}^{(B)}\otimes S^{(A)}\otimes S^{(B)}$ and $\mathcal{B}_2^{(A)}\otimes \mathcal{B}^{(B)}\otimes S^{(A)}\otimes S^{(B)}$, and so are all Hamiltonians. Therefore, if the commutation relations are satisfied on both of those subspaces separately, they are satisfied for the entire system. In the subspace $\mathcal{B}_1^{(A)}\otimes \mathcal{B}^{(B)}\otimes S^{(A)}\otimes S^{(B)}$ the first commutation relation simplifies to
    \begin{equation*}
    \begin{aligned}
    & [(U_1 \otimes  \id_{\mathcal{B}_2^{(B),2}}) \, , H^{(A)} + H^{(B)} +(H_1^{\mathcal{B}^{(A)}}+ \Delta_1^{\mathcal{B}^{(A)}}\id_{\mathcal{B}_1^{(A)}} ) + H^{\mathcal{B}^{(B)}} ] = \\
    & =[U_1 \, , H^{(A)} + H^{(B)} +(H_1^{\mathcal{B}^{(A)}}+ \Delta_1^{\mathcal{B}^{(A)}}\id_{\mathcal{B}_1^{(A)}} ) + H_1^{\mathcal{B}^{(B)}} ] = \\
    & =[U_1 \, , H^{(A)} + H^{(B)} +H_1^{\mathcal{B}^{(A)}} + H^{\mathcal{B}_1^{(B)}} ] + [U_1, \Delta_1^{\mathcal{B}^{(A)}}\id_{\mathcal{B}_1^{(A)}} )] = 0 + 0 = 0
    \end{aligned}
    \end{equation*}
    where the first commutator is equal to zero because $U_1$ satisfies the condition \eqref{SLTO_constr1}, and the second commutator is zero trivially. 
    The proof is completed by noting that the considerations for the second subspace in the direct sum follow the same line of reasoning, and likewise for the second commutation relation \eqref{SLTO_constr2}.
    \end{proof}
    
    Finally, we are equipped to prove the inclusion of LTOCC in SLTO, beginning with the case of single-round LTOCC.

    \begin{lem} \label{lem:one_round_LTOCC_in_SLTO}
    One-round Local Thermal Operations and Classical communication (LTOCC$_1$) (with projective energy-incoherent measurements) form a subset of semilocal thermal operations (SLTO), even when considering non energy-incoherent states.
    \end{lem}

    \begin{proof}
    Let us start the proof by decomposing one round LTOCC into conditioned thermal operations and thermal postprocessing. The latter is clearly a semi-local thermal operation since it involves thermal interaction between only one subsystem and its bath. Thus, without loss of generality in what follows, we assume the postprocessing to be trivial. 
    
    In order to demonstrate that conditioned thermal operations can be realised as SLTO, we rewrite them in an equivalent way that does not explicitly use classical memory. Consider a set of Bob's thermal operations $\{\Lambda_k^{(B)} \}_{k = 1}^n$, which are to be conditioned on Alice's measurement outcomes. Each of the operations will be performed using some thermal bath $\mathcal{B}_k^{(B)}$ with Hamiltonian $H_k^{\mathcal{B}^{(B)}}$  and unitary $U_k$. 
    
    To show that conditioned operation $\sum_{k = 1}^n\Lambda_k^{(B)}(P_k^{(A)} \rho^{(AB)}P_k^{(A)})$, where $P_k^{(A)}$ are projectors constituting Alice's measurement, can be performed as SLTO, let us define a joint Bob's bath $\mathcal{B}^{(B)} = \bigotimes_{k} \mathcal{B}_k^{(B)}$ with total Hamiltonian $H^{\mathcal{B}^{(B)}} = \sum_k \mathcal{H}_k^{B^{(B)}}$.
    We start the proof by focusing on an $n$-dimensional degenerate subspace of Alice's bath, thus setting $\mathcal{B}^{(A)} = \fC^{n}$ with Hamiltonian $H^{\mathcal{B}^{(A)}} \propto \id_n$, and later discuss how to generalize it to the full picture.
    Note that the Alice bath is completely degenerate, hence, it cannot be used to move any population between energy levels. Thus as such, it cannot be used to perform thermal operations on Alice's subsystem. However, it may serve the role of an ancilla for energy-incoherent measurements, which we are going to leverage.
    The global unitary in such a case is given by
    \begin{equation}
    \label{scary_unitary}
    U = \sum_{k,l} P_k^{(A)} \otimes U_k \otimes \id_{\mathcal{B}^{(B)} \setminus \mathcal{B}_k^{(B)} }  \otimes  |(k+l) \text{ mod } n \ra \la l|_{\mathcal{B}^{(A)}}
    \end{equation}
    where  the identity $\id_{\mathcal{B}^{(B)} \setminus \mathcal{B}_k^{(B)} }$ is set on all baths except the k$^{th}$ one.
    Semilocal thermal operation defined in such a way is then written as:
    \begin{equation}
    \begin{aligned}
    \label{eq:another_monster}
    &\Tr_{\mathcal{B}^{(A)}\mathcal{B}^{(B)}}[U(\gamma^{\mathcal{B}^{(B)}} \otimes \rho^{(AB)}  \otimes \gamma^{\mathcal{B}^{(A)}} )U^\dagger ] = \\
    & =\Tr_{\mathcal{B}^{(A)}\mathcal{B}^{(B)}}\Big[\left(\sum_{k,l} P_k^{(A)} \otimes U_k \otimes \id_{\mathcal{B}^{(B)} \setminus \mathcal{B}_k^{(B)} } \otimes |(k+l) \text{ mod } n \ra \la l|_{\mathcal{B}^{(A)}} \right)
    \left(\gamma^{\mathcal{B}^{(B)}} \otimes \rho^{(AB)} \otimes \frac{1}{n} \id_{B^{(A)}} \right) \\
    & \hspace{7 cm}
    \left(\sum_{k' l'} P_{k'}^{(A)} \otimes U_{k'} \otimes \id_{\mathcal{B}^{(B)} \setminus \mathcal{B}_{k'}^{(B)} } \otimes |(k'+l') \text{ mod } n \ra \la l'|_{\mathcal{B}^{(A)}} \right)^\dagger \Big] = \\
    & = \sum_{k,l} \frac{1}{n}  \Tr_{\mathcal{B}^{(B)}}\left[\left(P_k^{(A)} \otimes U_k \otimes \id_{\mathcal{B}^{(B)} \setminus \mathcal{B}_k^{(B)} }\right)(\gamma^{\mathcal{B}^{(B)}} \otimes \rho^{(AB)})\left(P_k^{(A)} \otimes U_k \otimes \id_{\mathcal{B}^{(B)} \setminus \mathcal{B}_k^{(B)} }\right)^\dagger \right] = \\
    & = \sum_k  \Tr_{\mathcal{B}_k^{(B)}}\left[U_k(\gamma_k^{\mathcal{B}^{(B)}} \otimes (P_k^{(A)} \rho^{(AB)}P_k^{(A)} ) )U_k^\dagger \right] = \\
    & = \sum_{k} (\Lambda_k^{(B)} \otimes \id_{S^{(A)}})(P_k^{(A)} \rho^{(AB)}P_k^{(A)})
    \end{aligned}
    \end{equation}
    Hence, a conditioned thermal operation can be performed with unitary $U$. Note that the reasoning presented above is not restricted to incoherent input states.

    The last step is to show that $U$ is indeed unitary and that it satisfies properties \eqref{SLTO_constr}. The unitarity is straightforward,
    \begin{equation*}
    \begin{aligned}
    &U\; U^\dagger = \left(\sum_{k,l} P_k^{(A)} \otimes U_k \otimes \id_{\mathcal{B}^{(B)} \setminus \mathcal{B}_k^{(B)} }  \otimes  |(k+l) \text{ mod } n \ra \la l|_{\mathcal{B}^{(A)}}\right) \\
    & \hspace{6 cm}\left(\sum_{k',l'} P_{k'}^{(A)} \otimes U_{k'} \otimes \id_{\mathcal{B}^{(B)} \setminus \mathcal{B}_{k'}^{(B)} }  \otimes  |(k'+l') \text{ mod } n \ra \la l'|_{\mathcal{B}^{(A)}}\right)^{\dagger} = \\
    &= \sum_{k,k',l}  P_k^{(A)}  P_{k'}^{(A)} \otimes \left(U_k \otimes \id_{\mathcal{B}^{(B)} \setminus \mathcal{B}_k^{(B)} }  \otimes \right) \left(U_{k'}^\dagger \otimes \id_{\mathcal{B}^{(B)} \setminus \mathcal{B}_{k'}^{(B)} }   \right) \otimes |(k+l) \text{ mod } n \ra \la (k'+l) \text{ mod } n|_{\mathcal{B}^{(A)}} = \\
    & = \sum_{k,l} P_k^{(A)} \otimes \left(U_k\; U_k^\dagger\right) \otimes \id_{\mathcal{B}^{(B)} \setminus \mathcal{B}_k^{(B)} } \otimes |(k+l) \text{ mod } n \ra \la (k+l) \text{ mod } n|_{\mathcal{B}^{(A)}} =  \\
    & \sum_k P_k^{(A)} \otimes  \id_{S^{(B)}\mathcal{B}^{(B)}} \otimes \id_{\mathcal{B}^{(A)}} = \id~,
    \end{aligned}
    \end{equation*}
    where we used the orthogonality of projective measurements.
    Finally the commutation relations \eqref{SLTO_constr} follows as well
    \begin{equation}
    \label{eq:PVM_comm}
    \begin{aligned}
    & [U,  H^{(A)} + H^{(B)} + H^{\mathcal{B}^{(A)}} + H^{\mathcal{B}^{(B)}}] = \\
    & \left[\left(\sum_{k,l} P_k^{(A)} \otimes U_k \otimes \id_{\mathcal{B}^{(B)} \setminus \mathcal{B}_k^{(B)} }  \otimes  |(k+l) \text{ mod } n \ra \la l|_{\mathcal{B}^{(A)}}\right), H^{(A)} + H^{(B)} + H^{\mathcal{B}^{(B)}} + H^{\mathcal{B}^{(A)}}\right] = \\
    & = \sum_{k,l} [P_k^{(A)} \otimes U_k \otimes\id_{\mathcal{B}^{(B)} \setminus \mathcal{B}_k^{(B)} }\;, H^{(A)} + H^{(B)} + H^{\mathcal{B}^{(B)}}] = \\
    & = \sum_k [P_k^{(A)} \otimes U_k\;, H^{(A)} + H^{(B)}+ H_k^{\mathcal{B}^{(B)}}] = \\
    & = 0
    \end{aligned}
    \end{equation}
    where we used the tact that $H^{\mathcal{B}^{(A)}}$ is proportional to identity, then omitted the subspaces on which both Unitary and Hamiltonias are identity, and in the final step, we used the fact that $U_k$ commutes with $H^{(B)}+ H_k^{\mathcal{B}^{(B)}}$ and $P_k^{(A)}$ are energy-incoherent.
    The second commutation relation is satisfied by the same token.

    In order to properly take into account the entirety of Alice's bath, we consider $\mathcal{B}^{(A)}$ to be a large bath, which satisfies the properties listed in \cite{horodecki2013fundamental, Bera2021}, and for a while assume that energy degeneracy $g_{m}$, of $m^{th}$ (important) energy level of the bath, is multiple of number of measurements on Alice system. Then we divide each energy degenerated subspace into $g_{m}/n$ orthogonal subspaces and define unitary matrix $U$ in each of those small subspaces (extended by both systems and Bob's baths) in the same way as \eqref{scary_unitary}. Finally, even if we have access only to a bath, for which the degeneracies are not multiples of $n$, we can divide the energy degenerated subspaces of Alice's bath in the same way and set $U = \id$ 
    on the remaining $g_m - n\lfloor g_m/n\rfloor = \varepsilon_m$ basis states
    in each energy-degenerated subspace. 
    However, it is easy to see that in the thermodynamic limit, the ratio $\varepsilon_m/g_m \rightarrow 0$, and thus
    distortions will be negligible, along the lines of reasoning presented in \cite{horodecki2013fundamental}.
    \end{proof}

    With all the above results, we may finally demonstrate the full inclusion.

    \subsetSLTO*
    
    \begin{proof}
    Combining Lemma~ \ref{lem:one_round_LTOCC_in_SLTO} with Theorem~\ref{SLTO_closed} and Theorem~\ref{SLTO_convex}, the result is obtained immediately.
    \end{proof}

    \noindent In other words, LTOCC$_n$ is defined as an $n$-fold composition of LTOCC$_1$, which is handled by closedness of SLTO, and LTOCC$_n$+R defines convex hull of LTOCC$_n$.

    In the proof of Theorem \ref{thm:subset_of_SLTO}, we rely on the orthogonality of projective measurements\bl{, in order to guarantee unitarity of joint evolution}.
    However, \bl{even} when considering energy-incoherent POVMs, \bl{which satisfy the minimal requirement of (average) Gibbs-preservation, this orthogonality ceases to hold. Thus above reasoning cannot be directly applied for beyond PVMs.} 
    
    \bl{Furthermore, as presented below, the action of an energy-incoherent POVM is not equivalent to a convex combination of energy-incoherent PVM. Thus, one cannot construct LTOCC protocol utilizing POVM by probabilistic mixture of LTOCC protocols based on PVMs.
    
    \begin{obs}
    Non-projective energy-incoherent POVM cannot be implemented as a convex combination of energy-incoherent PVMs. 
    \end{obs}

    \begin{proof}
    Assume that there exists an implementation of an energy-incoherent POVM $\{F_i = \sum_{j} q_{ij} \op{E_j}|\}$ as a convex combination of some, possibly many different, projective measurements. Then each element $F_i$ is decomposed into some combination of projectors $F_i = \sum_{k = 0}^{k_{max}(i)} r_{ik} P_{ik}$, where the coefficients are positive, $r_{ik}>0$, and the projectors $P_{ik}$ 
    are diagonal and may overlap. To simplify following calculation we define binary functions $d_{ik}(j) \in \{0,1\}$ which states if computational state $|E_j\ra$ belong to subspace of $P_{ik}$, thus $P_{ik} = \sum_j d_{ik}(j)|E_j\ra\la E_j|$. 

    In this setting one can obtain the probability of measuring basis state $|E_j\ra$ by POVM element $F_i$ in two different ways. The direct calculation gives $p = \Tr[F_i |E_j\ra\la E_j|] = q_{ij}$, while the composition of PVMs yields $p =\Tr[F_i |E_j\ra\la E_j|] = \sum_k r_{ik} \Tr[P_{ik} |E_j\ra\la E_j|] = \sum_k r_{ik} d_{ik}(j)$. Thus we obtain consistency requirement $q_{ij} = \sum_k r_{ik}d_{ik}(j)$.
    
    The application of POVM on arbitrary quantum state $\rho$ gives, up to normalization $p_i$,
    \begin{equation}
        p_i \rho^{(i)} = \sqrt{F_i}\rho\sqrt{F_i} = \sum_{j,j'} \sqrt{q_{ij}q_{ij'}}\;|E_j\ra\la E_j| \rho |E_{j'}\ra\la E_{j'}|  = \sum_{j,j'}\sqrt{\sum_{k,k'} r_{ik} d_{ik}(j) r_{ik'} d_{ik'}(j')}\; |E_j\ra\la E_j| \rho |E_{j'}\ra\la E_{j'}|.
    \end{equation}
    On the other hand, the combination of projective measurements creates post-measurement state
    \begin{equation}
        \sum_{k} r_{ik} P_{ik}\rho P_{ik} = \sum_{j,j'} \sum_k r_{ik} d_{ik}(j)d_{ik}(j')\; |E_j\ra\la E_j| \rho |E_{j'}\ra\la E_{j'}|.
    \end{equation}
    Comparing these two equations element-wise, one obtains
    \begin{equation*}
        \forall_{j,j'}\;\; \sum_k r_{ik} d_{ik}(j)d_{ik}(j') = \sqrt{\sum_{k,k'} r_{ik} d_{ik}(j) r_{ik'} d_{ik'}(j')},
    \end{equation*}
    which, after squaring both sides, and small rearagement of terms results in
    \begin{equation*}
        \forall_{j,j'}\;\;\sum_{kk'}\Big(1 - d_{ik}(j')d_{ik'}(j)\Big)d_{ik}(j) d_{ik'}(j') r_{ik} r_{ik'} = 0.
    \end{equation*}
    Each term in the above sum is nonnegative, so the sum may be equal zero only if each summand is zero. Furthermore, since $r_{ik}>0$ we may equivalently write
    \begin{equation}
    \label{one_more_eq}
        \forall_{j,j'}\forall_{k,k'}\;\;\Big(1 - d_{ik}(j')d_{ik'}(j)\Big)d_{ik}(j) d_{ik'}(j')  = 0.
    \end{equation}
    If the measurement $F_i$ is not projective, then some projectors $\{P_{ik}\}_{k = 0}^{k_{max}(i)}$ must differ from each other. Let $P_{ik_0}$ and $P_{ik_1}$ be two such different projectors, with $|E_{j_0}\ra$ being a basis state such that $P_{ik_0}|E_{j_0}\ra = |E_{j_0}\ra$ and $P_{ik_1}|E_{j_0}\ra = 0$. Substituting $k,k',j$ by $k_0, k_1, j_0$ the equation \eqref{one_more_eq} simplifies to 
    $\forall_{j'}\;\;d_{ik'}(j')  = 0,$
    which means that the entire projection $P_{ik_1}$ is equal $0$. Therefore energy-incoherent POVM cannot be considered as any convex combination, or probabilistic post-processing, of an energy-incoherent PVM.
    \end{proof}
    
    The above observation can be strengthened by separately considering the diagonal and off-diagonal elements of the output state, which leads to the conclusion that POVMs in general preserve more coherence than composition of PVMs with the same measurement probabilities. This excludes the implementation of POVM as a combination of PVMs followed by decohering thermal post-processing.  
    
    With these observations in mind, together with the discrepancy between thermal operations and Gibbs-preserving operations for coherent states, we believe that the extension of Theorem \ref{thm:subset_of_SLTO} by  POVMs, if possible, requires significantly different proof techniques.
    Therefore, the problem of} the inclusion of LTOCC in SLTO \bl{remains open}:

    \begin{conj}\label{conj:subset_of_SLTO_POVM}
    Local thermal Operations and Classical communication with randomness (LTOCC$_n$+R) extended by measurements defined by arbitrary energy-incoherent POVMs form a subset of semilocal thermal operations (SLTO).
    \end{conj}

    \section{Coherence evolution under SLTO and LTOCC($+$M)}\label{App:coch_evol}

    In this Appendix we consider the evolution of coherence under SLTO and LTOCC operators. First, we present auxiliary results restricting such evolutions and then apply them to bound the CHSH inequality for thermal operations presented in Theorem \ref{thm:LTO_bound_CHSH}.

    We start our discussion by extending the notion of modes of coherence, defined in \cite{Lostaglio2015QCoherence} for SLTO transformations.

    \begin{defi}
        Let us denote differences between Alice energies as $\{\omega_A\}$ and the differences between Bobs energies as $\{\omega_B\}$, then  we define a mode of coherence $\rho^{(\omega_A,\omega_B)}$ of bipartite quantum state $\rho^{(AB)}$ as
        \begin{equation}
            \rho^{(\omega_A, \omega_B)} := \sum_{\substack{E_i-E_k =\omega_A \\ \newline E_j -E_l = \omega_B}} |ij\ra\la ij|\rho^{(AB)}|kl\ra\la kl|~.
        \end{equation}
        Furthermore, we define joint mode of coherence as
        \begin{equation}
            \rho^{(\omega)} := \sum_{\omega_A + \omega_B = \omega} \rho^{(\omega_A, \omega_B)} = \sum_{(E_i+E_j) -(E_k+E_l) = \omega} |ij\ra\la ij|\rho^{(AB)}|kl\ra\la kl|~.
        \end{equation}
    \end{defi}
    Note that modes of coherence provide full decomposition of bipartite state $\rho^{(AB)} = \sum_{\omega_A,\omega_B} \rho^{(\omega_A,\omega_B)}$. Furthermore, as we presented above, they simplify the description of state's evolution under SLTO.
    
    \begin{lem}\label{lem:SLTO_coch_mod}
        If Alice and Bob operates at different temperatures $\beta^{(A)}\neq \beta^{(B)}$,
        modes of coherence evolve independently under SLTO transformations, ie, if $\Lambda^{(AB)}(\rho^{(AB)}) = \sigma^{(AB)}$, then
        \begin{equation}
            \Lambda^{(AB)}(\rho^{(\omega_A,\omega_B)}) = \sigma^{(\omega_A,\omega_B)} \;\; \forall_{\omega_A,\omega_B}
        \end{equation}
        Furthermore, even if temperatures are the same, joint modes of coherence still evolve independently: 
        \begin{equation}
            \Lambda^{(AB)}(\rho^{(\omega)}) = \sigma^{(\omega)} \;\; \forall_{\omega}
        \end{equation}
    \end{lem}

    \begin{proof}
    Let us start with the case $\beta^{(A)}\neq\beta^{(B)}$.
    From the defining condition of SLTO \eqref{SLTO_constr}, one can easily deduce time translational symmetry and weighted time translational symmetry of SLTO transformations
    \begin{equation}
    \label{SLTO_Ttranslation}
    \begin{aligned}
    &\Lambda^{(AB)}\left(e^{i t(H^{(A)} + H^{(B)})}\rho^{(AB)}e^{- i t(H^{(A)} + H^{(B)}) } \right) = e^{i t(H^{(A)} + H^{(B)})}\Lambda^{(AB)}\left(\rho^{(AB)} \right)e^{- i t(H^{(A)} + H^{(B)}) }~, \\      &\Lambda^{(AB)}\left(e^{i t(\beta^{(A)}H^{(A)} + \beta^{(B)}H^{(B)})}\rho^{(AB)}e^{- i t(\beta^{(A)}H^{(A)} + \beta^{(B)}H^{(B)}) } \right) = e^{i t(\beta^{(A)}H^{(A)} + \beta^{(B)}H^{(B)})}\Lambda^{(AB)}\left(\rho^{(AB)} \right)e^{- i t(\beta^{(A)}H^{(A)} + \beta^{(B)}H^{(B)}) }~. \\  
    \end{aligned}
    \end{equation}
    On the other hand, the time translation for modes of coherence is given by
    \begin{equation}
    \label{modes_translation}
    \begin{aligned}
    & e^{i t(H^{(A)} + H^{(B)})}\rho^{(\omega_A,\omega_B)}e^{- i t(H^{(A)} + H^{(B)})  }= e^{i t (\omega_A +\omega_B)}\rho^{(\omega_A,\omega_B)}, \\
    &e^{i t(\beta^{(A)}H^{(A)} + \beta^{(B)}H^{(B)})}\rho^{(\omega_A,\omega_B)}e^{- i t(\beta^{(A)}H^{(A)} + \beta^{(B)}H^{(B)}) } =  e^{i t (\beta^{(A)}\omega_A +\beta^{(B)}\omega_B)}\rho^{(\omega_A,\omega_B)}. 
    \end{aligned}
    \end{equation}
    Therefore in order to make translation symmetries \eqref{SLTO_Ttranslation} compatible with decomposition into modes of coherence \eqref{modes_translation} for any $t$, one must have
    \begin{equation*}
    \Lambda^{(AB)}(\rho^{(\omega_A,\omega_B)}) = \sigma^{(\omega_A,\omega_B)}\;\; \forall_{\omega_A,\omega_B}.
    \end{equation*}
    
    When the temperatures are the same $\beta^{(A)} = \beta^{(B)}$, both equations in \eqref{SLTO_Ttranslation} coincide, since there is no distinction between Alice and Bob's energy difference. However, since we are interested in the joint energy difference $\omega = \omega_A + \omega_B$, it suffices to repeat the above reasoning. 
    \end{proof}

    A similar property can be derived for the LTOCC($+$M) scenario as well.

    \begin{lem}\label{lem:LTOCC_coch_mod}
        Modes of coherence evolve independently under LTOCC transformations, ie, for any LTOCC$+M$ protocol denoted as $\Lambda^{(AB)}$ if $\Lambda^{(AB)}(\rho^{(AB)}) = \sigma^{(AB)}$, then
        \begin{equation}
            \Lambda^{(AB)}(\rho^{(\omega_A,\omega_B)}) = \sigma^{(\omega_A,\omega_B)} \;\; \forall_{\omega_A,\omega_B}.
        \end{equation}
    \end{lem}

    \begin{proof}
        Let us first consider the more elegant case of $\beta^{(A)} \neq \beta^{(B)}$ and LTOCC operations without memory. Then LTOCC $\subset$ SLTO, hence the claim follows from Lemma \ref{lem:SLTO_coch_mod}.

        In the general case of LTOCC$+$M or $\beta^{(A)} = \beta^{(B)}$, one must decompose operations into measurement and (conditioned) local thermal operations, and consider them separately. First, notice that the action of energy-incoherent projective measurement is equivalent to choosing a block of density matrix corresponding to measured outcome and setting all external matrix elements to zero; 
        thus they do not mix modes of coherence.
        Next, since each applied (conditional) thermal operation belongs to SLTO and is local, it cannot move coherences between different modes of coherence even if $\beta^{(A)} = \beta^{(B)}$. It is so, because in order to stay in joint mode of coherence, one must conserve $\omega = \omega_A +\omega_B$ and only one of summands could be affected by a local operation, thus neither $\omega_A$ nor $\omega_B$ can be changed.
    \end{proof}

    \CHSHLTOCCbound*    

    \begin{proof}
    To see how the above lemma affects the entanglement detection using LTOCC, notice that by using LTO we already fully leveraged entanglement within energy-degenerated subspaces. The only way to raise the bound is to use entanglement-bearing coherences between subspaces of different energies, which one must move into degenerated subspaces in order to apply non-diagonal measurements.

    Since entanglement-carrying coherences of interest are between different degeneracies, they belong to modes $\rho^{(\omega_A,\omega_B)}$ with $\omega_A \neq 0$ or $\omega_B \neq 0$, whereas the coherences between degeneracies belong to mode $\rho^{(0,0)}$. However, due to Lemma \ref{lem:LTOCC_coch_mod}, one cannot transform them in the desired way to raise the bound.
    \end{proof}

    \section{Approximating logic gates using LTOCC}\label{app:subsec_logic_gates}

    In this final Appendix, we provide examples of LTOCC maps without memory, which approximate standard logic operations on bits: CNOT and SWAP (between subsystems).
    Since both of those operations are nonlocal, they cannot be constructed from just local thermal operations. On the other hand, both of those operations have a straightforward construction using LOCC, which we present.

    To compare the original gate $M$ with the thermal counterpart $M^{(\beta)}$ we treat them as a (rescaled) probability distribution and calculate a total variation distance between them $\delta(M, M^{(\beta)}) = \frac{1}{2}\sum_{ijkl} |M_{ij,kl}-M_{ij,kl}^{(\beta)}|$~\cite{levin2008markov}. Let us denote Alice's and Bob's inverse temperature as $\beta^{(A)}$ and $\beta^{(B)}$. For convenience let us also define $g_A = e^{- \beta^{(A)} (E_1^{(A)} - E_0^{(A)})}$, $g_B = e^{- \beta^{(B)} (E_1^{(B)} - E_0^{(B)})}$ where $E_{0}^{(A)}, E_{1}^{(A)}$, $E_{0}^{(B)}, E_{1}^{(B)}$ are the energies of ground state and excited states of Alice's and Bob's systems respectively.

    We start with the construction of thermal CNOT by the following scheme. First Alice measures her bit and sets the information to Bob who depending on the outcome performs a (thermal) swap $S^{(\beta_B)}$ between its energy levels. The resulting thermal approximant $CNOT^{(\beta)}$ has the following form:
    \begin{equation*}
        CNOT^{(\beta)}_{A\rightarrow B} = 
        \begin{pmatrix}
        1 & 0 & 0 & 0\\
        0 & 1 & 0 & 0\\
        0 & 0 & 1 - g_B & 1\\
        0 & 0 & g_B & 0
    \end{pmatrix}
    \end{equation*}
    Thus in the limit of large temperature $g_B \to 1$ the thermal approximant $CNOT^{(\beta)}$ coincides with the standard controlled NOT.
    The statistical distance between them is given by 
    \begin{equation*}
        \delta(CNOT,CNOT^{(\beta)}) =  g_B 
    \end{equation*}

    The construction of thermal swap between subsystem $SWAP^{(\beta)}$, similarly to standard $SWAP$, requires three rounds of alternate thermal CNOT's.
    Following this, one obtains two alternative approximations
    \begin{equation*}
    \begin{aligned}
        SWAP_A^{(\beta)} & =  CNOT^{(\beta)}_{A\rightarrow B}  CNOT^{(\beta)}_{B\rightarrow A}  CNOT^{(\beta)}_{A\rightarrow B}    
          =
        \begin{pmatrix}
         1 & 0 & 0 & 0 \\
         0 & 1-g_A & g_B & 0 \\
         0 & g_A & \left(1-g_B\right){}^2 & 1-g_B \\
         0 & 0 & \left(1-g_B\right) g_B & g_B \\
        \end{pmatrix}  \\
        SWAP_B^{(\beta)} & =  CNOT^{(\beta)}_{B\rightarrow A}  CNOT^{(\beta)}_{A\rightarrow B}  CNOT^{(\beta)}_{B\rightarrow A}   
         =
        \begin{pmatrix}
         1 & 0 & 0 & 0 \\
         0 & \left(1-g_A\right){}^2 & g_B & 1-g_A \\
         0 & g_A & 1-g_B & 0 \\
         0 & \left(1-g_A\right) g_A & 0 & g_A \\
        \end{pmatrix}
    \end{aligned}
    \end{equation*}
    which, in the limit of large temperatures $g_A, g_B \to 1$, both coincide with standard $SWAP$
    \begin{equation*}
        SWAP = 
        \begin{pmatrix}
         1 & 0 & 0 & 0 \\
         0 & 0 & 1 & 0 \\
         0 & 1 & 0 & 0 \\
         0 & 0 & 0 & 1 \\
        \end{pmatrix} ,
    \end{equation*}
    and the statistical distance between them is given by 
    \begin{equation*}
    \begin{aligned}
    & \delta(SWAP,SWAP_A^{(\beta)}) = 3 - g_A - 2g_B ~, \\
    &\delta(SWAP,SWAP_B^{(\beta)}) = 3 - 2g_A - g_B.
    \end{aligned}
    \end{equation*}

    Thus, for any finite temperature, there is an inherent error associated with approximating even basic conditional logic operations when one attempts to implement them with operations preserving the thermal equilibria of the local systems.

    \end{widetext}

\bibliography{biblio}

\end{document}